\documentclass{article}

\usepackage[a4paper]{geometry}
\usepackage{amsmath}
\usepackage{amsfonts}
\usepackage{graphicx}
\graphicspath{{./Figures/}}
\usepackage[utf8]{inputenc}
\usepackage{amsthm}
\usepackage{amssymb}
\newtheorem{theorem}{Theorem}
\newtheorem{prop}[theorem]{Proposition}

\usepackage{xcolor}
\usepackage[font=small]{caption}
\usepackage{enumerate}
\usepackage{subcaption}
\usepackage[section]{placeins}
\DeclareMathOperator{\tr}{tr}
\usepackage[backend=biber,style=ieee]{biblatex}
\title{The Radial Hedgehog Solution in the Landau–de Gennes Theory: Effects of the Bulk Potentials}
\author{Sophie McLauchlan, Yucen Han, Matthias Langer, Apala Majumdar}
\date{June 27, 2023}

\begin{document}
\maketitle

\begin{abstract}
We study equilibrium configurations in spherical droplets of nematic liquid crystal with strong radial anchoring, within the Landau–de Gennes theory with a sixth-order bulk potential. The sixth-order potential predicts a bulk biaxial phase for sufficiently low temperatures, which the conventional fourth-order potential cannot predict. We prove the existence of a radial hedgehog solution, which is a uniaxial solution with a single isotropic point defect at the droplet centre, for all temperatures and droplet sizes, and prove that there is a unique radial hedgehog solution for moderately low temperatures, but not deep in the nematic phase. We numerically compute critical points of the Landau–de Gennes free energy with the sixth order bulk potential, with rotational and mirror symmetry, and find at least two competing stable critical points: the biaxial torus and split core solutions, which have biaxial regions around the centre, for low temperatures. The size of the biaxial regions increases with decreasing temperature. We also compare the properties of the radial hedgehog solution with the fourth-order and sixth-order potentials respectively, in terms of the Morse indices as a function of the temperature and droplet radius; the role of the radial hedgehog solution as a transition state in switching processes; and compare the bifurcation plots with temperature, with the fourth- and sixth-order potentials. Overall, the sixth-order potential has a stabilising effect on biaxial critical points and a de-stabilising effect on uniaxial critical points and we discover an altogether novel bulk biaxial critical point of the Landau–de Gennes energy with the sixth-order potential, for which the bulk biaxiality is driven by the sixth-order potential.
\end{abstract}

 \section{Introduction}
 \label{sec:introduction}
 Nematic liquid crystals (NLCs) are classical examples of mesogenic materials that combine fluidity with the orientational ordering of crystalline solids \parencite{deGennes}. NLCs have distinguished material directions that correspond to preferred directions of averaged molecular alignment, referred to as \emph{nematic directors}. Consequently, NLCs have direction-dependent physical, optical and rheological properties i.e. they are anisotropic materials and this anisotropy drives NLC applications in science and technology \parencite{lagerwallscalia}. 

 Defects are a defining feature of confined NLC samples. Loosely speaking, a defect is a point, line or surface wherein the NLC directors are not uniquely or properly defined \cite{deGennes}. Defects have pronounced optical signatures and NLCs are often recognised by the celebrated \emph{Schlieren textures} \cite{deGennes}. Defects can be undesirable in applications, since they can result in poor optical resolution but equally, defects can play crucial roles in self-assembly mechanisms acting as attractors or repellents for assembling mechanisms \parencite{Lavrentovich1994, lagerwallscalia}. For example, there is experimental evidence that stable arrays of point defects in liquid crystals can be exploited for novel applications in photonics and sensors \parencite{cholesterics}.

 There are several open questions about a rigorous mathematical description of NLC defects \parencite{ballreview}. 
 There are at least three different competing continuum theories for NLCs in the literature - the Oseen–Frank theory, the Ericksen theory and the Landau–de Gennes theory, ordered in terms of increasing generality \cite{deGennes, MottramNewton2014}. The Oseen–Frank and Ericksen theories are restricted to uniaxial NLCs, or NLC phases with a single nematic director such that all directions perpendicular to the uniaxial director are equivalent; these descriptions are limited in their abilities to describe higher-dimensional defects. The Landau-de Gennes theory is the most general continuum theory, in the sense that it can account for both uniaxiality and biaxiality, for which the NLC phase can have a primary and secondary director, along with defects of all dimensionality \cite{MajumdarZarnescu2010}. In this paper, we work within the celebrated Landau–de Gennes (LdG) theory for NLCs, for which the NLC state is described by a LdG $\mathbf{Q}$-tensor order parameter, whose eigenvectors model the nematic director(s) and the eigenvalues are a measure of the degree of the orientational order about the corresponding eigenvector \cite{deGennes}. The LdG theory is a variational theory, so that physically observable configurations are modelled by local or global minimisers of an appropriately defined LdG free energy, which is a nonlinear and non-convex functional that depends on both the order parameter and its derivatives. The LdG free energy density will typically comprise a bulk potential, which determines the NLC phase as a function of the temperature, and an elastic energy density which penalises spatial inhomogeneities and can account for geometric frustration and/or boundary effects. Mathematically, the energy minimisers are typically classical solutions of the associated Euler–Lagrange equations, which are a system of nonlinear and coupled partial differential equations, and the energy minimisers (or indeed any critical points of the LdG free energy) are analytic \cite{MajumdarZarnescu2010}. This makes a rigorous mathematical definition of a NLC defect in the LdG theory challenging, since they can naturally appear in the solution profiles without obvious blow-up characteristics. However, whilst the $\mathbf{Q}$-tensor solutions of the Euler–Lagrange equations are analytic, one can associate NLC defects with the discontinuities of the eigenvectors or with interfaces on which the number of distinct eigenvalues changes, or regions of normalised energy concentration \cite{Majumdar2012, MajumdarZarnescu2010}.

 In this paper, we focus on the canonical Radial Hedgehog (RH) defect \cite{RossoVirga1996}; this defect has been studied by several authors and we do not provide comprehensive references. The RH defect is essentially a degree $+1$ vortex in superconductivity and there are analogies with cavitation in elasticity as well \parencite{bbhbookGL}. Simply put, the RH defect is a spherically symmetric nematic point defect, such that the nematic director points radially outwards everywhere away from the RH defect. We study the RH defect on a three-dimensional spherical droplet with homeotropic boundary conditions i.e. the director is everywhere radial on the droplet surface or normal to the droplet surface. It is intuitively clear from symmetry considerations that one might expect a RH solution, with a single point defect at the droplet centre for which the director is not defined, and the director points radially outwards away from the centre to match the homeotropic boundary conditions. The RH solution has been studied mathematically in the Oseen–Frank and Landau–de Gennes frameworks, see for example \parencite{LinLiu2001} where the authors survey results on the existence and stability of the RH solution on spherical droplets in the Oseen–Frank theory, as a function of material properties. In recent years, there has been a splurge of mathematical activity on the study of NLC defects in the LdG framework. In a batch of papers \parencite{SchopohlSluckin1987, RossoVirga1996, GartlandMkaddem1999, MkaddemGartland2000, Majumdar2012, Lamy2013, IgnatNguyen2015}, and the list is certainly not complete, the authors prove the existence of a RH solution as a critical point of a LdG free energy on spherical droplets with homeotropic boundary conditions, study its stability and also study other competing critical points which might be energetically preferable depending on the droplet size, material properties and the temperature. In these papers, the authors consider a LdG energy with a fourth-order bulk potential, that is a quartic polynomial in the LdG $\mathbf{Q}$-tensor order parameter. The fourth-order potential can only admit isotropic (disordered) or uniaxial critical points i.e. the minimiser of the fourth-order potential is an ordered uniaxial phase for low temperatures and more details are given in the next section. In the LdG framework with the fourth-order potential, the RH solution is a uniaxial critical point of the LdG free energy, with the nematic director being the radial unit vector and the RH point defect being an isolated isotropic point at the droplet centre. It is known that the RH solution is globally stable for small droplets and for relatively high temperatures (that can be quantified) and unstable for large droplets and low temperatures. In particular, the authors numerically observe the competing \emph{biaxial torus} and \emph{split core} critical points of the LdG free energy for low temperatures, which replace the isotropic point defect of the RH solution by biaxial structures around the droplet centre. The isotropic point defect is energetically expensive for low temperatures, since the fourth-order bulk potential has an energy maximum at the isotropic phase for high temperatures, and hence, biaxiality arises from geometric frustration and energetic considerations for these examples.

 It is natural to ask if the properties of the RH solution strongly depend on the form of the LdG free energy density, in particular the choice of the bulk potential and the elastic energy density. In fact, there has been little work on the effect of the bulk potential on NLC defects and LdG solution landscapes. The fourth-order LdG bulk potential is the simplest polynomial that allows for a first-order isotropic-nematic phase transition, but higher-order polynomials are possible, that can allow for greater diversity in bulk NLC phases \parencite{AllenderLonga2008}. In this paper, we compare and contrast the fourth-order and a sixth-order LdG bulk potential. Following previous work in \cite{AllenderLonga2008}, we compute the minimisers of the sixth-order potential as a function of the temperature, and the minimisers are uniaxial with positive order parameter (so that the molecules, on average, align along the uniaxial director, whereas a uniaxial state with negative order parameter describes a state wherein the nematic molecules are approximately perpendicular to the uniaxial director) for moderately low temperatures, and there are no stable uniaxial critical points of the sixth-order potential for sufficiently low temperatures. In fact, the bulk energy minimiser is biaxial for sufficiently low temperatures. 
 
 We study critical points of a LdG free energy, with the sixth-order bulk potential, on a spherical droplet with homeotropic boundary conditions. The key difference, compared to previous work, is that biaxiality can now be a bulk effect as opposed to a localised phenomenon in the fourth-order case. We prove the existence of a RH solution in this case, with a uniaxial radial director and an isotropic point defect at the droplet centre. For moderately low temperatures for which the sixth-order bulk potential favours an ordered uniaxial phase, the qualitative properties of the RH solution are almost identical for the fourth- and sixth-order bulk potentials i.e. there is a unique RH solution with a monotonically increasing order parameter profile, away from the droplet centre. For low temperatures, when the sixth-order potential favours a bulk biaxial phase, we obtain multiple RH solutions, including RH solutions with negative order parameters. By contrast, the RH solution is unique with positive order parameter, for all low temperatures, with a fourth-order potential. We compare the stability of the RH solution as a function of the droplet size and temperature, with the fourth- and sixth-order potentials. As expected, the RH solution has a smaller domain of stability with the sixth-order potential, since the sixth-order potential promotes biaxiality for low temperatures and the RH solution is uniaxial everywhere away from the droplet centre. We also numerically compute the biaxial torus and split core solutions with symmetry constraints, with the sixth-order potential, and as expected the biaxial regions are larger and these solutions have enhanced stability with the sixth-order potential. In fact, we demonstrate that the RH solution can act as a transition state between the biaxial torus and split core solutions for low temperatures i.e. if one wants to design a switching process between the biaxial torus and split core solutions, the switching can be mediated by a RH solution. Heuristically, the biaxial torus shrinks to an isotropic point defect at the droplet centre for the RH solution, and then grows into the split core biaxial defect and vice-versa during the switching process. These numerical results are also complemented by bifurcation plots (under symmetry assumptions) as a function of the temperature; the bifurcation plots are qualitatively similar for the fourth- and sixth-order potentials with shifted bifurcation points that reflect the enhanced stability of the biaxial torus solution and reduced stability of the RH solution, with the sixth-order potential. However, the exciting questions pertain to the existence of altogether new LdG critical points or LdG energy minimisers, in the presence of a sixth-order bulk potential, which are biaxial in the bulk, and not merely near defects or in localised regions akin to the biaxial torus and the split core solutions. The answer is affirmative and we have not performed an exhaustive study on these lines, but have provided an example of a stable LdG critical point, that is biaxial in the bulk, for low temperatures and a large droplet (to be made precise), and the bulk biaxiality is driven by the sixth-order potential. The findings of this paper suggest that the choice of the LdG bulk potential can impact the multiplicity of (unstable) uniaxial solutions e.g. RH solutions; the domains of stability of uniaxial solutions; and importantly, give rise to new stable biaxial structures, which are outside the remit of the fourth-order potential. However, the experimental implications of these bulk biaxial solutions needs careful discussion.

 In Section~\ref{sec:prelim}, we set up our modelling framework and discuss the sixth-order potential and its stationary points, including minimisers in Section~\ref{sec:sixthfb}. In Section~\ref{sec:analysis}, we prove a batch of analytical results for the RH solution with the sixth-order potential, drawing out on the similarities and differences between the results with the fourth- and sixth-order potentials, respectively. In Section~\ref{sec:numerical}, we present illustrative numerical results including computations of the Morse indices of the RH solution and bifurcation plots as a function of the temperature. In Section~\ref{sec:conclusion}, we give a numerical example of a stable bulk biaxial LdG critical point, which is biaxial almost everywhere, away from the droplet centre and the droplet boundary, and the bulk biaxiality is driven by the sixth-order potential. We discuss the implications of this numerical observation i.e. is it an artefact of the mathematical model or can it be foundational for new experiments. We conclude with some perspectives and open questions.

\section{Preliminaries}
\label{sec:prelim}
We work with the Landau–de Gennes (LdG) theory wherein the NLC configuration is modelled by the LdG order parameter - the \textbf{Q}-tensor, which is a symmetric, traceless, \(3\times3\) matrix with five degrees of freedom \parencite{deGennes}. The \textbf{Q}-tensor can be written as \parencite{MajumdarZarnescu2010, MottramNewton2014}
	\begin{equation}
		\textbf{Q} = s\bigg(\boldsymbol{n}\otimes\boldsymbol{n} - \frac{1}{3}\textbf{I}\bigg) + p\bigg(\boldsymbol{m}\otimes\boldsymbol{m} - \frac{1}{3}\textbf{I}\bigg),
	\end{equation}
	where \(\boldsymbol{n}\) and \(\boldsymbol{m}\) are orthonormal eigenvectors that model the nematic directors, while \(s\) and \(p\) are scalar order parameters that measure the degree of order about \(\boldsymbol{n}\) and \(\boldsymbol{m}\), respectively. The liquid crystal configuration is biaxial if both \(s\) and \(p\) are nonzero and non-equal; uniaxial if only one of \(s\) and \(p\) is nonzero or if $s=p$; and isotropic if both \(s = p = 0\). Physically, a biaxial state has two preferred directions of orientational ordering or two directors, whereas a uniaxial state has a uniquely defined director that corresponds to the eigenvector with the largest positive eigenvalue. An isotropic state has no orientational ordering, so that all directions in space are physically equivalent and there is no notion of a director.
 
	Our domain is a spherical droplet, $B(0,R) = \left\{ \mathbf{x} \in \mathbb{R}^3; |\mathbf{x}| \leq R \right\}$, where $R$ is the droplet radius and we impose uniaxial homeotropic boundary conditions i.e. the uniaxial director is normal/orthogonal to the droplet surface, $|\mathbf{x}| = R$. The equilibrium configurations are critical points of the LdG free energy, which, in the absence of surface energies and external fields, is of the form \parencite{deGennes}
	\begin{equation}
		\mathcal{F}[\textbf{Q}] = \int_{B(0,R)}\frac{L}{2}|\nabla\textbf{Q}|^2 + f_B(\textbf{Q})\,dV, \label{LdGgeneral}
	\end{equation}
	where \(\frac{L}{2}|\nabla\textbf{Q}|^2\) is the one-constant elastic energy density with material-dependent elastic constant $L>0$, and $|\nabla\textbf{Q}|^2 = Q_{ij,k}Q_{ij,k}$, $Q_{ij,k} = \frac{\partial Q_{ij}}{\partial x_k}$, $i,j,k = 1,2,3$ penalises spatial inhomogeneities, noting that we use the Einstein summation convention here and throughout this paper. Further, \(f_B(\textbf{Q})\) is the bulk potential which determines the preferred bulk NLC phase (uniaxial/biaxial/isotropic) in spatially homogeneous systems as a function of temperature.
	
	We work with two different bulk potentials throughout this paper, a fourth-order potential which only admits uniaxial or isotropic critical points and a sixth-order potential which allows for uniaxial, biaxial and isotropic critical points \parencite{LewisThesis, AllenderLonga2008}. Our aim is to assess the impact of the bulk potential on the emergence of biaxiality for equilibrium configurations, whether it arises from mere geometric frustration or whether biaxiality can arise from bulk effects too. The fourth-order bulk potential is given by \parencite{MottramNewton2014}:
	\begin{equation} \label{eq:f1}
		f_B(\textbf{Q}) = \frac{A}{2}\tr\textbf{Q}^2 - \frac{B}{3}\tr\textbf{Q}^3 + \frac{C}{4}\big(\tr\textbf{Q}^2\big)^2,
	\end{equation}
	where $\tr \textbf{Q}^2 = Q_{ij} Q_{ij}$, $\tr \textbf{Q}^3 = Q_{ij} Q_{jk} Q_{ki}$, $i,j,k = 1,2,3$. The constant \(A\) is a material- and temperature-dependent constant and \(B\) and \(C\) are material-dependent constants. The physical meaning of these constants is not entirely established in the literature, but it is commonly accepted that  \(B > 0\) corresponds to rod-like molecules; and \(B < 0\) corresponds to discotic molecules \parencite{GramsbergenLonga1986}. In this manuscript, we assume $B>0$. This is the simplest form of bulk potential which captures a first-order phase transition between the nematic and isotropic phases \parencite{Lamy2013,MajumdarZarnescu2010} i.e. the critical points of \eqref{eq:f1} are either uniaxial or isotropic. 
    When $A\leq \frac{B^2}{27 C}$, the minimisers of $f_B$ are uniaxial; when $0\leq A\leq \frac{B^2}{24C}$, the minimisers are either uniaxial or isotropic; when $A\geq\frac{B^2}{24C}$, the unique minimiser is isotropic. The parameter regime $A <0$ describes the low temperature regime, for which the bulk potential strongly favours an ordered uniaxial nematic phase, so that biaxiality is only induced by geometric frustration or the competition between $f_B$ and the elastic energy density.
	
	We also consider a more general and more complicated  sixth-order bulk potential in this paper, which is of the form \parencite{AllenderLonga2008}
	\begin{equation}\label{eq:f2}
		f_B(\textbf{Q}) = \frac{A}{2}\tr\textbf{Q}^2 - \frac{B}{3}\tr\textbf{Q}^3 + \frac{C}{4}\big(\tr\textbf{Q}^2\big)^2 + \frac{D}{5}\tr\textbf{Q}^2\tr\textbf{Q}^3 + \frac{E}{6}\big(\tr\textbf{Q}^2\big)^3 + \frac{(F - E)}{6}\big(\tr\textbf{Q}^3\big)^2.
	\end{equation}
	The sixth-order bulk potential admits biaxial  critical points in addition to uniaxial and isotropic bulk critical points, so that biaxiality can be a bulk effect as opposed to the fourth-order potential in \eqref{eq:f1}. Again, \(A\) is a material- and temperature-dependent constant, while \(B,C,D,E,\) and \(F\) are material-dependent constants.  Moreover, we require that \(E \geq 0, F > 0\) to guarantee the stability of the expansion \parencite{AllenderLonga2008}. As with \eqref{eq:f1}, $A<0$ describes the low-temperature phase, and the bulk potential \eqref{eq:f2} admits biaxial minimisers for sufficiently low temperatures, as will be demonstrated in the next section.
	
	Throughout this paper, we work with a nondimensionalised version of the LdG free energy (\ref{LdGgeneral}), inspired by \parencite{HuQu2016}. Let
	\begin{equation}
		\tilde{\boldsymbol{x}} = \frac{\boldsymbol{x}}{R}, \quad \widetilde{\textbf{Q}} = \sqrt{\frac{27C^2}{2B^2}}\textbf{Q}. \label{nondimvariables}
	\end{equation}
	Then the dimensionless LdG free energy with the fourth-order potential is given by
	\begin{equation}
		\mathcal{F}_{four} = \widetilde{\mathcal{F}}[\widetilde{\textbf{Q}}] = \int_{B(0,1)}\Bigg(\frac{\varepsilon^2}{2}|\nabla\widetilde{\textbf{Q}}|^2 + \frac{t}{2}\tr\widetilde{\textbf{Q}}^2 - \sqrt{6}\tr\widetilde{\textbf{Q}}^3 + \frac{1}{2}\big(\tr\widetilde{\textbf{Q}}^2\big)^2\Bigg)\,d\widetilde{V},\label{4thNDLdG}
	\end{equation}
	and with the sixth-order potential is given by
	\begin{multline}
		\mathcal{F}_{six}=\widetilde{\mathcal{F}}[\widetilde{\textbf{Q}}] = \int_{B(0,1)}\Bigg(\frac{\varepsilon^2}{2}|\nabla\widetilde{\textbf{Q}}|^2 + \frac{t}{2}\tr\widetilde{\textbf{Q}}^2 - \sqrt{6}\tr\widetilde{\textbf{Q}}^3 + \frac{1}{2}\big(\tr\widetilde{\textbf{Q}}^2\big)^2 \\
		+ \frac{d}{5}\tr\widetilde{\textbf{Q}}^2\tr\widetilde{\textbf{Q}}^3 + \frac{e}{6}\big(\tr\widetilde{\textbf{Q}}^2\big)^3 + \frac{(f - e)}{6}\big(\tr\widetilde{\textbf{Q}}^3\big)^2\Bigg)\,dV,\label{6thNDLdG}
	\end{multline}
	where the characteristic length scale \(\xi = \sqrt{\frac{27CL}{B^2}}\),  \(\varepsilon = \frac{\xi}{R}\) and 
	\begin{equation}
		t = \frac{27AC}{B^2}, \quad d = \frac{2\sqrt{6}BD}{9C^2}, \quad e = \frac{4B^2E}{27C^3}, \quad f = \frac{4B^2F}{27C^3}. \nonumber
	\end{equation}
	This rescaling reduces the computational domain to the unit ball in three dimensions, $B(0,1)$, and the geometrical properties are captured by the parameter \(\varepsilon\). We refer to \(t\) as the temperature for convenience, although it is, more precisely, a function of the absolute temperature. We drop the tildes for brevity in the remainder of this manuscript and all results are interpreted in terms of the dimensionless variables.
	
	The homeotropic boundary condition is encoded by the Dirichlet condition \parencite{SonnetKillian1995}
	\begin{equation}
		\textbf{Q}_{s_+} = s_+\bigg(\hat{\boldsymbol{r}}\otimes\hat{\boldsymbol{r}} - \frac{1}{3}\textbf{I}\bigg), \quad \boldsymbol{r} \in \partial B(0,1), \label{DirichletBC}
	\end{equation}
	where \(\hat{\boldsymbol{r}}\) is the unit vector in the radial direction,  \(s_+\) is the minimiser of
 \[
 \left\{ f_B(\mathbf{Q}): \mathbf{Q} = s\left(\mathbf{n}\otimes \mathbf{n} - \mathbf{I}/3 \right); \mathbf{n} \in \mathbb{R}^3; |\mathbf{n}|=1; s \geq 0, \right\}
 \]
 and $f_B$ is given by either \eqref{eq:f1} or \eqref{eq:f2}.  
  There is an explicit expression for $s_+$ for the fourth-order potential in \eqref{eq:f1}:
	\begin{equation}
		s_+ = \sqrt{\frac{3}{2}}\frac{3 + \sqrt{9 - 8t}}{4},
	\end{equation}
 when $t<\frac{9}{8}$.
	For the sixth-order potential in \eqref{eq:f2}, \(s_+\) is the largest positive minimiser of the function
	\begin{equation}
		g(s) := \frac{t}{3}s^2 - \frac{2\sqrt{6}}{9}s^3 + \frac{2}{9}s^4 + \frac{4d}{135}s^5 + \frac{4e}{81}s^6 + \frac{2(f - e)}{243}s^6, \label{gdefn}
	\end{equation}
	and the function \(g\) is simply the potential \eqref{eq:f2} restricted to uniaxial $\mathbf{Q}$-tensors.
	
	The equilibrium configurations are (classical) solutions of the Euler–Lagrange (EL) equations associated with the LdG free energy. In the case of \eqref{eq:f1}, the EL equations are given by:
	\begin{equation}
		\varepsilon^2\Delta Q_{ij} = tQ_{ij} - 3\sqrt{6}\bigg(Q_{ik}Q_{kj} - \frac{1}{3}\delta_{ij}\tr\textbf{Q}^2\bigg) + 2Q_{ij}\tr\textbf{Q}^2, \label{4thEL}
	\end{equation}
	and with the sixth-order potential \eqref{eq:f2}, the EL equations are given by
	\begin{multline}
		\varepsilon^2\Delta Q_{ij} = tQ_{ij} - 3\sqrt{6}\bigg(Q_{ik}Q_{kj} - \frac{1}{3}\delta_{ij}\tr\textbf{Q}^2\bigg) + 2Q_{ij}\tr\textbf{Q}^2
		+ \frac{2d}{5}Q_{ij}\tr\textbf{Q}^3 \\
        + \frac{3d}{5}\tr\textbf{Q}^2\bigg(Q_{ik}Q_{kj} - \frac{1}{3}\delta_{ij}\tr\textbf{Q}^2\bigg) + eQ_{ij}\big(\tr\textbf{Q}^2\big)^2 + (f - e)\tr\textbf{Q}^3\bigg(Q_{ik}Q_{kj} - \frac{1}{3}\delta_{ij}\tr\textbf{Q}^3\bigg). \label{6thELQ}
	\end{multline}
	We note that
	\begin{equation}
		\sqrt{6}\,\delta_{ij}\tr\textbf{Q}^2, \quad \text{and} \quad
		\frac{1}{3}\delta_{ij}\tr\textbf{Q}^2\bigg(3\sqrt{6} - \frac{3d}{5}\tr\textbf{Q}^2 - (f - e)\tr\textbf{Q}^3\bigg) \nonumber
	\end{equation}
	are Lagrange multipliers  for the tracelessness constraint.
	
	The admissible space for the LdG $\mathbf{Q}$-tensors is taken to be  \parencite{Majumdar2012}
	\begin{equation}
		\mathcal{A}_{\textbf{Q}} := \{\textbf{Q} \in W^{1,2}(B(0,1),\bar{S}):\textbf{Q} = \textbf{Q}_{s_+} \,\,\text{on}\,\,\partial B(0,1)\}, \label{AQ}
	\end{equation}
	where \(W^{1,2}(B(0,1),\bar{S})\) is the Sobolev space
	\begin{equation}
		W^{1,2}(B(0,1),\bar{S}) = \Biggl\{\textbf{Q}:B(0,1) \to\bar{S}:\int_{B(0,1)}|\textbf{Q}|^2 + |\nabla\textbf{Q}|^2\,dV < \infty\Biggr\}, \nonumber
	\end{equation}
	and \(\bar{S}\) is the space of symmetric, traceless \(3\times3\) matrices $
		\bar{S} := \{\textbf{Q}\in \mathbb{M}^{3\times3}:Q_{ij} = Q_{ji}, Q_{ii} = 0\}.$

	We focus on a special exact solution of the EL equations in \eqref{4thEL} and \eqref{6thELQ} in the admissible space \(\mathcal{A}_\textbf{Q}\); the so-called radial hedgehog (RH) solution on spherical droplets with homeotropic anchoring conditions
	\begin{equation}
		\textbf{Q}^*(\boldsymbol{r}) = s^*(r)\bigg(\hat{\boldsymbol{r}}\otimes\hat{\boldsymbol{r}} - \frac{1}{3}\textbf{I}\bigg). \label{Qrh}
	\end{equation}
	This is a uniaxial solution and the uniaxial director is the radial unit vector, with a scalar order parameter \(s^*\) that only depends on the radial distance, \(r\), from the droplet centre. The corresponding admissible space for \(s^*\) is \parencite{Majumdar2012}
	\begin{equation} \label{eq:As}
		\mathcal{A}_s := \big\{s\in W^{1,2}([0,1],\mathbb{R}):s(1) = s_+\big\}.
	\end{equation}
 In what follows, we study the qualitative properties of the RH solution with the fourth- and sixth-order bulk potentials in \eqref{eq:f1} and \eqref{eq:f2}, to compare and contrast the effects of the bulk potential on $s^*$, and also on other competing critical points (uniaxial and biaxial) of the LdG free energy and the role of biaxiality on the corresponding solution landscapes.
 
 \section{The Sixth-Order Potential}
 \label{sec:sixthfb}

In what follows, we study the critical points of the sixth-order potential in \eqref{eq:f2}, firstly when restricted to uniaxial $\mathbf{Q}$-tensors and secondly, in the whole space $\bar{S}$ of symmetric, traceless $3\times 3$ matrices. There are two key differences when compared to the fourth-order potential in \eqref{eq:f1}: (i) in the restricted class of uniaxial $\mathbf{Q}$-tensors,  \eqref{eq:f2} has two non-trivial minimisers, with positive and negative order parameters, $s_+$ and $s_-$ respectively, and for sufficiently low temperatures, $s_-$ is the global minimiser of $g(s)$ defined in \eqref{gdefn}. In contrast, the uniaxial critical point with positive order parameter, $s_+$ is always the global minimiser of the fourth-order potential in \eqref{eq:f1} for low temperatures. (ii) Secondly, the sixth-order potential admits biaxial critical points, while the fourth-order potential can only admit uniaxial or isotropic critical points, and the global minimiser of \eqref{eq:f2} is actually biaxial, deep in the nematic phase. This is outside the scope of the fourth-order potential in \eqref{eq:f1}. We give more details in the sub-sections below. 
	
	\subsection{Uniaxial Critical Points of the Sixth-Order Potential}
In what follows, we work with parameters in \eqref{eq:f2} so that $f_B(\mathbf{Q})$, restricted to uniaxial $\mathbf{Q}$-tensors, has a single isotropic critical point for high temperatures and two well-defined critical points (with positive and negative uniaxial order parameter) for low temperatures. To this end, we consider the quartic  polynomial \(g'(s)/s\), where \(g\) is defined in \eqref{gdefn}.

 We determine the nature of the roots of the quartic polynomial \(g'(s)/s\) by considering its discriminant \parencite{Rees1922}
	\begin{multline} \nonumber
		\Delta = \frac{512}{59049}(5e + f)^3t^3 + \frac{1024}{177147}\bigg(-\frac{512}{27}(5e + f)^2 - d^4 + \frac{32}{3}(5e + f) + \frac{32\sqrt{6}}{9}d(5e + f)^2\bigg)t^2 \\
		+ \frac{1024}{6561}\bigg(-\frac{64}{243}d^2 + \frac{512}{243}(5e + f) + \frac{32}{9}(5e + f)^2 - \frac{2}{27}d^2(5e + f) - \sqrt{6}d^2 + \frac{320\sqrt{6}}{243}d(5e + f)\bigg)t \\
		+ \frac{1024}{2187}\bigg(-\frac{64}{81}(5e + f) + \frac{4\sqrt{6}}{81}d^3 + \frac{2}{3}d(5e + f) + \frac{8}{81}d^2 - (5e + f)^2\bigg),
	\end{multline}
	and the quantities 
	\begin{equation}
		P = 8(5e + f) - d^2, \quad R = d^2 - \frac{32}{27}f(5e + f) - \frac{16\sqrt{6}}{27}(5e + f)^2. \nonumber
	\end{equation}
	The signs of the discriminant, \(\Delta\), and the quantities \(P\) and \(R\) characterize the roots as follows. For \(P > 0\), and \(R \neq 0\), we conclude that a quartic polynomial has two real roots and two complex conjugate roots if \(\Delta < 0\); two pairs of complex conjugate roots if \(\Delta > 0\); and a real double root and two complex conjugate roots if \(\Delta = 0\).
 
 We restrict our parameters \(d, e\), and \(f\) to satisfy one of these three sets of conditions for all values of \(t\) and \(\varepsilon\). Figure \ref{fig:discriminant} demonstrates that for the specific choice, \(d = 1, e = 0, f = 1\), \(P > 0\) and \(R \neq 0\) and there is some transition temperature \(t_0\) such that \(\Delta < 0\) when \(t < t_0\); \(\Delta = 0\) when \(t = t_0\); and \(\Delta > 0\) when \(t > t_0\). These results can be translated into properties of the function \(g\) in (\ref{gdefn}) in this parameter regime. Namely, under these conditions, the function \(g\) will have one real stationary point, \(s = 0\), above some transition temperature \(t_0\), two stationary points at \(t_0\), and three real stationary points at temperatures below \(t_0\). In particular, we note that \(g\) is a double-well potential at lower temperatures $t< t_0$ for which there is at least a positive local minimiser $s=s_+$ of the sixth-order polynomial $g(s)$. In the remainder of this manuscript, we work with $e =0$ and $d = f= 1$.
	
	\begin{figure}
		\centering
		\includegraphics[width = 0.3\textwidth, angle = -90]{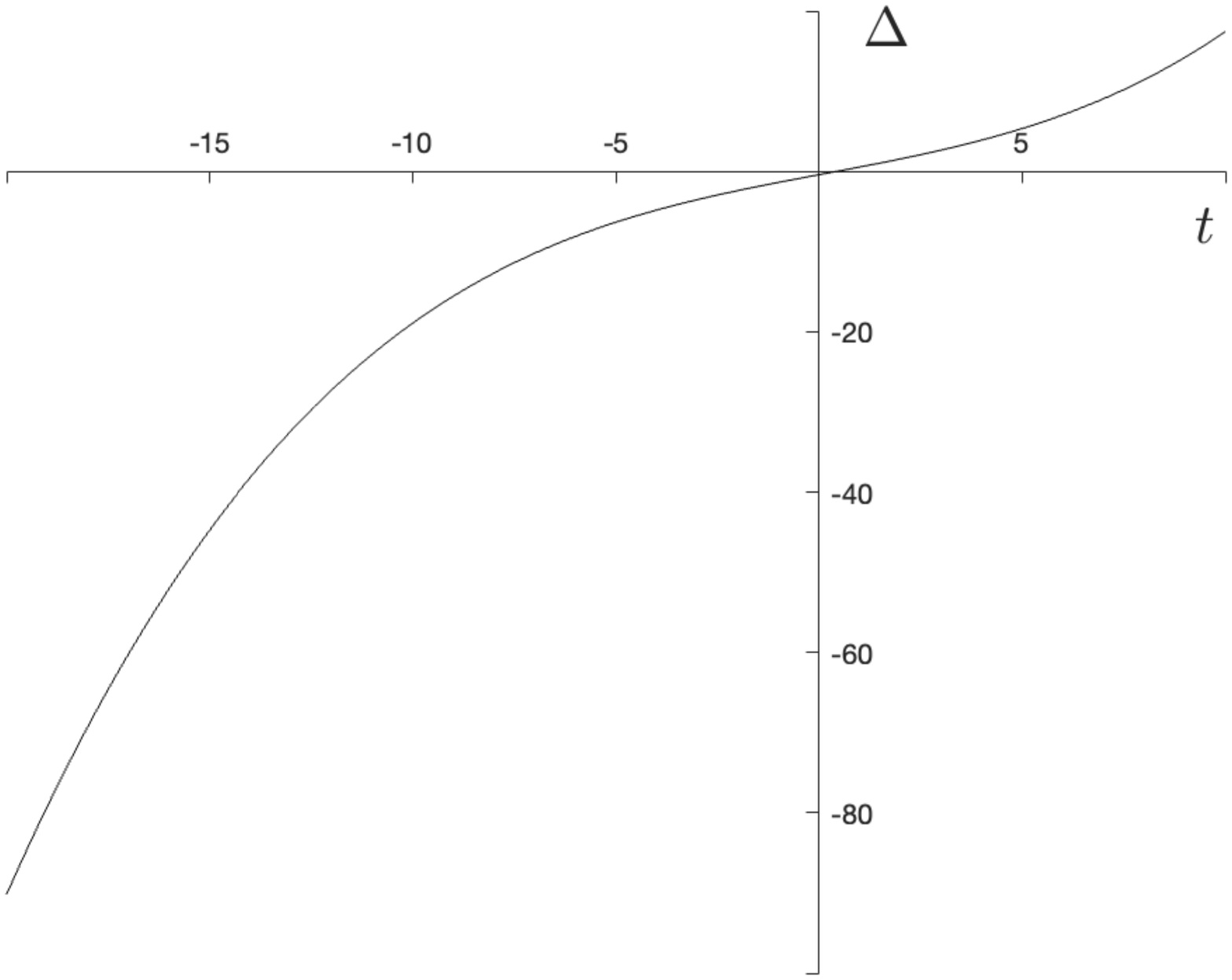}
		\caption{The function \(\Delta\) for \(d = 1, e = 0, f = 1\) and \(t\) from \(-20\) to 5.} \label{fig:discriminant}
	\end{figure}
 	
 In Figure \ref{fig:6thgplots}, we plot the function \(g\) for three different temperatures. One can clearly see the isotropic state is the global minimiser for the high temperature $t=5$; \(s_+ >0\) is the global minimiser for the low temperature $t=-25$; and as the temperature further decreases to $t=-100$, \(s = s_- <0\) is the global minimiser of \(g \). Nonzero critical points first appear at the transition temperature \(t_0 \approx 0.97\), and the minimisers \(s_+\) and \(s_-\) have the same energy at the transition temperature \(t^* \approx -48.6\). This is further illustrated in Figure \ref{fig:6thunisp}, where we plot the stationary points of \(g\) as a function of $t$, and indicate their stability. We observe that for $t<0$, there are two local minimisers, $s=s_+>0$ and $s=s_{-}<0$ of $g(s)$, and  \(s = s_+\) is the global minimiser for moderately low temperatures, and $s=s_{-}$ is the global minimiser of $g$ for $t < -50$. 
	\begin{figure}[!ht]
		\begin{minipage}{0.2\textwidth}
			\centering
			\includegraphics[width=\textwidth]{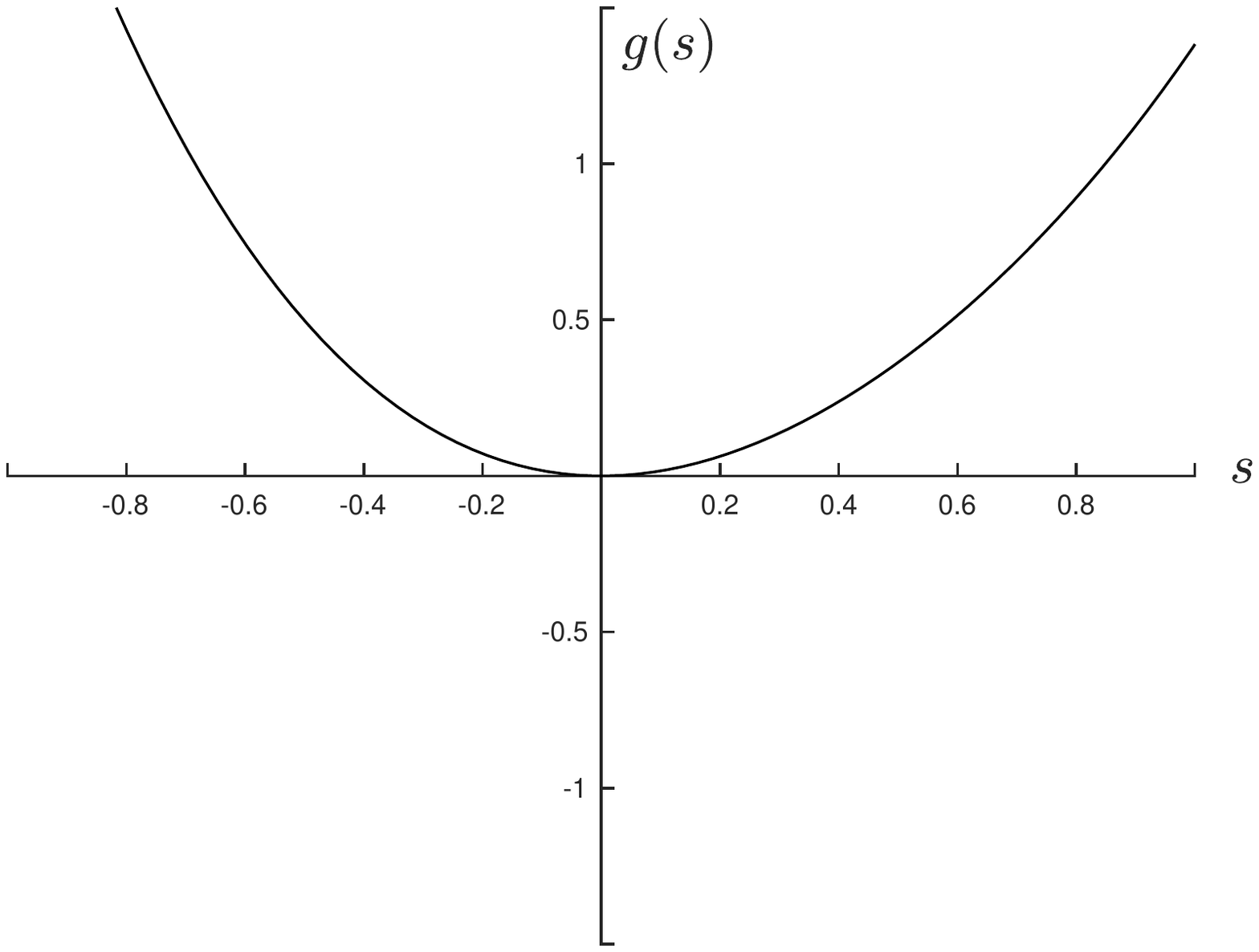}
			\subcaption{}
		\end{minipage}\hfill
		\begin{minipage}{0.2\textwidth}
			\centering
			\includegraphics[width=\textwidth]{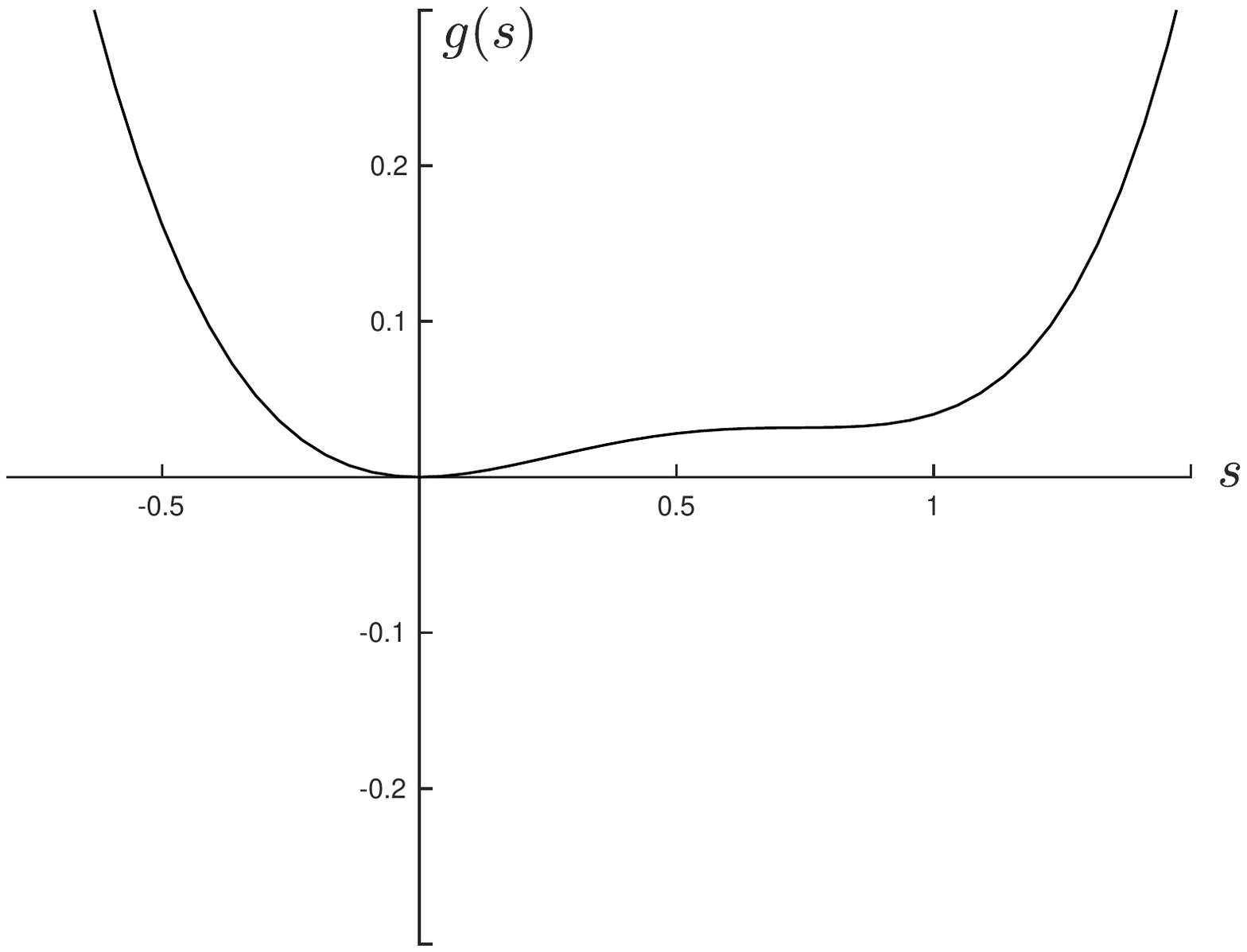}
			\subcaption{}
		\end{minipage}\hfill
		\begin{minipage}{0.2\textwidth}
			\centering
			\includegraphics[width=\textwidth]{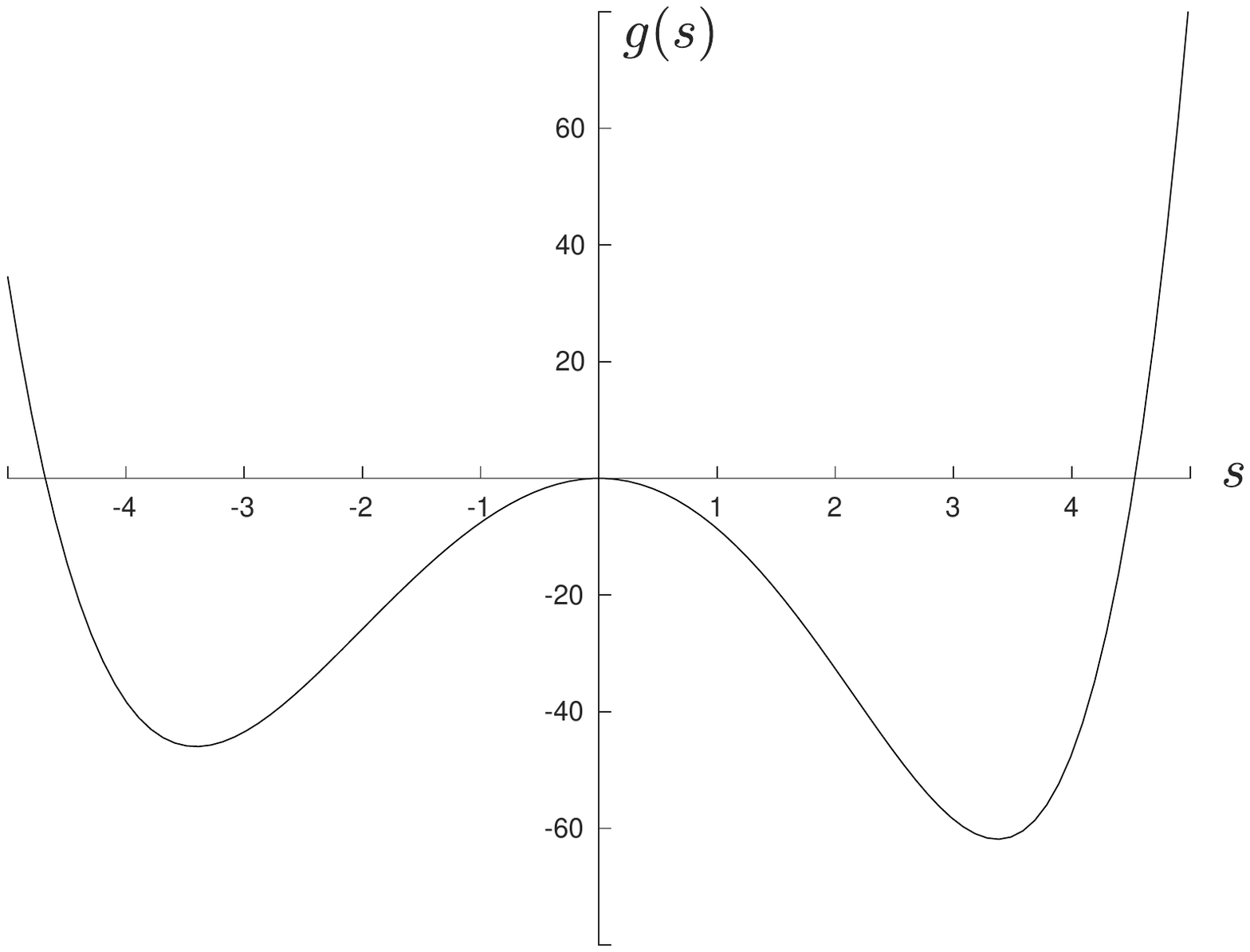}
			\subcaption{}
		\end{minipage}\hfill
            \begin{minipage}{0.2\textwidth}
			\centering
			\includegraphics[width=\textwidth]{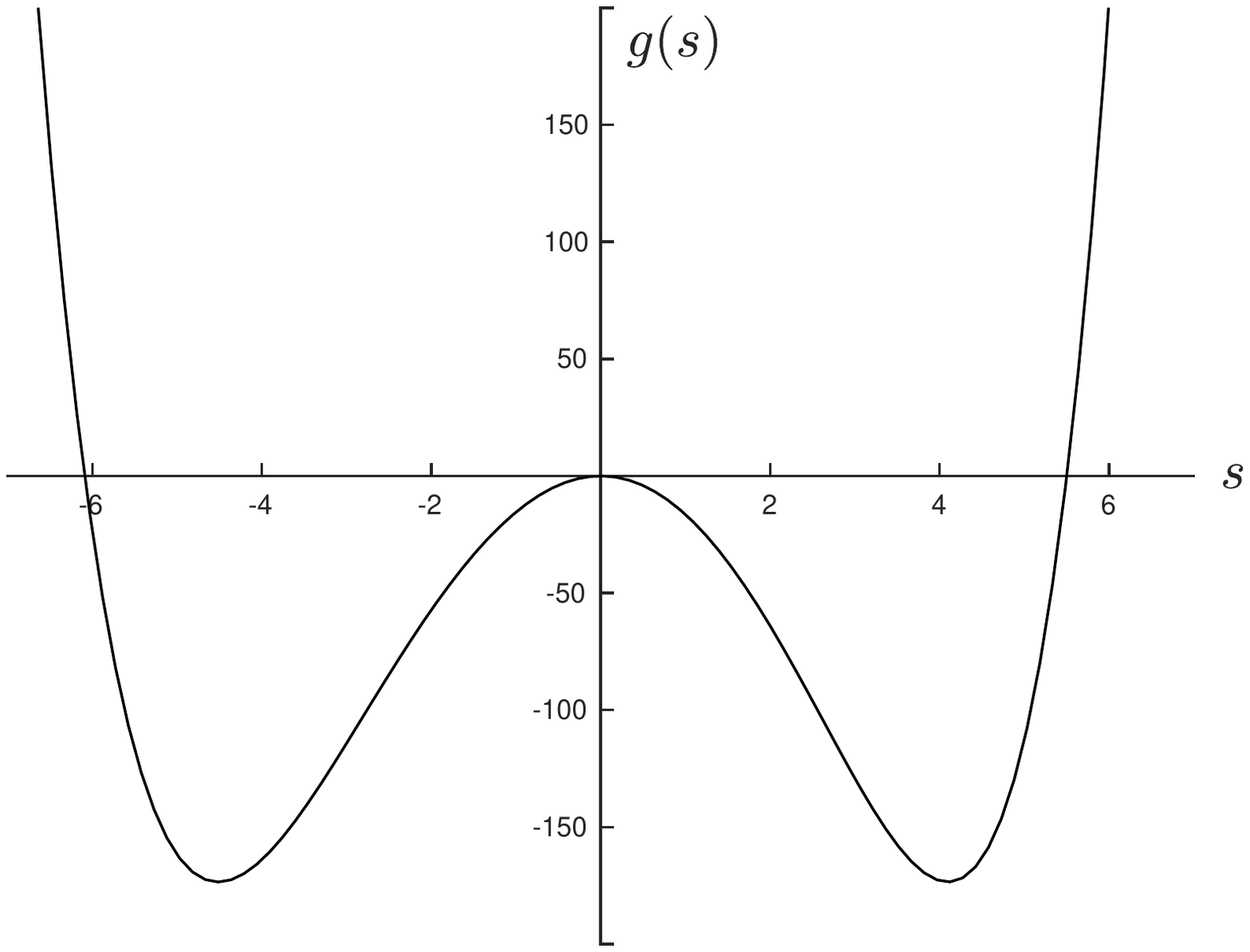}
			\subcaption{}
		\end{minipage}\hfill
            \begin{minipage}{0.2\textwidth}
			\centering
			\includegraphics[width=\textwidth]{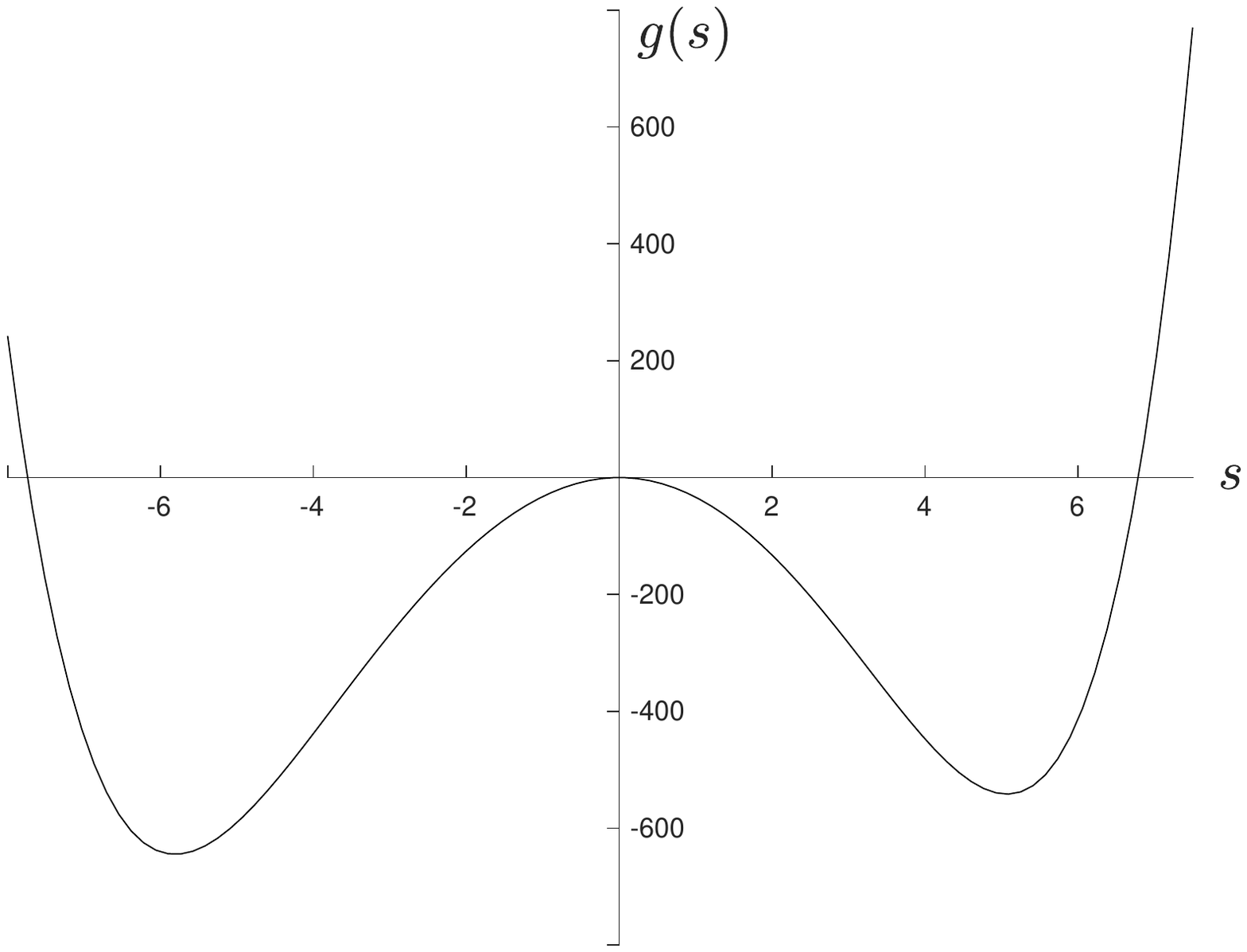}
			\subcaption{}
		\end{minipage}
		\caption{The function \(g\) with \(d = 1, e = 0, f = 1\) at (a) \(t = 5\); (b) \(t = t_0\); (c) \(t = -2\); (d) \(t = t^*\); and (e) \(t = -100\).}\label{fig:6thgplots}
	\end{figure}
	\begin{figure}[!ht]
		\centering
		\includegraphics[width=0.3\textwidth, angle=-90]{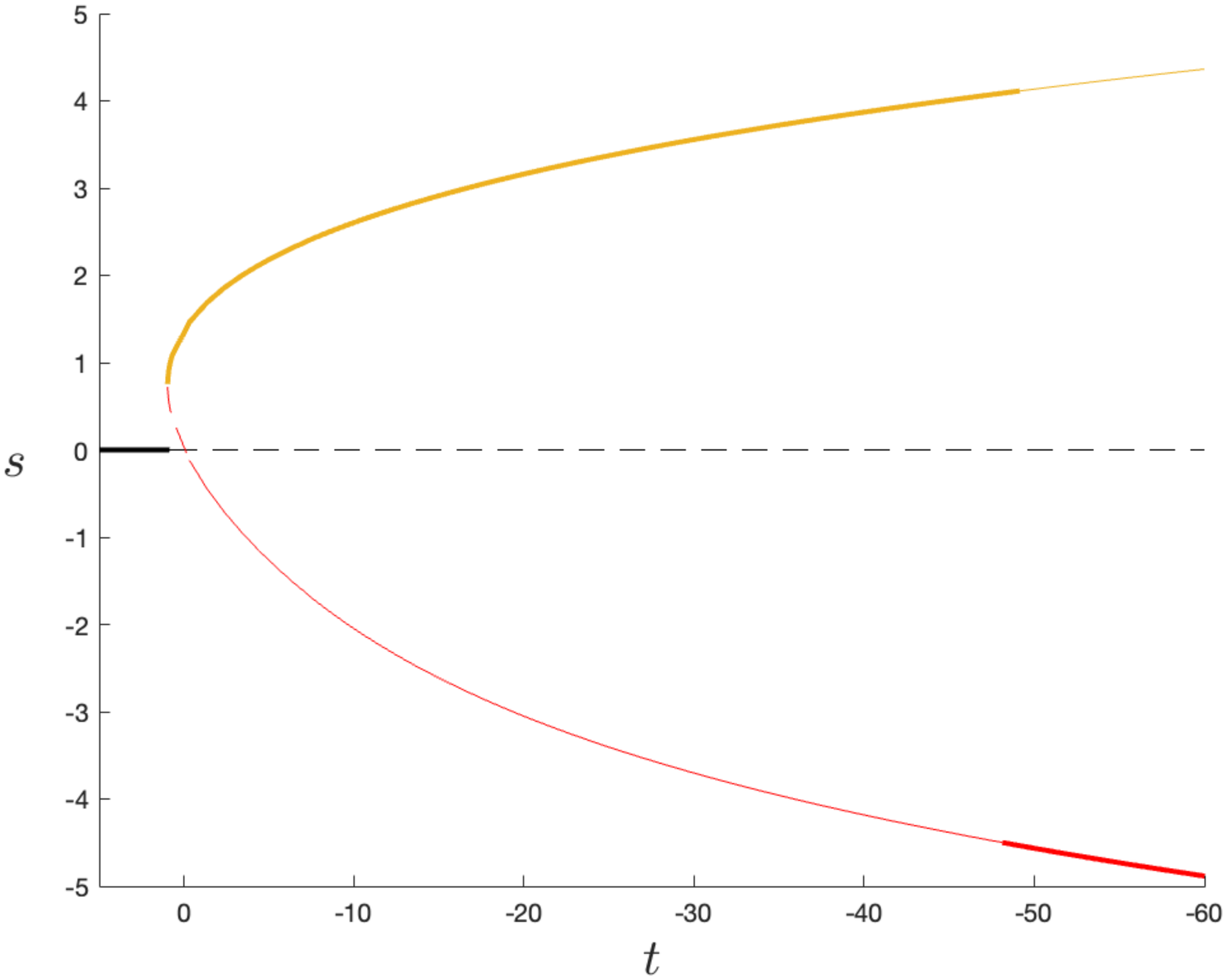}
		\caption{The stationary points of the function \(g\) for decreasing temperature. Bold lines indicate a global minimium, thin solid lines indicate a local minimum and dashed lines indicate instability (negative second derivative of $g$).} \label{fig:6thunisp}
	\end{figure}
	
	\subsection{Biaxial Stationary Points of the Sixth-Order Potential}
	
	Next, we consider critical points of $f_B$ in \eqref{eq:f2} in the full class of LdG $\mathbf{Q}$-tensors i.e.
	\begin{equation} \label{eq:b}
		\textbf{Q} = s\bigg(\boldsymbol{n}\otimes\boldsymbol{n} - \frac{1}{3}\textbf{I}\bigg) + p\bigg(\boldsymbol{m}\otimes\boldsymbol{m} - \frac{1}{3}\textbf{I}\bigg),
	\end{equation}
	i.e. we substitute \eqref{eq:b} into \eqref{eq:f2} and compute the stationary points in terms of the pairs $(s, p)$.
	
	There are no biaxial stationary points of the fourth-order bulk potential (see \parencite[Proposition 1]{Majumdar2010}). In Figure \ref{fig:4thstpts}, we plot the stationary points of the fourth-order bulk potential in (\ref{eq:f1}). The isotropic phase is the global minimiser above \(t = 1.125\), and below this temperature we have four nonzero stationary points, two of which are global minimisers: one minimiser has positive $s$ and zero $p$ (yellow solid line in Fig. \ref{fig:4thstpts}); the second minimiser has $s = p <0$ (red solid line in Fig. \ref{fig:4thstpts}). Both minimisers are uniaxial $\mathbf{Q}$-tensors and rotations of each other. The same relationship holds for the two unstable stationary points: both correspond to a uniaxial configuration with negative order parameter. Thus, Figure \ref{fig:4thstpts} shows that below \(t = 1.125\), there is one stable uniaxial stationary point with positive order parameter, and one unstable uniaxial stationary point with negative order parameter; there are no biaxial stationary points; and the isotropic phase loses stability for $t<0$. These facts are well-known in the literature about the fourth-order potential.
	
	\begin{figure}[!ht]
		\begin{minipage}{0.4\textwidth}
			\centering\hspace*{-0.5cm}\includegraphics[width=0.75\textwidth,angle=-90]{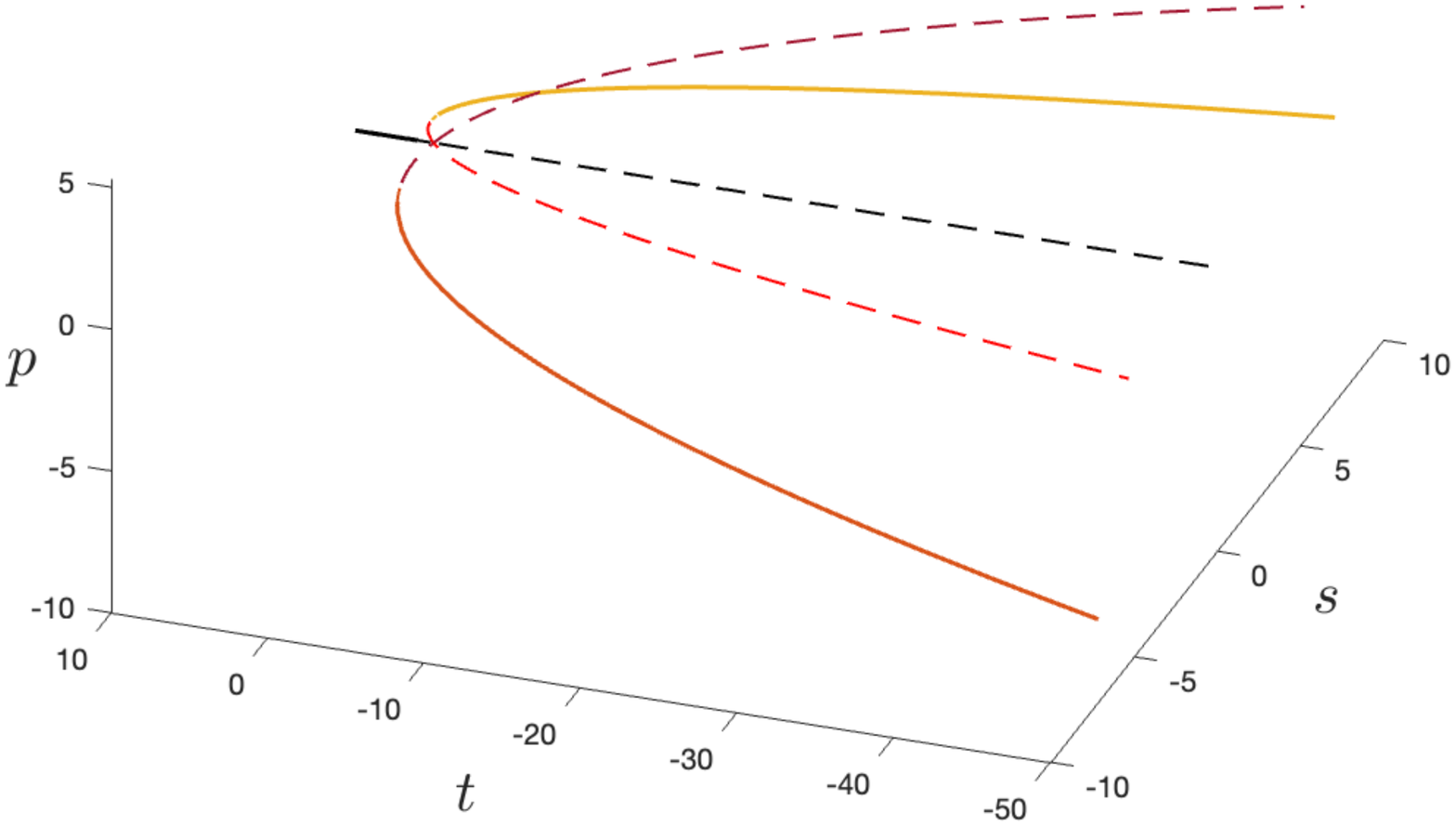}
			\subcaption{}
		\end{minipage}\hfill
		\begin{minipage}{0.29\textwidth}
			\centering\hspace*{-0.25cm}\includegraphics[width=0.75\textwidth,angle=-90]{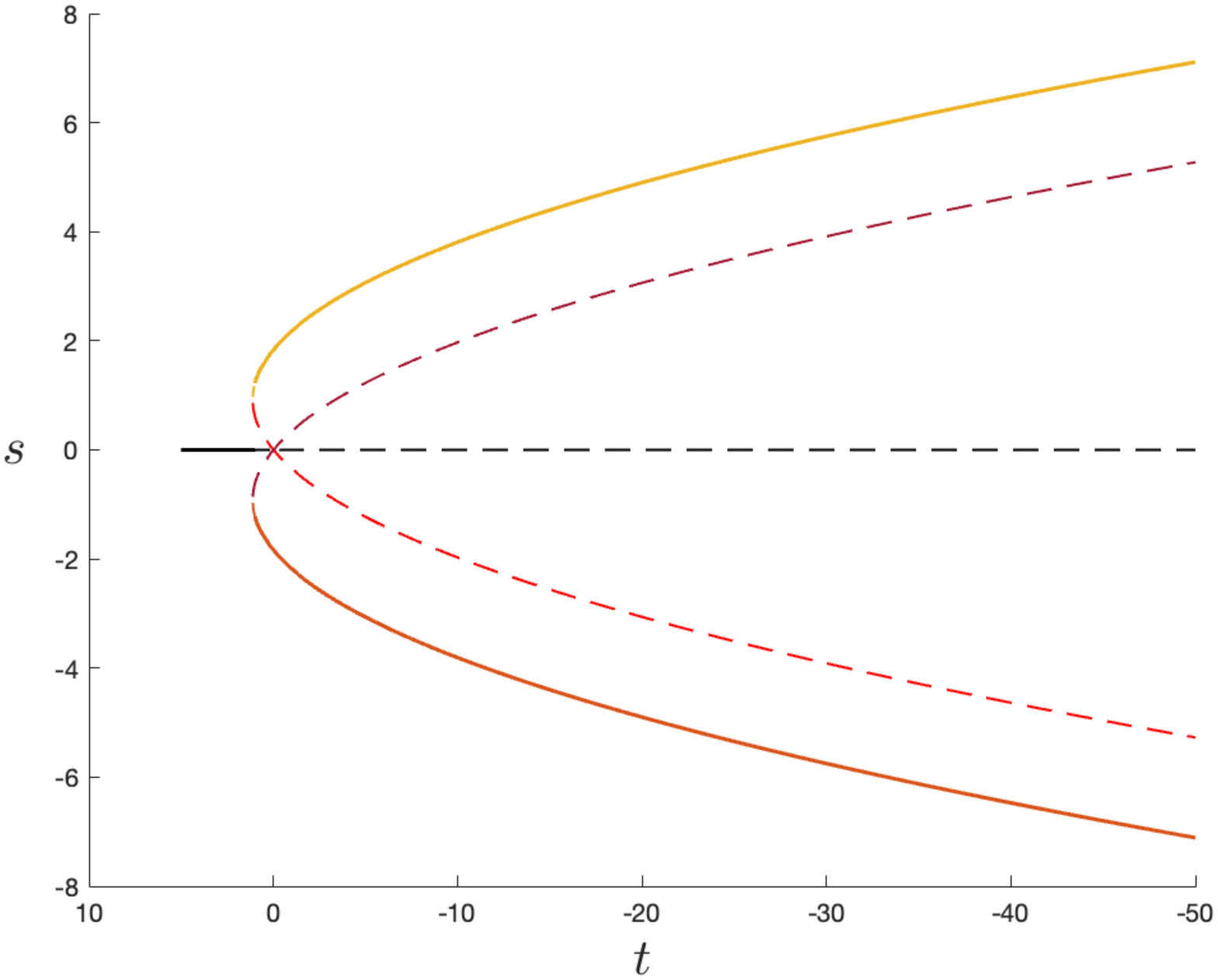}
			\subcaption{}
		\end{minipage}\hfill
		\begin{minipage}{0.29\textwidth}
			\centering\hspace*{-0.25cm}\includegraphics[width=0.75\textwidth,angle=-90]{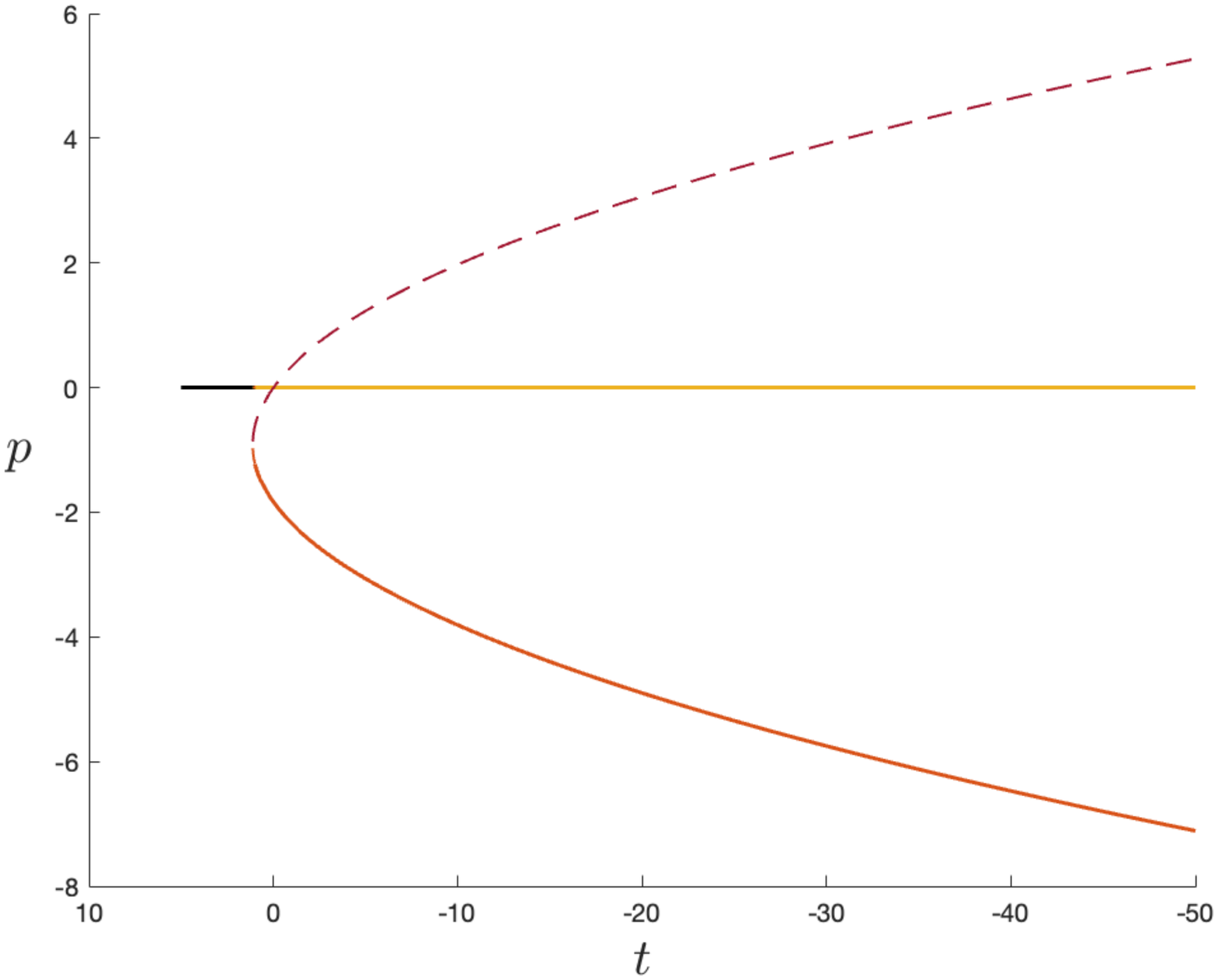}
			\subcaption{}
		\end{minipage}
		\caption{Stationary points of the fourth-order potential, for the temperature range \(t = 5\) to \(t = -50\). (a) Both scalar order parameters, \(s\) and \(p\), plotted against \(t\). (b) Scalar order parameter \(s\) plotted against \(t\). (c) Scalar order parameter \(p\) plotted against \(t\).} \label{fig:4thstpts}
	\end{figure}
	
	We fix \(d = 1, e = 0, f = 1\) and plot the critical points of \eqref{eq:f2} (the sixth-order potential) in Figure \ref{fig:6thstpts}. There are certain similarities to Figure \ref{fig:4thstpts} i.e. the isotropic phase is the global minimiser at high temperatures; there are four non-zero stationary points for moderate temperatures below some transition temperature; and the isotropic phase loses stability at \(t = 0\). These stationary points correspond to one uniaxial global minimiser and one unstable uniaxial critical point for the same reasons as in the fourth-order case, as detailed above. The two uniaxial stationary points emerge at the approximate transition temperature \(t = 1\), with \((s,p) =(s_+,0)\) being the global minimiser (red and yellow solid lines in Fig. \ref{fig:6thstpts}), and \((s,p)=(s_-, 0)\) being the unstable stationary point (purple and red dashed lines), where \(s_+\) and \(s_-\) are defined above. The stationary point $(s_+,0) $ remains stable until biaxial stationary points appear at approximately \(t = -11.6\), at which point (unlike with the fourth-order potential in \eqref{eq:f1}) there are no stable uniaxial stationary points and there is a \emph{unique} global biaxial minimiser of \eqref{eq:f2} (three blue solid lines).
	
	\begin{figure}[!ht]
		\begin{minipage}{0.4\textwidth}
			\hspace*{-0.5cm}\includegraphics[width=0.75\textwidth,angle=-90]{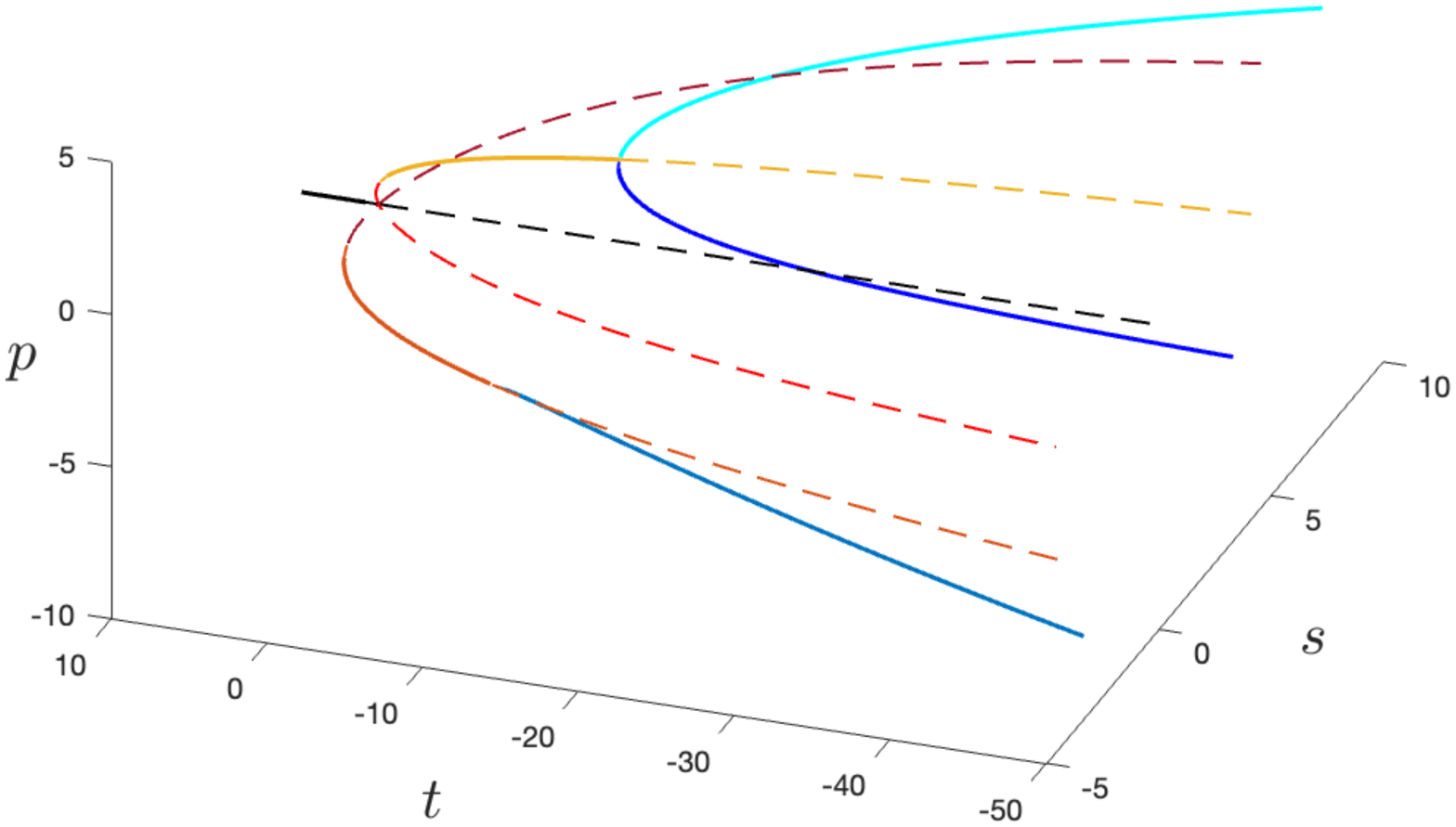}
			\subcaption{}
		\end{minipage}\hfill
		\begin{minipage}{0.29\textwidth}
			\hspace*{-0.25cm}\includegraphics[width=0.75\textwidth,angle=-90]{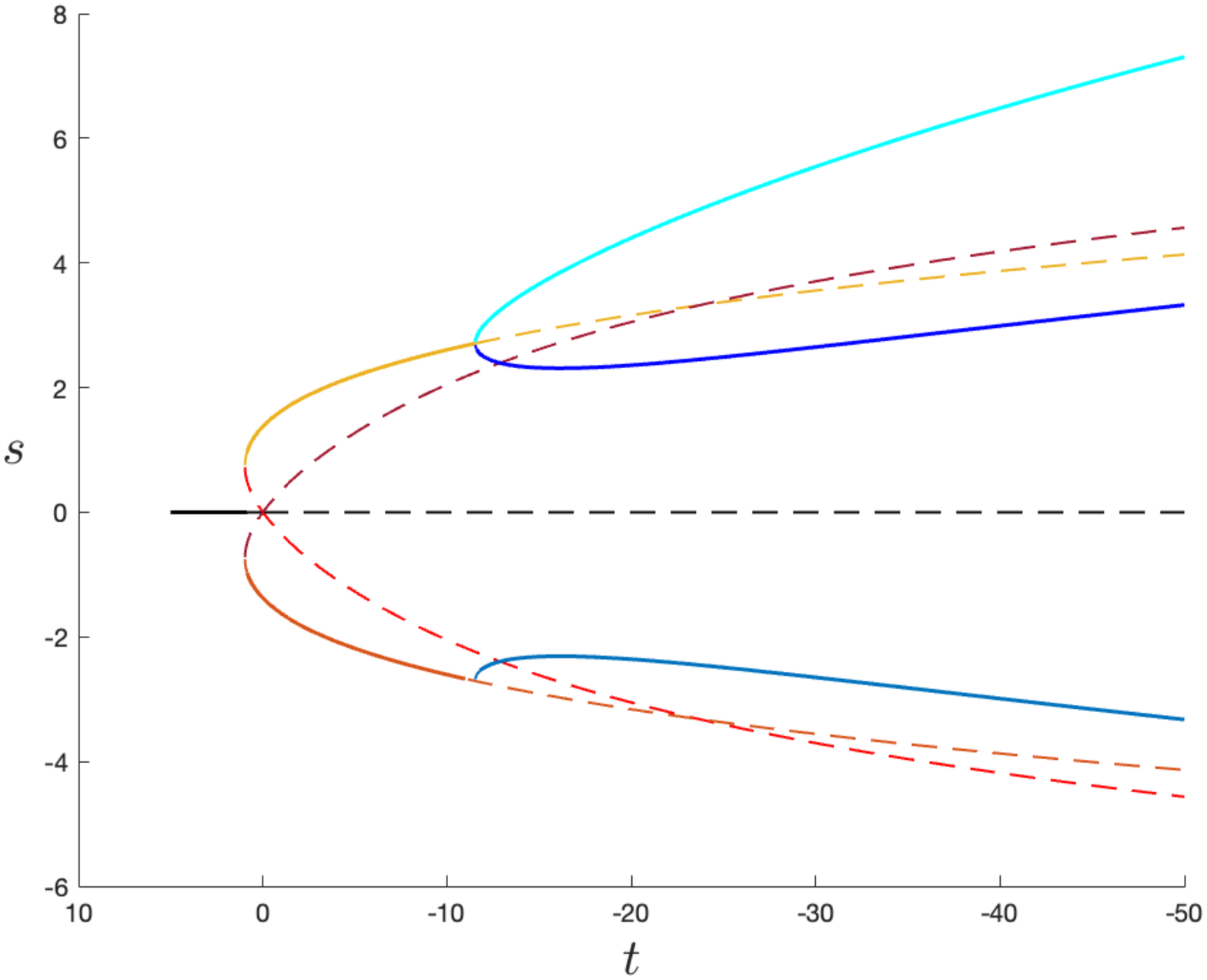}
			\subcaption{}
		\end{minipage}\hfill
		\begin{minipage}{0.29\textwidth}
			\hspace*{-0.25cm}\includegraphics[width=0.75\textwidth,angle=-90]{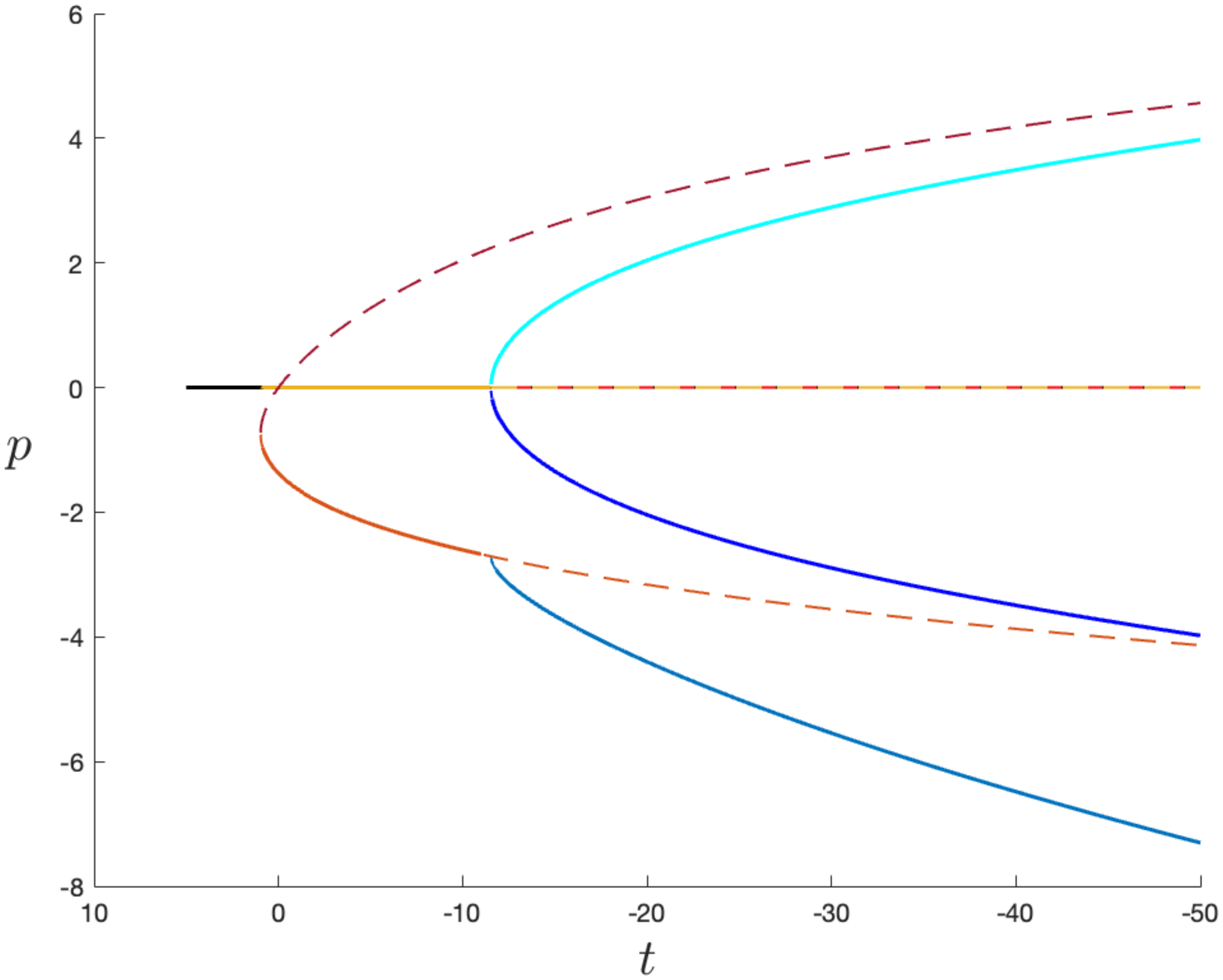}
			\subcaption{}
		\end{minipage}
		\caption{Stationary points of the sixth-order potential with \(d = 1, e = 0, f = 1\), for \(t = 5\) to \(t = -50\). (a) Scalar order parameters, \(s\) and \(p\), plotted against \(t\). (b) The scalar order parameter \(s\) plotted against \(t\). (c) The scalar order parameter \(p\) plotted against \(t\). Yellow and red lines label uniaxial stationary points; and blue lines label biaxial stationary points.} \label{fig:6thstpts}
	\end{figure}

	\section{Analysis of the Radial Hedgehog Solution}
 \label{sec:analysis}
	The RH solution has been studied extensively with the fourth-order potential in the literature. In particular, for the LdG energy with the fourth-order potential in \eqref{eq:f1}, there are strong analytic results on the existence of the RH solution, the corresponding order parameter $s^* $ is positive away from the origin, is monotonic, bounded and there is a unique RH solution for $t<0$ \parencite{Majumdar2012,Lamy2013}. Furthermore, it is known that the RH solution is stable for sufficiently small droplets, and is unstable for large droplets and for low temperatures \parencite{Majumdar2012, HenaoMajumdar2017,IgnatNguyen2015,RossoVirga1996}. In this section, we perform a parallel analysis of the RH solution with the sixth-order potential \eqref{eq:f2} to understand the dependence of the RH solution on the choice of $f_B$ and the new possibilities offered by the more general nature of the sixth-order potential in \eqref{eq:f2}. 
	
	Our first result concerns the existence of the RH solution and is analogous to Proposition 2.1 in \parencite{Majumdar2012}. 
 The proof is similar to those in \parencite{Majumdar2012} 
 and is therefore omitted.
	\begin{prop} 
		\begin{enumerate}[(a)]
			\item Consider the energy functional
			\begin{equation}
				I[s] = \int_{0}^{1}\Bigg(\varepsilon^2\bigg(\frac{1}{2}\bigg(\frac{ds}{dr}\bigg)^2 + \frac{2}{r^2}s^2\bigg) + \frac{t}{3}s^2 - \frac{2\sqrt{6}}{9}s^3 + \frac{2}{9}s^4 + \frac{4d}{135}s^5 + \frac{4e}{81}s^6 + \frac{2(f - e)}{243}s^6\Bigg)r^2\,dr, \label{LdGs}
			\end{equation}
			defined for functions \(s \in \mathcal{A}_s\). There exists a global minimiser \(s^*\in\mathcal{A}_s\) for \(I\). The function \(s^*\) is a solution of the ordinary differential equation
			\begin{equation}
				\varepsilon^2\Bigg(\frac{d^2s}{dr^2} + \frac{2}{r}\frac{ds}{dr} - \frac{6}{r^2}s\Bigg) = ts - \sqrt{6}s^2 + \frac{4}{3}s^3 + \frac{2d}{9}s^4 + \frac{4e}{9}s^5 + \frac{2(f - e)}{27}s^5, \label{sODE}
			\end{equation}
			subject to the boundary conditions
			\begin{equation}
				s(0) = 0, \quad s(1) = s_+. \label{sBCs}
			\end{equation}
			The global minimiser \(s^*\) is analytic for all \(r \geq 0\).
			
			\item The RH solution is defined in \eqref{Qrh}, where \(s^*\) is a global minimiser of \(I\) in the admissible space \(\mathcal{A}_s\), and the RH solution is a critical point of the LdG energy functional {\normalfont (\ref{LdGgeneral})}.
			\item The function \(s^*\) satisfies \((s^*)'(0) = 0\).
		\end{enumerate}
	\end{prop} Next, we derive a maximum principle which yields upper bounds for $s^*$ in \eqref{Qrh}.
	\begin{prop}
		A global minimiser \(\normalfont \textbf{Q}^*\) of the LdG free energy {\upshape (\ref{6thNDLdG})}, in the class of uniaxial {\upshape \textbf{Q}}-tensors, in the admissible space, \(\normalfont \mathcal{A}_{\textbf{Q}}\) in {\upshape (\ref{AQ})}, satisfies the upper bound \(\normalfont |\textbf{Q}^*|^2 \leq \frac{2}{3}\max\big\{s_+^2, s_-^2\big\}\) on \(B(0,1)\), where \(s_+\) and \(s_-\) are the two nonzero stationary points of \((\ref{gdefn})\).
	\end{prop}
	
	In the next proposition, we work with fixed $d=1, e=0, f=1$ and with temperatures for which the potential \eqref{eq:f2} has a uniaxial global minimiser with positive order parameter $s_+$, consistent with the imposed Dirichlet condition in \eqref{sBCs}. In this case, we can prove that $s^*$ is bounded, positive, monotonic and unique, by direct analogy with the results for \eqref{eq:f1}. The differences arise for low temperatures, for which \eqref{eq:f2} has a biaxial global minimiser and no stable uniaxial critical points, and we use heuristic arguments to show that $s^*$ is negative and non-monotonic deep in the nematic phase. 
    Recall that \(s_+\) is the largest positive minimiser of the function \(g\) in \eqref{gdefn}.
	Since \(g(s) \to +\infty\) as \(s \to +\infty\),
 then \(g'(s) > 0\) for \(s > s_+\). Therefore, \(s_+\) increases as $|t|$ increases  for $t<0$ and bounds on $s_+$ can be translated to bounds for $t$.
	
	\begin{prop} \label{prop:2}
		Let \(s^*\) be the global minimiser of \(I\) in {\normalfont (\ref{LdGs})} in the moderately low temperature regime for which \eqref{eq:f2} has a global uniaxial minimiser, characterised by \(t < 0\), 
		\begin{equation} \label{eq:ranget}
			s_+^2 - \frac{15\sqrt{6}}{2d} < 0,
		\end{equation}
  and
  \begin{equation} \label{eq:m2}
			\frac{4(f + 5e)}{81}s_+^3 + \frac{4d}{27}s_+^2 + \frac{8}{9}s_+ - \frac{2\sqrt{6}}{3} < 0.
		\end{equation}
  Then, $s^*$ is unique;
 vanishes at the origin; satisfies the bounds \(0 \leq s^* \leq s_+\); and is positive and monotonic for $r>0$.
	\end{prop}

 The proof of Proposition~\ref{prop:2} follows from analogous arguments for the fourth-order potential in \parencite{Lamy2013}, precisely because the sixth-order potential \eqref{eq:f2} has a uniaxial global minimiser in the temperature ranges specified by \eqref{eq:ranget} and \eqref{eq:m2}. Some details are given in the Appendix, and the same arguments do not work when \eqref{eq:f2} admits a biaxial minimiser for lower temperatures.
	
	The next result demonstrates that the RH solution is the only LdG critical point for droplets of sufficiently small radius with the sixth-order potential, and is hence globally stable in this regime. This follows from the convexity of the LdG free energy with polynomial bulk potentials for small domains.
	\begin{prop}
		For \(\varepsilon\) sufficiently large, the radial hedgehog configuration {\normalfont \(\textbf{Q}^*\)} in {\normalfont (\ref{Qrh})} is the unique critical point of the Landau–de Gennes free energy {\normalfont (\ref{6thNDLdG})}.
	\end{prop}
	\begin{proof}
		First, we show that a critical point, \(\textbf{Q}^*\), of \(\mathcal{F}\) in the admissible space (\ref{AQ}) satisfies the upper bound
		\begin{equation}
			|\textbf{Q}^*| \leq \max\big\{M(t, d, e, f), |\textbf{Q}_{s_+}|\big\} =: M'
		\end{equation}
		on \(\overline{B(0,1)}\), where \(M\) is a constant depending only on \(t, d, e,\) and \(f\).
		
		We assume that the function \(|\textbf{Q}^*|:B(0,1) \to \mathbb{R}\) attains its maximum at the interior point \(\boldsymbol{r}^* \in B(0,1)\). Recall that \(\textbf{Q}^*\) is a solution of the Euler–Lagrange equations (\ref{6thELQ}). We multiply both sides of (\ref{6thELQ}) by \(Q_{ij}\) to find
		\begin{equation} \nonumber
			\frac{\varepsilon^2}{2}\Delta|\textbf{Q}^*|^2 = t|\textbf{Q}^*|^2 - 3\sqrt{6}\tr\textbf{Q}^{*3} + 2|\textbf{Q}^*|^4 + d|\textbf{Q}^*|^2\tr\textbf{Q}^{*3} + e|\textbf{Q}^*|^6 + \frac{(f - e)}{6}\big(\tr\textbf{Q}^{*3}\big)^2
		\end{equation}
		at \(\boldsymbol{r}^*\), since \(|\nabla\textbf{Q}^*| + Q_{ij}\Delta Q^*_{ij} = \frac{1}{2}\Delta|\textbf{Q}^*|^2\), and \(|\nabla\textbf{Q}^*| = 0\) at \(\boldsymbol{r}^*\). Note that \(\Delta|\textbf{Q}^*|^2 \leq 0\) at \(\boldsymbol{r}^* \in B(0,1)\) by assumption.
	Define
		$
			h(\textbf{Q}) = t|\textbf{Q}|^2 - 3\sqrt{6}\tr\textbf{Q}^{3} + 2|\textbf{Q}|^4 + d|\textbf{Q}|^2\tr\textbf{Q}^{3} + e|\textbf{Q}|^6 + \frac{(f - e)}{6}\big(\tr\textbf{Q}^{3}\big)^2.
		$
		Recalling that
		$		-\frac{1}{\sqrt{6}}|\textbf{Q}|^3 \leq \tr\textbf{Q}^{3} \leq \frac{1}{\sqrt{6}}|\textbf{Q}|^3, 
		$
		by \parencite[Lemma 1]{Majumdar2010}, consider two cases. First, if \(f - e \geq 0\), then
		\begin{equation} \nonumber
			h(\textbf{Q}) \geq t|\textbf{Q}|^2 - 3|\textbf{Q}|^3 + 2|\textbf{Q}|^4 - \frac{d}{\sqrt{6}}|\textbf{Q}|^5 + e|\textbf{Q}|^6 =: H_1(|\textbf{Q}|).
		\end{equation}
		The function \(H_1(|\textbf{Q}|)\) has \(n \leq 6\) real roots, \(\{|\textbf{Q}_{1i}|\}_{i = 1}^{n}\), with \(|\textbf{Q}_{11}| \leq ... \leq |\textbf{Q}_{1n}|\), and \(H_1\) is monotonically increasing for \(|\textbf{Q}| > |\textbf{Q}_{1n}|\) since \(e > 0\). If \(f - e < 0\), then
		\begin{equation} \nonumber 
			h(\textbf{Q}) \geq t|\textbf{Q}|^2 - 3|\textbf{Q}|^3 + 2|\textbf{Q}|^4 - \frac{d}{\sqrt{6}}|\textbf{Q}|^5 + e|\textbf{Q}|^6 + \frac{(f - e)}{36}|\textbf{Q}|^6 =: H_2(|\textbf{Q}|).
		\end{equation}
		Similarly, \(H_2(|\textbf{Q}|)\) has \(n \leq 6\) real roots, \(\{|\textbf{Q}_{2i}|\}_{i = 1}^{n}\), with \(|\textbf{Q}_{21}| \leq ... \leq |\textbf{Q}_{2n}|\), and \(H_2\) is monotonically increasing for \(|\textbf{Q}| > |\textbf{Q}_{2n}|\) since \(e > 0, f > 0\).
		
		If
		\begin{equation} \nonumber
			|\textbf{Q}^*(\boldsymbol{r}^*)| > \begin{cases}
				|\textbf{Q}_{1n}|, \quad &\text{if} \,\, f - e \geq 0, \\
				|\textbf{Q}_{2n}|, \quad &\text{if} \,\, f - e < 0,
			\end{cases}
		\end{equation} 
		then \(\Delta|\textbf{Q}^*|^2 > 0\) at \(\boldsymbol{r}^*\), which is a contradiction.
		Combining the above, set 
		\begin{equation}
			M(t, d, e, f) := \begin{cases}
				|\textbf{Q}_{1n}|, \quad &\text{if} \,\, f - e \geq 0, \\
				|\textbf{Q}_{2n}|, \quad &\text{if} \,\, f - e < 0,
			\end{cases}
		\end{equation}
		and obtain $
			|\textbf{Q}^*| \leq \max\big\{M(t, d, e, f), |\textbf{Q}_{s_+}|\big\}
		$	on \(\overline{B(0,1)}\).
		
		Next, we demonstrate the convexity of the LdG free energy (\ref{6thNDLdG}) for sufficiently large \(\varepsilon\), closely following arguments in \parencite{Lamy2013, HanHarris2021}. Let
			$ X = \big\{\textbf{Q}\in W^{1,2}(B(0,1), \bar{S}): |\textbf{Q}| \leq M'\big\}$. Then for
		\(\textbf{Q}_u, \textbf{Q}_v \in X\),
		\begin{align} \nonumber
			\mathcal{F}\bigg[\frac{1}{2}(\textbf{Q}_u + \textbf{Q}_v)\bigg] &= \int_{B(0,1)}\Bigg(\frac{\varepsilon^2}{8}|\nabla\textbf{Q}_u + \nabla\textbf{Q}_v|^2 + f_B\bigg(\frac{1}{2}(\textbf{Q}_u + \textbf{Q}_v)\bigg)\Bigg)\,dV \\
			&= \begin{aligned}[t]
				&\frac{1}{2}\mathcal{F}[\textbf{Q}_u] + \frac{1}{2}\mathcal{F}[\textbf{Q}_v] - \frac{\varepsilon^2}{8}||\nabla(\textbf{Q}_u - \textbf{Q}_v)||_{L^2}^2 \\
				&+ \int_{B(0,1)}\Bigg(f_B\bigg(\frac{1}{2}(\textbf{Q}_u + \textbf{Q}_v)\bigg) - \frac{1}{2}f_B(\textbf{Q}_u) - \frac{1}{2}f_B(\textbf{Q}_v)\Bigg)\,dV 
			\end{aligned}
		\end{align}
		where we have used the fact that
			$|\nabla\textbf{Q}_u + \nabla\textbf{Q}_v|^2 = 2|\nabla\textbf{Q}_u|^2 + 2|\nabla\textbf{Q}_v|^2 - |\nabla\textbf{Q}_u - \nabla\textbf{Q}_v|^2. $
		By the Poincar\'e inequality, we have that
		\begin{equation*}
			-\frac{1}{8}||\nabla(\textbf{Q}_u - \textbf{Q}_v)||_{L^2}^2 \leq - c_1||\textbf{Q}_u - \textbf{Q}_v||_{L^2}^2,
		\end{equation*}
		for some positive constant \(c_1\). Therefore,
		\begin{multline*}
			\mathcal{F}\bigg[\frac{1}{2}(\textbf{Q}_u + \textbf{Q}_v)\bigg] \leq  \frac{1}{2}\mathcal{F}[\textbf{Q}_u] + \frac{1}{2}\mathcal{F}[\textbf{Q}_v] - c_1\varepsilon^2||(\textbf{Q}_u - \textbf{Q}_v)||_{L^2}^2 \\
			+ \int_{B(0,1)}f_B\bigg(\frac{1}{2}(\textbf{Q}_u + \textbf{Q}_v)\bigg) - \frac{1}{2}f_B(\textbf{Q}_u) - \frac{1}{2}f_B(\textbf{Q}_v)\,dV.
		\end{multline*}
		Furthermore, we note that
		\begin{equation*}
			f_B\bigg(\frac{x + y}{2}\bigg) - \frac{1}{2}f_B(x) - \frac{1}{2}f_B(y) \leq ||f_B||_{W^{2,\infty}}|x - y|^2,
		\end{equation*}
		for \(x, y\) satisfying \(|x|, |y| \leq M'\). Hence, for some \(c_2 = c_2(M', f_B) > 0\), we may write
		\begin{equation*}
			\mathcal{F}\bigg[\frac{1}{2}(\textbf{Q}_u + \textbf{Q}_v)\bigg] \leq
			\frac{1}{2}\mathcal{F}[\textbf{Q}_u] + \frac{1}{2}\mathcal{F}[\textbf{Q}_v] - \frac{c_1\varepsilon^2}{2}||\textbf{Q}_u - \textbf{Q}_v||_{L^2}^2 - c_2\bigg(\frac{c_1\varepsilon^2}{2c_2} - 1\bigg)||\textbf{Q}_u - \textbf{Q}_v||_{L^2}^2.
		\end{equation*}
		Then, if \(\varepsilon^2 > \frac{2c_2}{c_1}\), we find that
		\begin{equation*}
			\mathcal{F}\bigg[\frac{1}{2}(\textbf{Q}_u + \textbf{Q}_v)\bigg] < \frac{1}{2}\mathcal{F}[\textbf{Q}_u] + \frac{1}{2}\mathcal{F}[\textbf{Q}_v],
		\end{equation*}
		\(\forall \textbf{Q}_u, \textbf{Q}_v \in X, \textbf{Q}_u \neq \textbf{Q}_v\). Thus, \(\mathcal{F}\) is strictly convex on \(X\).
		
		Let us assume for the remainder of the proof that we are working with \(\varepsilon\) large enough to guarantee convexity of \(\mathcal{F}\) in (\ref{6thNDLdG}). To show that a critical point of \(\mathcal{F}\) is unique, let us assume that there exist two distinct solutions \(\textbf{Q}_1\) and \(\textbf{Q}_2\) of (\ref{6thELQ}) in \(X\), as is done in the proof of \parencite[Lemma 8.3]{Lamy2014}. Then, for \(v \in [0,1]\), the derivative of \(\mathcal{F}[v\textbf{Q}_1 + (1 - v)\textbf{Q}_2]\) vanishes at \(v = 0\) and \(v = 1\). However, the strict convexity of \(\mathcal{F}\) implies that \(\mathcal{F}\) can have only one critical point. Therefore, \(\textbf{Q}_1\) and \(\textbf{Q}_2\) cannot both be solutions of the Euler–Lagrange equations (\ref{6thELQ}), so a critical point of \(\mathcal{F}\) must be unique.
		
		Finally, Proposition 1 guarantees the existence of a RH solution for any \(\varepsilon\) and we are also guaranteed the existence of a global LdG energy minimiser of \eqref{6thNDLdG} for all $\varepsilon$, so the RH configuration \(\textbf{Q}^*\) in (\ref{Qrh}) is the unique critical point and consequently, the unique global minimiser of the LdG free energy (\ref{6thNDLdG}) when \(\varepsilon\) is sufficiently large.
	\end{proof}
	
	Next, we demonstrate that the RH solution is not globally energy minimising for the LdG energy \eqref{6thNDLdG}, in the low temperature regime, by constructing a biaxial perturbation with lower energy, following arguments as in Proposition 3.3 in \parencite{Majumdar2012}.
	\begin{prop}
		The RH solution \(\normalfont \textbf{Q}^*\) in \eqref{Qrh} is not the global minimiser of the LdG free energy {\normalfont (\ref{6thNDLdG})} in the admissible space \(\mathcal{A}_{\textbf{Q}}\) when $t<0$ and $|t|$ is sufficiently large. In particular, the biaxial state
		\begin{equation}
			\normalfont \hat{\textbf{Q}}(\boldsymbol{r}) = \begin{cases}
				\textbf{Q}^*(\boldsymbol{r}) + \dfrac{1}{(r^2 + 12)^2}\bigg(1 - \dfrac{r}{\sigma}\bigg)\bigg(\boldsymbol{z}\otimes\boldsymbol{z} - \dfrac{1}{3}\textbf{I}\bigg), &0 \leq r \leq \sigma, \\
				\textbf{Q}^*(\boldsymbol{r}), &\sigma \leq r \leq 1, \label{perturbation}
			\end{cases}
		\end{equation}
		where \(\boldsymbol{z}\) is the unit vector in the \(z\)-direction, has lower LdG free energy than \(\normalfont \textbf{Q}^*\) for \(\sigma = 0.1\).
	\end{prop}
	
	The above results make evident many parallels between the RH solution with the fourth-order bulk potential in the literature and the RH solution with the sixth-order bulk potential, at least for moderately low temperatures specified by \eqref{eq:ranget} and \eqref{eq:m2}. Key differences are that we do not have an explicit expression for \(s_+\) with the sixth-order potential \eqref{eq:f2}, and that there are parameter regimes for which the RH scalar order parameter \(s^*\) might not be unique and monotonic, and could be negative. We explore this further using some heuristic arguments, working in a parameter regime for which $g(s)$ in \eqref{gdefn} has two minimisers - a local positive minimiser $s_+$ and a global negative minimiser, $s_-$ i.e. deep in the nematic phase. In terms of \eqref{eq:f2}, these minima of $g$ do not translate to stable critical points of \eqref{eq:f2}, and \eqref{eq:f2} has a biaxial minimiser in these parameter regimes. Consider a profile, $s_1:[0,1] \to \mathbb{R}$, in the admissible space \eqref{eq:As} for which $s_1 \geq 0$ for $r\in [0,1]$. Then the energy is bounded from below by:
 \[
 I[s_1] \geq g(s_+),
 \]
 where $I$ is given by \eqref{LdGs} and we use the fact that $s_+$ is the global minimiser of $g$ for non-negative $s$.
 Consider a competitor map $s_2:[0,1] \to \mathbb{R} $ in the admissible space \eqref{eq:As}, with $s_2=s_{-}$ for $\varepsilon < r < 1- \varepsilon$ for $\varepsilon$ sufficiently small. Assuming a linear transition layer near $r=0$, where $s_2(0) = 0$ and a linear transition layer near $r=1$ to match the boundary condition, $s_2(1) = s_+$, $I[s_2]$ is bounded from above by:
 \[
 I[s_2] \leq g(s_-) (1 - 2\varepsilon) + C \varepsilon, 
 \] where the positive constant $C$ depends on $s_+, s_-$. Given that $d, e, f$ are fixed by assumption, this implies that the constant $C$ only depends on $t$. For fixed elastic constants, $\varepsilon$ only depends on the droplet radius $R$. Comparing $I[s_1]$ and $I[s_2]$, we deduce that
 $
 I[s_2] < I[s_1]
 $ if $g(s_-) < g(s_+)$ and $\varepsilon$ is sufficiently small or if $R$ is sufficiently large. Since the lower bound for $I[s_1]$ is valid for all RH order parameter profiles with non-negative order parameter, we deduce that provided $\varepsilon$ is sufficiently small, the global minimiser of \eqref{LdGs} cannot be positive for all $r\in [0,1]$, for sufficiently low temperatures.

    These heuristics can be verified numerically, using a finite element method to numerically compute solutions of the ODE (\ref{sODE}) for \(t = -100\), in a large droplet specified by \(\varepsilon = 0.1\). We numerically obtain at least two solutions in Figure \ref{fig:hedgehogprofiles}. Figure \ref{fig:+vehedegehog} follows from the initial condition \(s(r) = 0, r \in [0,1]\), and the initial guess \(s(r) = 0.5s_-, r \in [0,1]\) is used to numerically compute Figure \ref{fig:-vehedgehog}, where \(s_-\) is the negative minimiser of (\ref{gdefn}) at \(t = -100\). The second profile in Figure \ref{fig:-vehedgehog} has lower energy than the non-negative profile in Figure \ref{fig:+vehedegehog}.
	
	\begin{figure}[!ht]
		\begin{minipage}{0.49\textwidth}
			\centering
			\hspace*{-0.2cm}\includegraphics[width=0.6\textwidth,angle=-90]{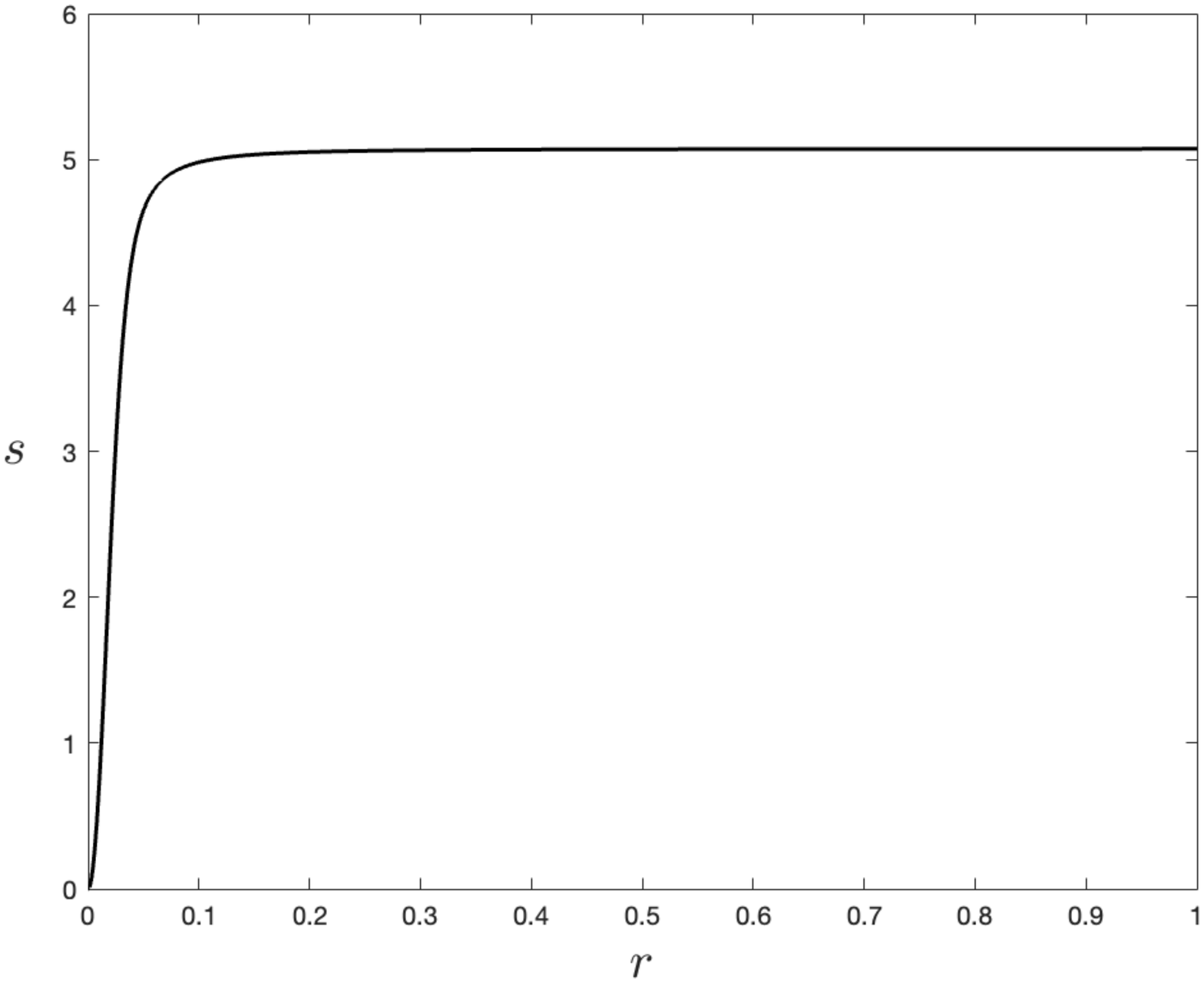}
			\subcaption{}\label{fig:+vehedegehog}
		\end{minipage}\hfill
		\begin{minipage}{0.49\textwidth}
			\centering
			\hspace*{-0.2cm}\includegraphics[width=0.6\textwidth,angle=-90]{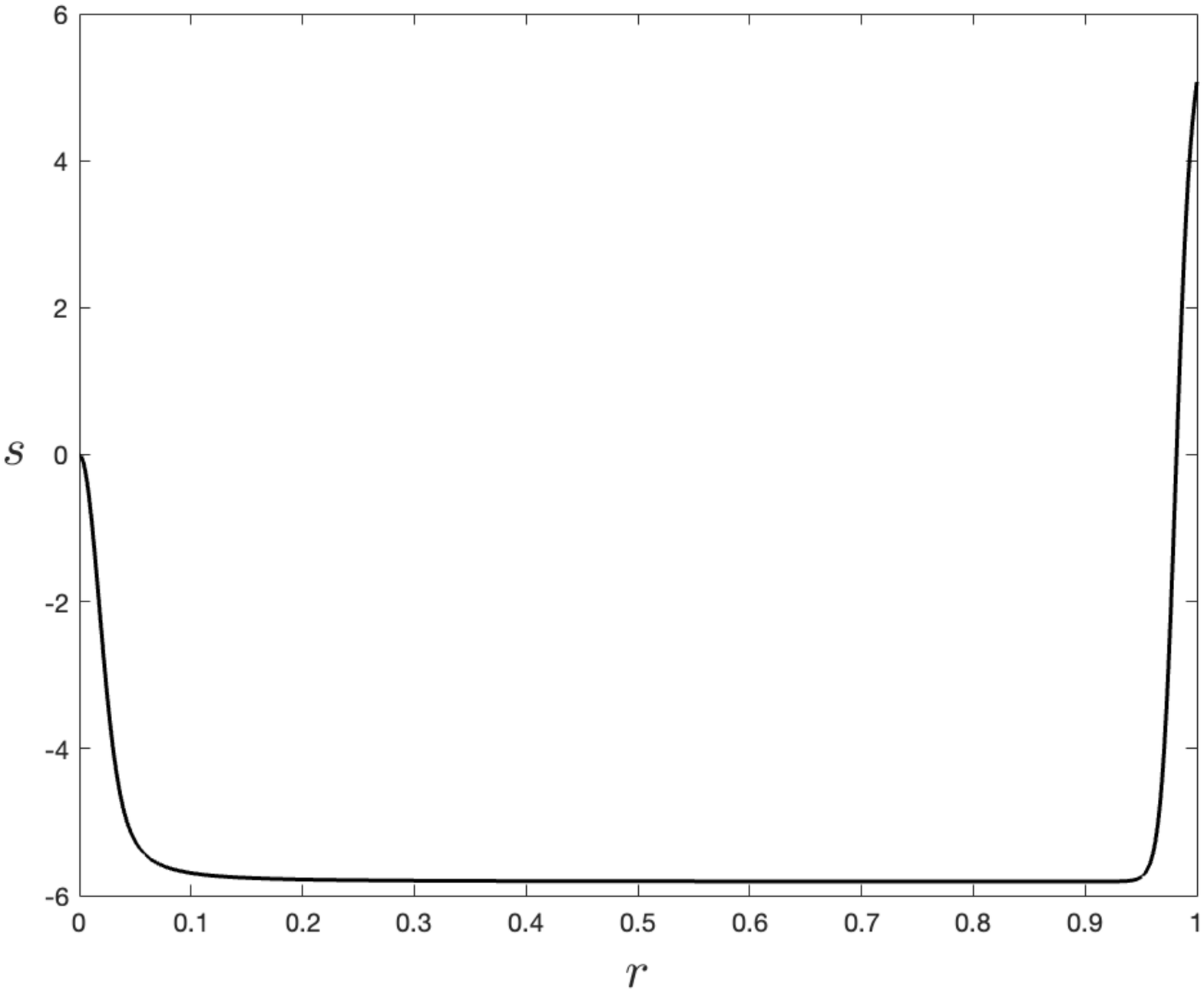}
			\subcaption{}\label{fig:-vehedgehog}
		\end{minipage}\hfill
		\caption{RH scalar order parameter profiles from initial guess (a) \(s(r) = 0, \, r \in [0,1]\), and (b) \(s(r) = 0.5s_-, \, r \in [0,1]\).}\label{fig:hedgehogprofiles}
	\end{figure}

\section{Numerical Results}
\label{sec:numerical}
	\subsection{Problem Formulation}
	The critical points of the LdG free energy with the fourth-order potential have been well-studied in a batch of papers \parencite{AlamaBronsard2016, RossoVirga1996, MkaddemGartland2000, SchopohlSluckin1987, Canevari2015}, and we do not claim to have exhaustive references. With the fourth-order potential in \eqref{eq:f1}, for large droplets and low temperatures, it is well-known that there exist at least two further critical points of the LdG energy in the admissible space \eqref{AQ} - the biaxial torus and the split core solutions, both of which have small biaxial regions near the origin and rotational symmetry, with mirror symmetry across the plane normal to the axis of rotational symmetry. 
The first question of interest is: do the biaxial torus and the split core critical points survive as critical points of the LdG energy with the sixth-order potential \eqref{eq:f2}? If they do survive as critical points, then their existence is naturally dictated by the symmetries of the geometry and the solution profiles, and less so by the precise form of the bulk potential. Following previous work in \parencite{MkaddemGartland2000} in this section, we numerically compute critical points of the LdG energy in \eqref{6thNDLdG} with rotational symmetry about the \(z\)-axis, and mirror symmetry across the \(xy\)-plane. In this case, our domain is reduced to a quarter circle rotated \(2\pi\) radians about the \(z\)-axis. We work in cylindrical polar coordinates \((r,\theta,z)\), where \(r \in [0,1]\) and \(z \in [0,1]\), while \(\theta\) is the angle in the \(xy\)-plane.
	
	Since the \textbf{Q}-tensor order parameter has five degrees of freedom in the most general setting, it is possible to represent the \textbf{Q}-tensor in terms of five basis tensors \parencite{SonnetKillian1995}. However, in this section, we assume that the \textbf{Q}-tensor always has an eigenvector in the direction \(\boldsymbol{e}_{\theta}\), normal to the \(rz\)-plane, so that the \(\textbf{Q}\)-tensor can be expressed in terms of just three basis tensors as
	\begin{equation}
		\textbf{Q}(r,z) = q_1(r,z)\textbf{E}_1 + q_2(r,z)\textbf{E}_2 + q_3(r,z)\textbf{E}_3, \label{Qrz}
	\end{equation}
	where
	\begin{equation} \nonumber
		\textbf{E}_1 = \frac{1}{\sqrt{6}}\begin{bmatrix}
			-1 & 0 & 0 \\
			0 & -1 & 0 \\
			0 & 0 & 2
		\end{bmatrix}, \quad \textbf{E}_2 = \frac{1}{\sqrt{2}}\begin{bmatrix}
			1 & 0 & 0 \\
			0 & -1 & 0 \\
			0 & 0 & 0
		\end{bmatrix}, \quad \textbf{E}_3 = \frac{1}{\sqrt{2}}\begin{bmatrix}
			0 & 0 & 1 \\
			0 & 0 & 0 \\
			1 & 0 & 0
		\end{bmatrix}.
	\end{equation}
	We can then rewrite the LdG energies \eqref{4thNDLdG} and \eqref{6thNDLdG} in terms of \(q_1, q_2, q_3\).
	First, the \textbf{Q}-tensor in (\ref{Qrz}) is transformed to cylindrical polar coordinates via the relations:
	\begin{equation} \nonumber
		\textbf{Q}(r,\theta, z) = \textbf{R}(\theta)\textbf{Q}(r,z)\textbf{R}(\theta)^T,
	\end{equation}
	where \(\textbf{R}(\theta)\) is the rotation matrix
	\begin{equation} \nonumber
		\textbf{R}(\theta)  = \begin{bmatrix}
			\cos\theta & -\sin\theta & 0 \\
			\sin\theta & \cos\theta & 0 \\
			0 & 0 & 1
		\end{bmatrix}.
	\end{equation}
	Then
	\begin{equation}
		\textbf{Q}(r,\theta,z) = \begin{bmatrix}
			-\dfrac{q_1}{\sqrt{6}} + \dfrac{q_2}{\sqrt{2}}\big(\cos^2\theta - \sin^2\theta\big) & \dfrac{2q_2}{\sqrt{2}}\cos\theta\sin\theta & \dfrac{q_3}{\sqrt{2}}\cos\theta \\
			\dfrac{2q_2}{\sqrt{2}}\cos\theta\sin\theta & -\dfrac{q_1}{\sqrt{6}} + \dfrac{q_2}{\sqrt{2}}\big(\sin^2\theta - \cos^2\theta\big) & \dfrac{q_3}{\sqrt{2}}\sin\theta \\
			\dfrac{q_3}{\sqrt{2}}\cos\theta & \dfrac{q_3}{\sqrt{2}}\sin\theta & \dfrac{2q_1}{\sqrt{6}}
		\end{bmatrix}. \label{4thQrtz}
	\end{equation}
	The LdG energy in (\ref{4thNDLdG}) is then given by:
	\begin{multline}
		\mathcal{F}_{four}[\textbf{Q}] = \int_{B(0,1)}\Bigg(\frac{t}{2}\big(q_1^2 + q_2^2 + q_3^2\big) - q_1^3 + 3q_1q_2^2 - \frac{3}{2}q_1q_3^2 - \frac{3\sqrt{3}}{2}q_2q_3^2 \\
		+ \frac{1}{2}\big(q_1^4 + q_2^4 + q_3^4 + 2q_1^2q_2^2 + 2q_1^2q_3^2 + 2q_2^2q_3^2\big) \\
		+ \frac{\varepsilon^2}{2}\bigg(q_{1,r}^2 + q_{2,r}^2 + q_{3,r}^2 + q_{1,z}^2 + q_{2,z}^2 + q_{3,z}^2 + \frac{1}{r^2}\big(4q_2^2 + q_3^2\big)\bigg)\Bigg)\,dV. \label{4thLdGqi}
	\end{multline}
	and the LdG energy with the sixth-order potential (\ref{6thNDLdG}) is given by:
	\begin{equation}
		\begin{aligned}
			\mathcal{F}_{six}[\textbf{Q}] = 	\int_{B(0,1)}\Bigg(&\frac{\varepsilon^2}{2}\bigg(q_{1,r}^2 + q_{2,r}^2 + q_{3,r}^2 + q_{1,z}^2 + q_{2,z}^2 + q_{3,z}^2 + \frac{1}{r^2}\big(4q_2^2 + q_3^2\big)\bigg) \\
			&+ \frac{t}{2}\big(q_1^2 + q_2^2 + q_3^2\big) - q_1^3 + 3q_1q_2^2 - \frac{3}{2}q_1q_3^2 - \frac{3\sqrt{3}}{2}q_2q_3^2 \\
			&+ \frac{1}{2}\big(q_1^4 + q_2^4 + q_3^4 + 2q_1^2q_2^2 + 2q_1^2q_3^2 + 2q_2^2q_3^2\big) \\
			&+ \begin{aligned}[t]
				\frac{d}{5}\bigg(\frac{\sqrt{6}}{6}q_1^5 &- \frac{\sqrt{6}}{3}q_1^3q_2^2 + \frac{5\sqrt{6}}{12}q_1^3q_3^2 - \frac{\sqrt{6}}{2}q_1q_2^4 - \frac{\sqrt{6}}{4}q_1q_2^2q_3^2 \\
                &+ \frac{\sqrt{6}}{4}q_1q_3^4 + \frac{3\sqrt{2}}{4}q_1^2q_2q_3^2 + \frac{3\sqrt{2}}{4}q_2^3q_3^2 + \frac{3\sqrt{2}}{4}q_2q_3^4\bigg) 
			\end{aligned} \\
			&+ \begin{aligned}[t]
				\frac{e}{6}\big(q_1^6 + q_2^6 + q_3^6 &+ 3q_1^4q_2^2 + 3q_1^4q_3^2 \\
				&+ 3q_1^2q_2^4 + 3q_1^2q_3^4 + 3q_2^4q_3^2 + 3q_2^2q_3^4 + 6q_1^2q_2^2q_3^2\big) 
			\end{aligned} \\
			&+ \begin{aligned}[t]
				\frac{(f - e)}{6}\bigg(\frac{1}{6}q_1^6 &- q_1^4q_2^2 + \frac{1}{2}q_1^4q_3^2 + \frac{\sqrt{3}}{2}q_1^3q_2q_3^2 + \frac{3}{2}q_1^2q_2^4 \\
                &- \frac{3}{2}q_1^2q_2^2q_3^2 - \frac{3\sqrt{3}}{2}q_1q_2^3q_3^2 + \frac{3}{8}q_1^2q_3^4 + \frac{3\sqrt{3}}{4}q_1q_2q_3^4 + \frac{9}{8}q_2^2q_3^4\bigg)\Bigg)\,dV.
			\end{aligned}
		\end{aligned} \label{6thLdGqi}
	\end{equation}
	
	The last step is to specify the boundary conditions for $q_1, q_2, q_3$ with these symmetry assumptions. 	We work with the Dirichlet boundary condition in (\ref{DirichletBC}) on $r^2 + z^2 =1$, which can be translated into conditions for \(q_1, q_2, q_3\). The unit vector \(\hat{\boldsymbol{r}}\) can be written as $
		\hat{\boldsymbol{r}} = r\boldsymbol{e}_r + z\boldsymbol{e}_z, \,\, r^2 + z^2 = 1,
	$
	so that \eqref{DirichletBC} can be written as 
	\begin{equation} \nonumber
		s_+\bigg(\hat{\boldsymbol{r}}\otimes\hat{\boldsymbol{r}} - \frac{1}{3}\textbf{I}\bigg) = s_+\begin{bmatrix}			
			r^2\cos^2\theta - \dfrac{1}{3} & r^2\cos\theta\sin\theta & rz\cos\theta \\
			r^2\cos\theta\sin\theta & r^2\sin^2\theta - \dfrac{1}{3} & rz\sin\theta \\
			rz\cos\theta & rz\sin\theta & z^2 - \dfrac{1}{3}
		\end{bmatrix}.
	\end{equation}
 Comparing with \eqref{4thQrtz}, we obtain
	\begin{equation}
		q_1 = \sqrt{\frac{2}{3}}\bigg(1 - \frac{3r^2}{2}\bigg)s_+, \quad q_2 = \frac{r^2}{\sqrt{2}}s_+, \quad q_3 = \sqrt{2}rzs_+ \quad \text{on} \,\, r^2 + z^2 = 1. \label{radialBCs}
	\end{equation}
	There are additional boundary conditions to account for the assumed  rotational and mirror symmetry:
	\begin{equation}
		q_{1,z} = q_{2,z} = q_3 = 0 \quad \text{on} \,\, z = 0 \label{symBC1}
	\end{equation}
	for mirror symmetry across the \(xy\)-plane, and
	\begin{equation}
		q_{1,r} = q_2 = q_{2,r} = q_3 = 0 \quad \text{on} \,\, r = 0 \label{symBC2}
	\end{equation}
	for rotational symmetry about the \(z\)-axis.
	
	\subsection{Stationary Points of the LdG Energy}
 \label{sec:LdGstationary}
	We use a finite element method to solve for stationary/critical points of the weak formulations associated with the LdG free energy (\ref{4thLdGqi}) with fourth-order potential and (\ref{6thLdGqi}) with the sixth-order potential respectively. The finite element method is implemented in the open-source computing package FEniCS \parencite{Fenics} and the visualisation is carried out in an open source post-processing visualisation application, ParaView \parencite{ParaView}.
	
	We plot the biaxiality parameter of the numerically computed critical points, since biaxiality often labels defects and biaxiality also distinguishes the sixth-order potential from the fourth-order potential.
	\begin{equation}
		\beta = 1 - 6\frac{\big(\tr\textbf{Q}^3\big)^2}{\big(\tr\textbf{Q}^2\big)^3}, \label{beta}
	\end{equation}
	where \(0 \leq \beta \leq 1\), with \(\beta = 0\) corresponding to uniaxiality and \(\beta = 1\) corresponding to `maximal' biaxiality \parencite{KaiserWieseHess1992}. The radial hedgehog solution is purely uniaxial with $\beta = 0$ everywhere, whereas the split core and biaxial torus solutions have signature regions of biaxiality near the origin. We also plot the leading eigenvector of the \textbf{Q}-tensor in the examples below, which is the eigenvector with the largest positive eigenvalue, regarded as the \emph{nematic director}. A further good marker is the sign of the scalar order parameter at the origin. However, since we set \(q_2 = q_3 = 0\) at the origin in \eqref{symBC2}, this reduces to the sign of \(q_1\) at the origin. The radial hedgehog solution is isotropic at the origin i.e. \(q_1 (r = 0) = 0\), while the split core is negatively ordered at the origin, requiring \(q_1 (r = 0) < 0\), and the biaxial torus is positively ordered at the origin, requiring \(q_1 (r = 0) > 0\).
	
	With the fourth-order potential, there are known results in the literature \parencite{Majumdar2012,MkaddemGartland2000, SonnetKillian1995} that demonstrate the stability of the RH solution for high temperatures and small droplets; and at least local stability of the split core and biaxial torus solutions at lower temperatures and in droplets of larger radius. In Figure \ref{fig:4thconfigs}, we plot the biaxiality parameter, \(\beta\), and the leading eigenvector of the radial hedgehog, split core, and biaxial torus configurations obtained with the fourth-order potential in Figure \ref{fig:4thconfigs}. We plot the RH configuration in Figure \ref{fig:4thrh} with \(t = 0\) and \(\varepsilon = 1\). The RH solution has $\beta =0$ everywhere, with an isotropic point at $r=0$ (with $q_1 =0$ at $r=0$) and the leading eigenvector is simply the radial unit vector. We plot the split core and biaxial torus solutions for \(t = -10, \varepsilon = 0.5\) in Figures \ref{fig:4thsc} and \ref{fig:4thbt}, respectively. We observe the signature regions of biaxiality associated with the split core and biaxial torus solutions, labelled by the red regions, along with \(q_1(r=0) < 0\) for the split core solution, while \(q_1(r=0) > 0\) for the biaxial torus solution respectively. We compute the Morse index of each configuration in Figure \ref{fig:4thconfigs}, and find that each is  a locally stable critical point of the LdG free energy with the fourth-order potential \eqref{4thLdGqi}, for the specified values of \(t\) and \(\varepsilon\). Local stability of a LdG critical point implies that it is potentially observable in experiments and applications.
	
	\begin{figure}[!ht]
            \hfill
		\begin{minipage}{0.25\textwidth}
			\centering
			\includegraphics[width=0.7\textwidth, angle=-90]{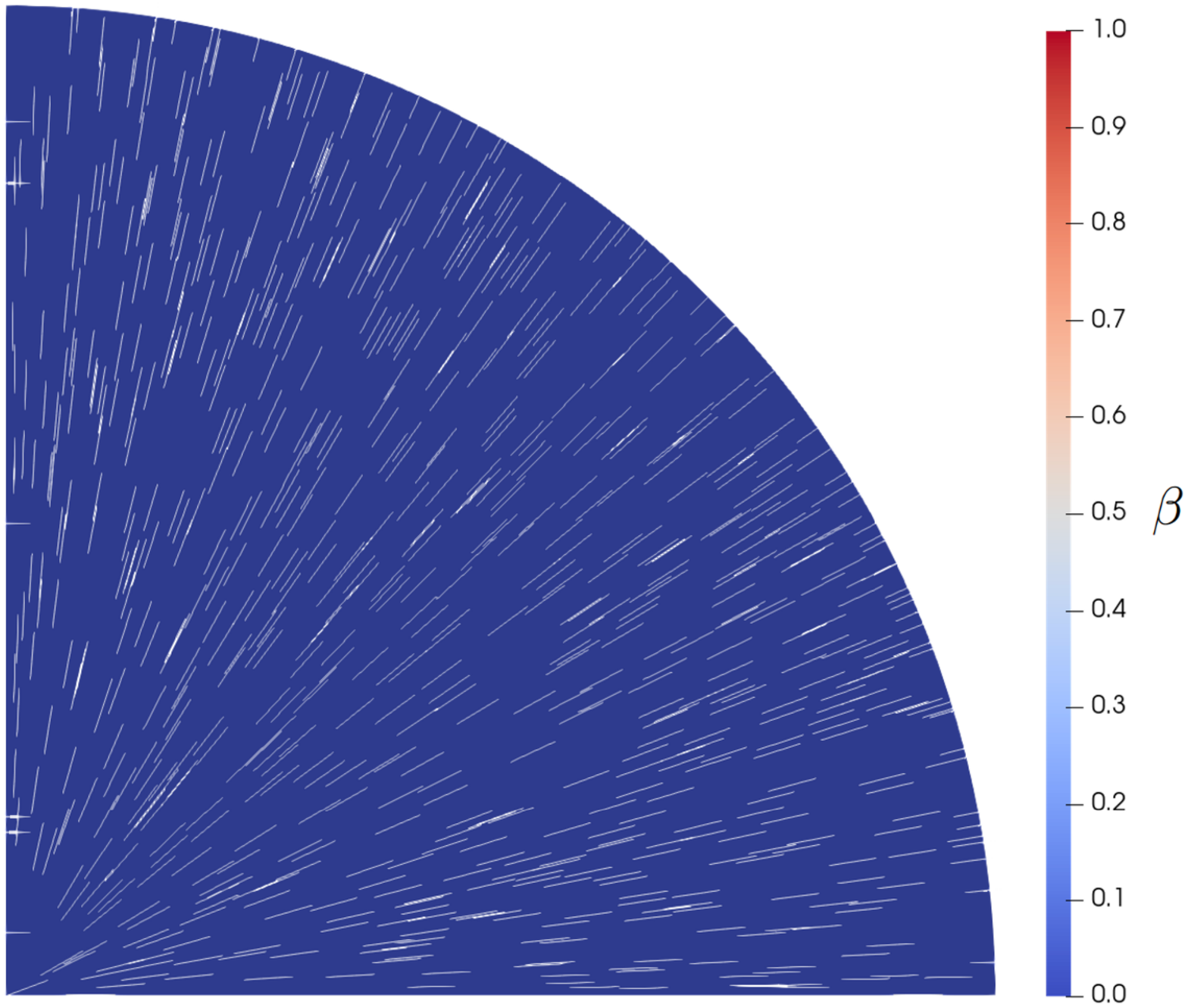}
			\subcaption{}\label{fig:4thrh}
		\end{minipage}\hfill
		\begin{minipage}{0.25\textwidth}
			\centering
			\includegraphics[width=0.7\textwidth,angle=-90]{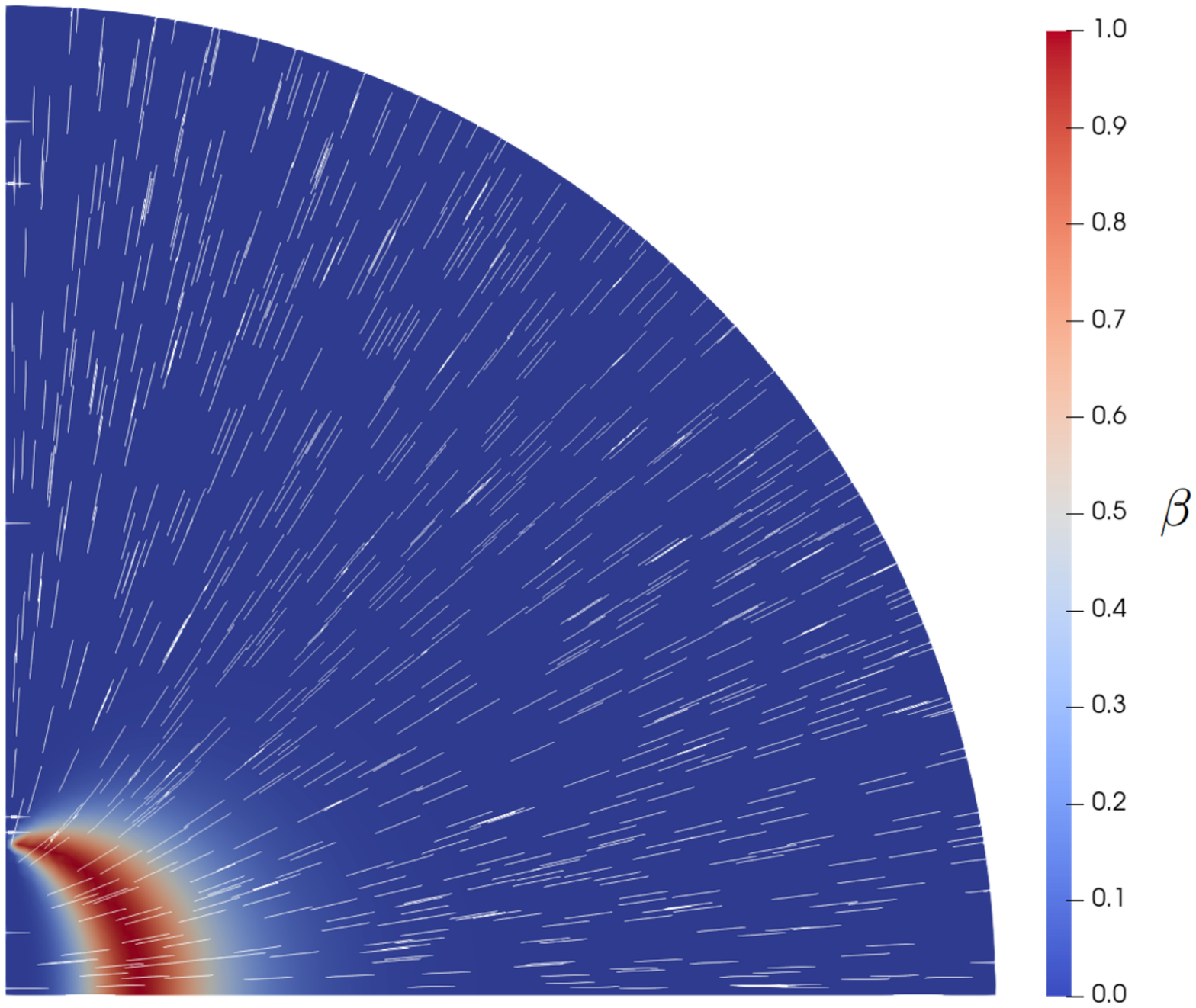}
			\subcaption{}\label{fig:4thsc}
		\end{minipage}\hfill
		\begin{minipage}{0.25\textwidth}
			\centering
			\includegraphics[width=0.7\textwidth,angle=-90]{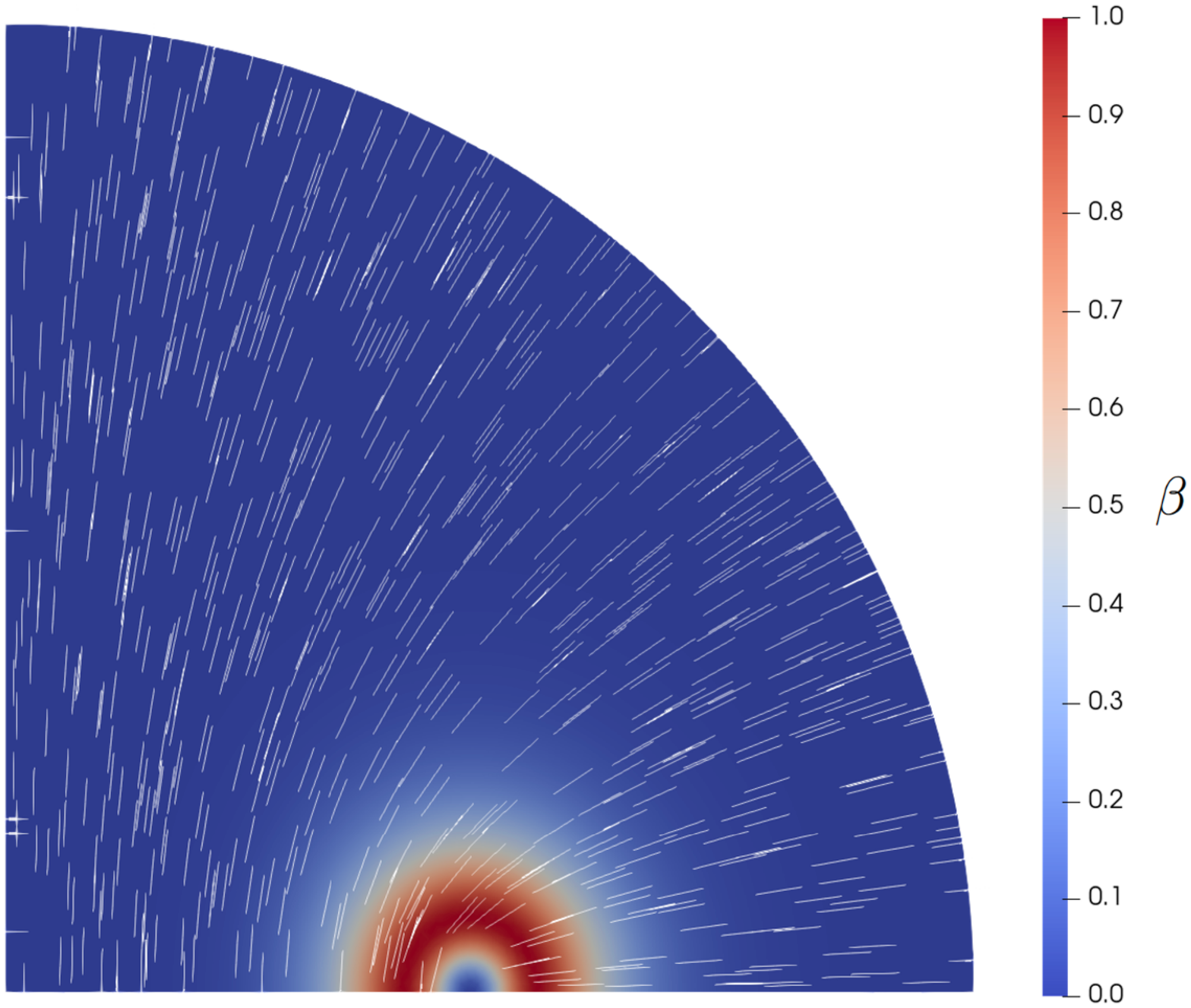}
			\subcaption{}\label{fig:4thbt}
		\end{minipage}
            \hfill
		\caption{Biaxiality parameter, \(\beta\), and leading eigenvector of LdG stationary points obtained with the fourth-order potential. (a) RH solution with \(t = 0, \varepsilon = 1\). (b) Split core solution with \(t = -10, \varepsilon = 0.5\). (c) Biaxial torus solution with \(t = -10, \varepsilon = 0.5\).}\label{fig:4thconfigs}
	\end{figure}
	
	Next, we repeat the same numerical investigations with the LdG energy with sixth-order potential, which has not been attempted in the literature to date. The behavior and trends are expected to be similar to those observed with the fourth-order potential, at least for moderately low temperatures, as suggested by the analysis in the previous section. In Figure \ref{fig:6thconfigs}, we plot stationary points of \eqref{6thNDLdG} with \(d = 1, e = 0, f = 1\), and correspond to the RH configuration at \(t = 0\) and \(\varepsilon = 1\), and the split core and biaxial torus solutions for \(t = -10\) and \(\varepsilon = 0.5\). We plot \(\beta\) and the leading eigenvector of the \textbf{Q}-tensor in each case. We find that \(q_1\) is approximately zero at the origin for the RH solution; negative at the origin for the split core solution; and positive at the origin for the biaxial torus solution. Comparing Figures \ref{fig:6thsc} and \ref{fig:6thbt} obtained with a sixth-order potential with Figures \ref{fig:4thsc} and \ref{fig:4thbt}, respectively, obtained at the same values of \(t\) and \(\varepsilon\) with the fourth-order potential, we observe that the regions of biaxiality of the split core and biaxial torus solutions are larger with the sixth-order potential. We again compute the smallest real eigenvalue of the Hessian associated with the LdG free energy \eqref{6thLdGqi} and each numerically computed stationary point is locally stable with the sixth-order potential as well.
	
	\begin{figure}[!ht]
            \hfill
		\begin{minipage}{0.25\textwidth}
			\centering
			\includegraphics[width=0.7\textwidth,angle=-90]{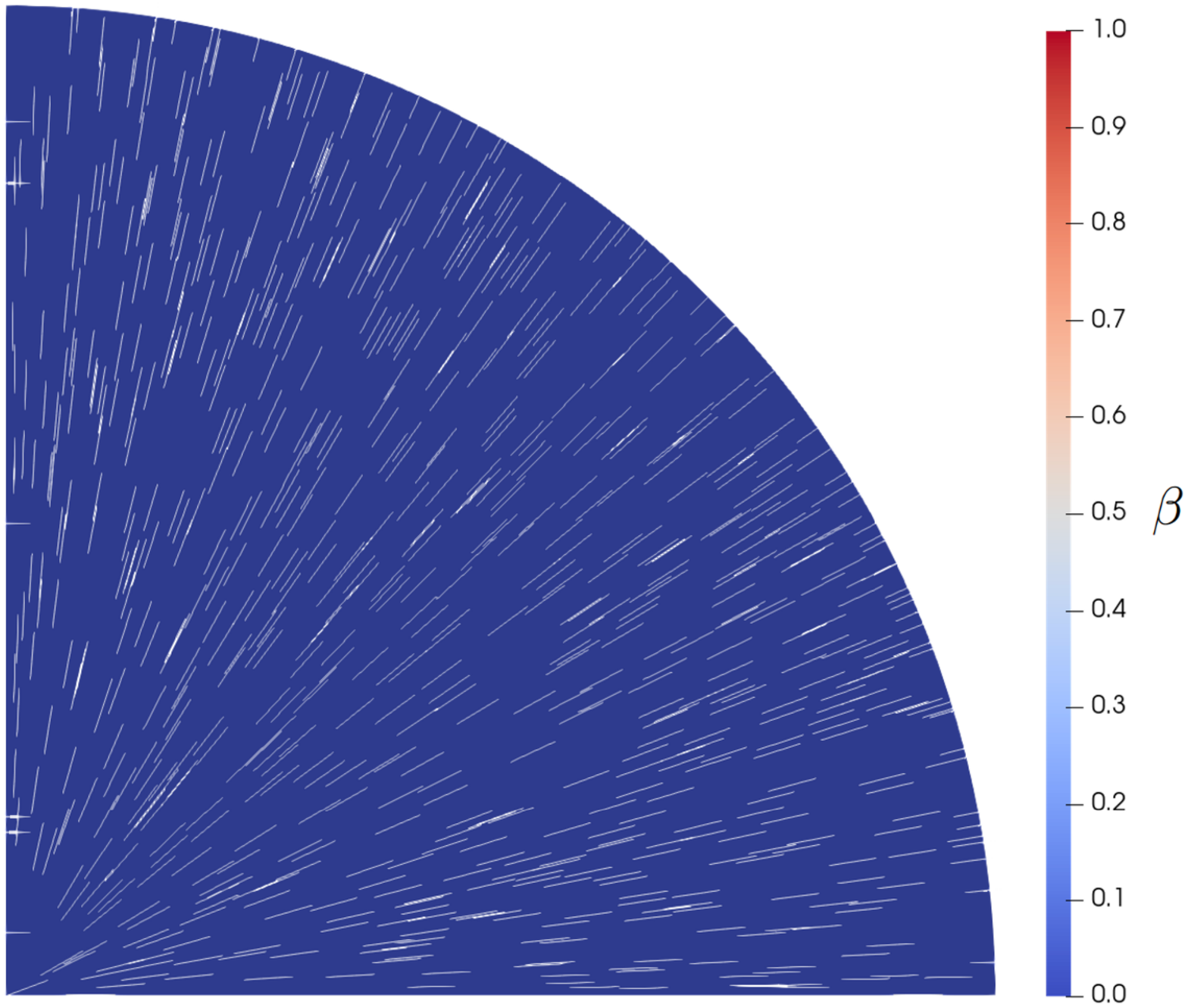}
			\subcaption{}\label{fig:6thrh}
		\end{minipage}\hfill
		\begin{minipage}{0.25\textwidth}
			\centering
			\includegraphics[width=0.7\textwidth,angle=-90]{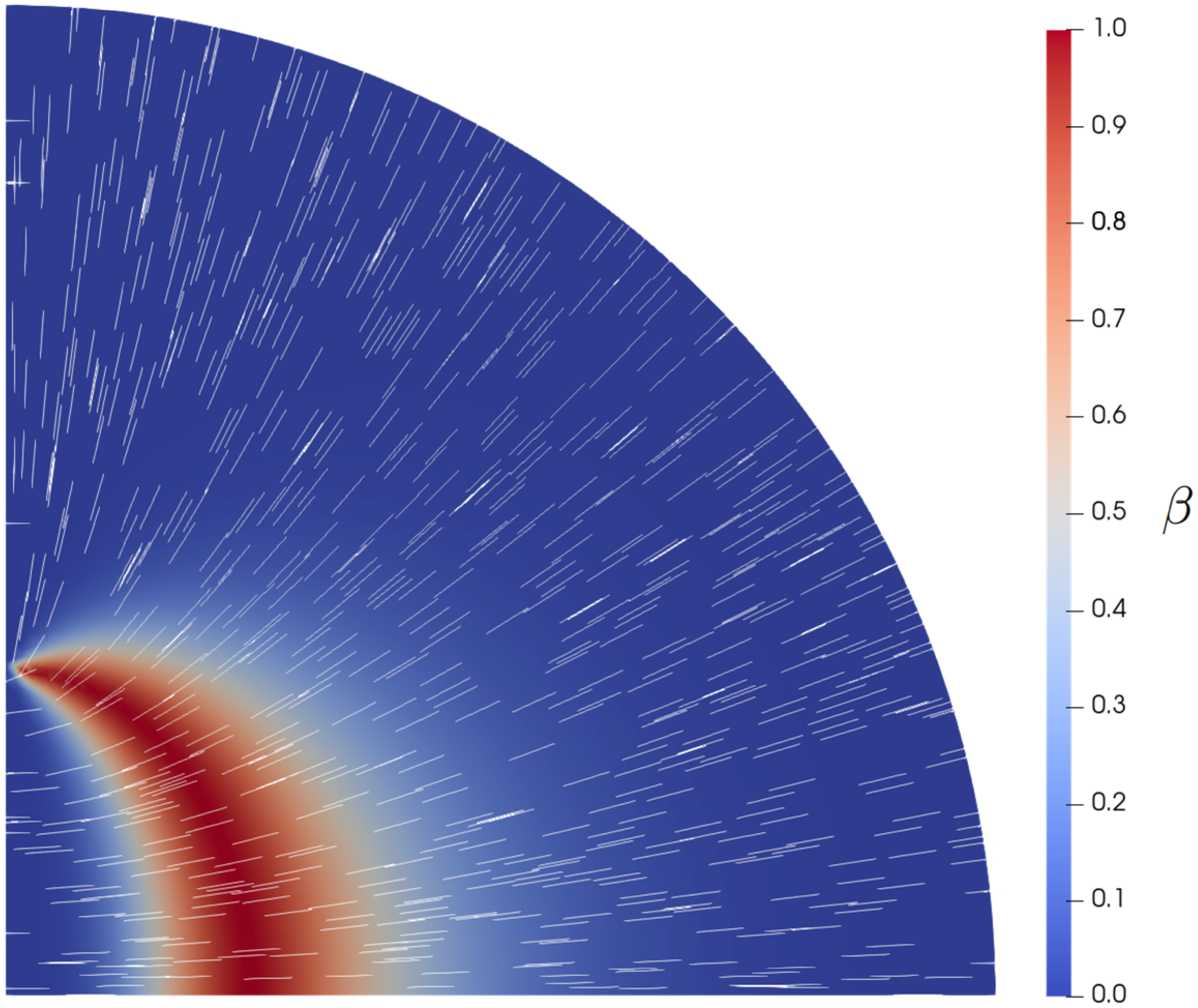}
			\subcaption{}\label{fig:6thsc}
		\end{minipage}\hfill
		\begin{minipage}{0.25\textwidth}
			\centering
			\includegraphics[width=0.7\textwidth,angle=-90]{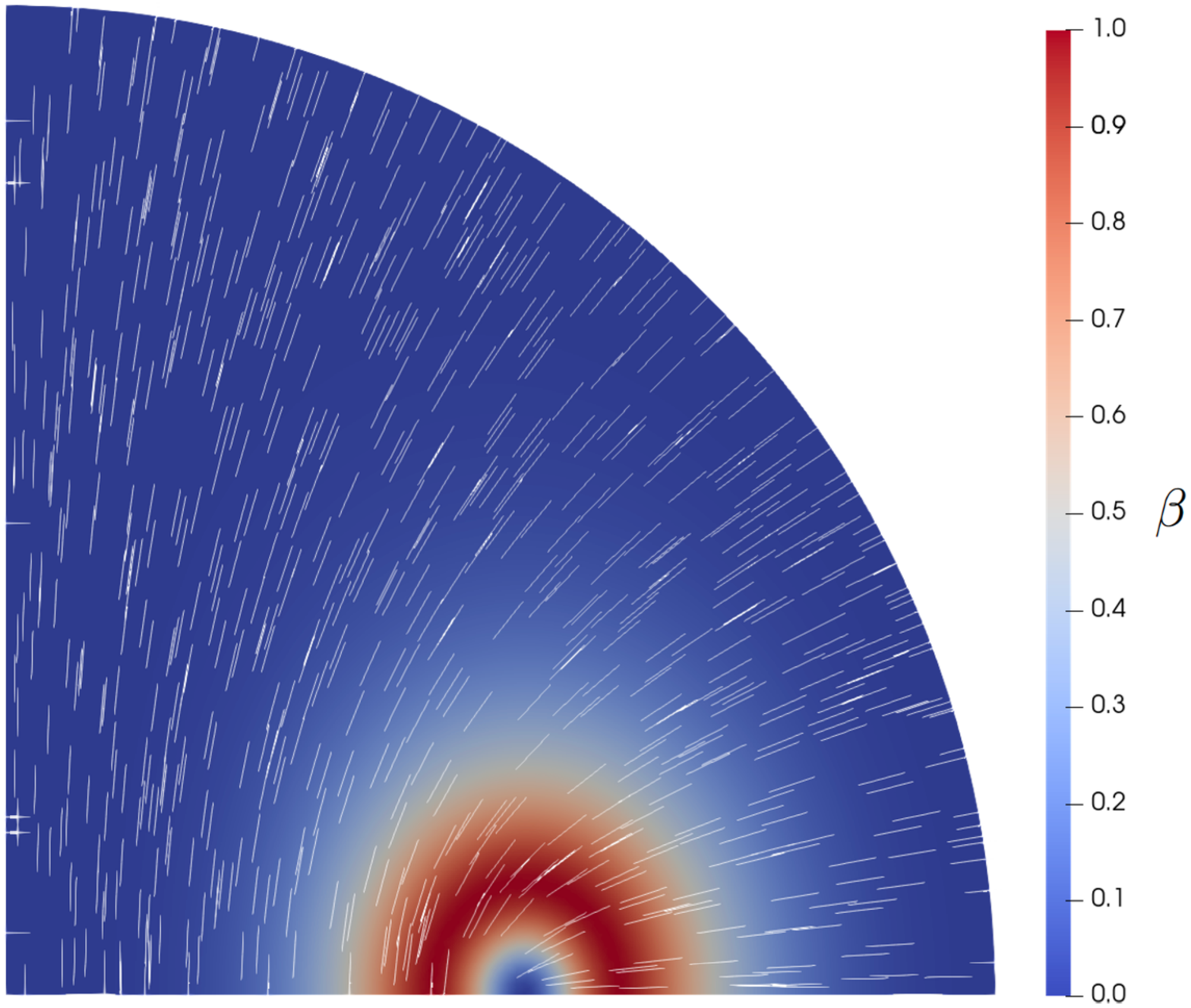}
			\subcaption{}\label{fig:6thbt}
		\end{minipage}
            \hfill
		\caption{Biaxiality parameter, \(\beta\), and leading eigenvector of LdG stationary points of \eqref{6thNDLdG} with \(d = 1, e = 0, f = 1\). (a) RH solution with \(t = 0, \varepsilon = 1\). (b) Split core solution with \(t = -10, \varepsilon = 0.5\). (c) Biaxial torus solution with \(t = -10, \varepsilon = 0.5\).}\label{fig:6thconfigs}
	\end{figure}
	
	\subsection{The Morse Index of the Radial Hedgehog Solution}
	We characterise the stability of the LdG stationary points using the Morse index, which is the number of negative real eigenvalues of their associated Hessian \parencite{MilnorSpivak1963}. The Morse index is calculated using the SLEPc eigenvalue solver \parencite{SLEPc}. An index-0 critical point, with no negative eigenvalues, is at least locally stable, while all index-\(k\) critical points, with \(k > 0\), are unstable.
	We numerically compute the Morse index of the RH solution, for a range of temperatures and droplet radii, to study the effects of temperature and droplet size on the stability of the RH solution, although our study is only restricted to the class of $\mathbf{Q}$-tensors with three degrees of freedom. We perform a parallel study of the RH solution as a critical point of the LdG energy with the fourth-order \eqref{4thNDLdG} and the sixth-order \eqref{6thNDLdG} potentials.
 In Tables \ref{fig:4thIndexTable} and \ref{fig:6thIndexTable}, we tabulate the Morse index of the RH solution with the fourth- and sixth-order potentials respectively, with \(d = 1, e = 0, f = 1\), for the given values of \(t\) and \(\varepsilon\). For each entry, the RH solution is in terms of $q^* = (q_1,q_2,q_3)$ where 
\begin{equation}\label{eq:RHq*}
		q_1 = \sqrt{\frac{2}{3}}\bigg(1 - \frac{3r^2}{2}\bigg)s^*, \quad q_2 = \frac{r^2}{\sqrt{2}}s^*, \quad q_3 = \sqrt{2}rzs^*. 
	\end{equation}
 The function \(s^*\) is the solution of the RH ODE (see \eqref{sODE} for the sixth-order potential).

	\begin{figure}[!ht]
		\centering
		\includegraphics[width=0.8\textwidth]{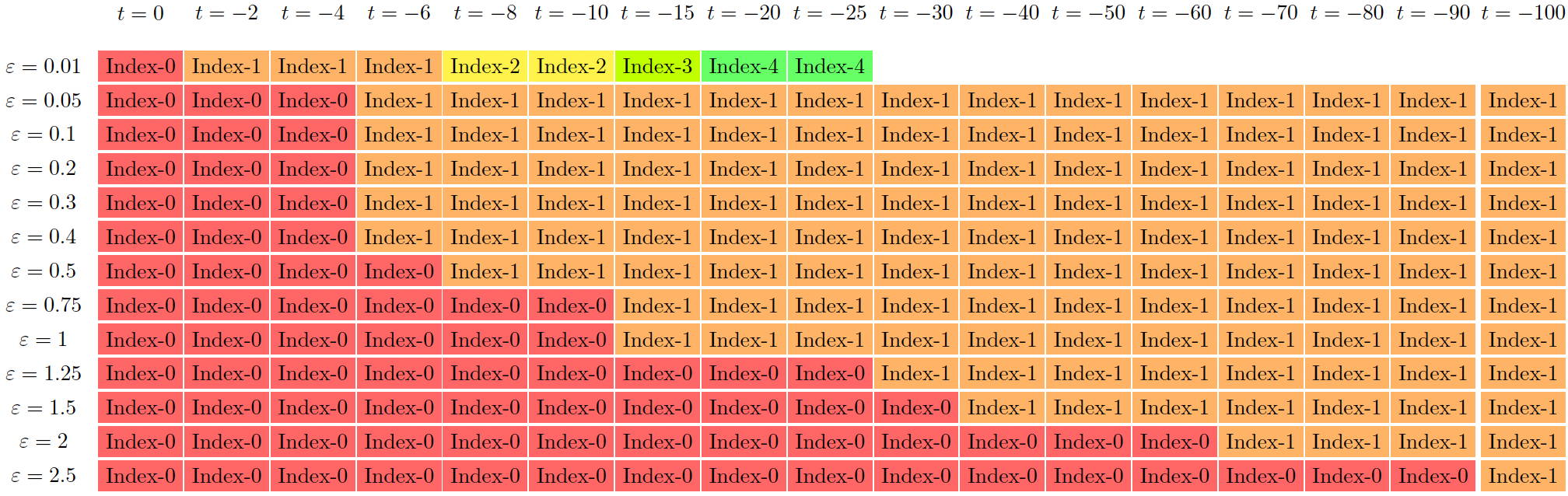}
		\caption{Morse index of the RH solution for the given value of \(t\) and \(\varepsilon\), with the fourth-order potential \eqref{4thLdGqi}.}
		\label{fig:4thIndexTable}
	\end{figure}
	
	\begin{figure}[ht!]
		\centering
		\includegraphics[width = 0.8\textwidth]{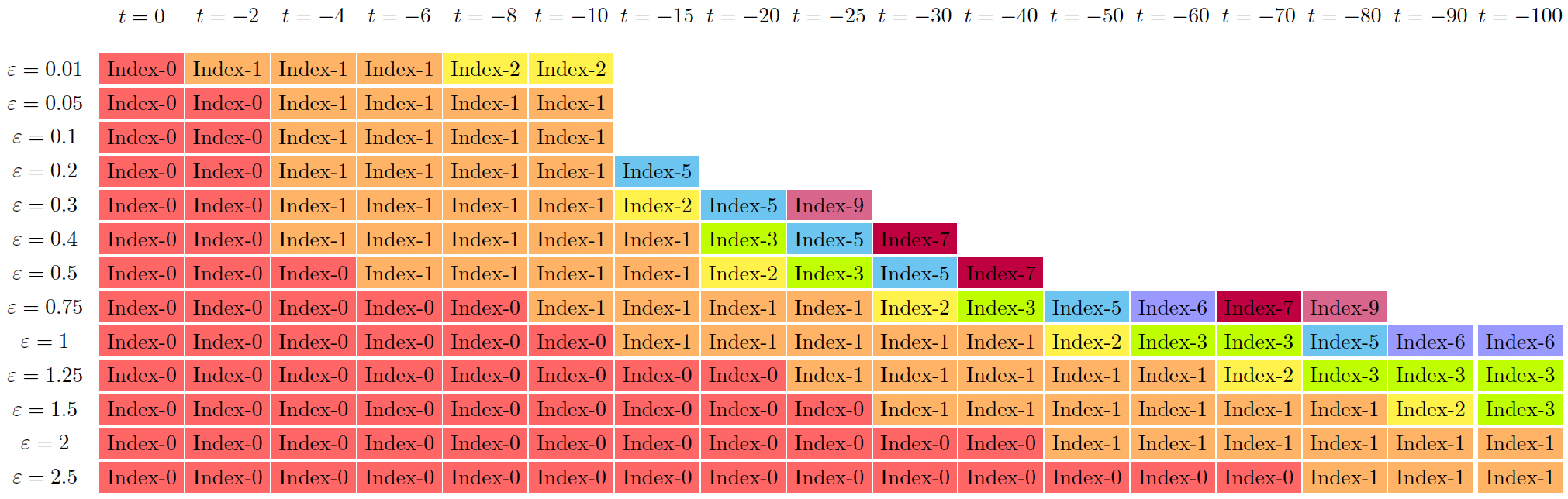}
		\caption{Morse index of the RH solution for the given value of \(t\) and \(\varepsilon\) with the sixth-order potential \eqref{6thLdGqi}, with \(d = 1\), \(e = 0\), \(f = 1\).} \label{fig:6thIndexTable}
	\end{figure}
	
	There are some generic trends - the index of the RH solution is lower for higher values of \(t\) and \(\varepsilon\) in both cases. This is consistent with the fact that the RH solution is stable closer to the isotropic-nematic transition temperatures and for smaller droplets. 
 Comparing the indices with the sixth- and fourth-order potentials, the RH solution has higher index in the sixth-order case compared to the fourth-order case. 
 This suggests that the sixth-order potential has a destabilising effect on the RH solution, which could be explained on the grounds that the sixth-order potential admits biaxial critical points, so there are more unstable biaxial eigendirections for the RH solution, resulting in a higher Morse index compared to the fourth-order potential which does not admit biaxial critical points.
	We compute at most the ten smallest eigenvalues. Blank spaces in the tables correspond to cases for which all ten computed eigenvalues are negative or when the solver fails to compute the ten smallest eigenvalues, but the numerically computed eigenvalues are negative.
	
	\subsection{The RH Solution as an Index-1 Transition State}
 \label{sec:transitionRH}
	
	In this section, we attempt to identify situations for which the RH solution acts as an index-$1$ saddle point, because index-$1$ saddle points are often referred to as \emph{transition states}, relevant for switching between two locally stable states \cite{KusumaatmajaMajumdar2015}. In other words, the transition state mediates the transition and may be observable in the non-equilibrium dynamics.
	
	We work with values of $t$ and $\varepsilon$ for which the RH solution is an index-$1$ critical point of the LdG energy \eqref{4thLdGqi} and \eqref{6thLdGqi}. Our aim is to compute the transition pathway between two index-$0$ LdG stationary points, through an index-$1$ RH solution $\boldsymbol{q}^*$ in \eqref{eq:RHq*}. Using a gradient flow method and taking small perturbations of the RH solution along the direction of the eigenvector associated with the negative eigenvalue of the Hessian as an initial condition, we solve the initial value problem
	\begin{equation}
		\frac{\partial\boldsymbol{q}}{\partial\tau} = -\nabla\mathcal{F}(\boldsymbol{q},\nabla\boldsymbol{q}) \quad \text{in} \,\, \Omega \,\, \text{for} \,\, \tau > 0, \label{GFIVP1}
	\end{equation}
	\begin{equation}
		\boldsymbol{q} = \boldsymbol{q}_0 = \boldsymbol{q}^* \pm \lambda\boldsymbol{u} \quad \text{in} \,\, \Omega \,\, \text{at} \,\, \tau = 0, \label{GFIVP2}
	\end{equation}
	with boundary conditions (\ref{radialBCs})-(\ref{symBC2}), where \(\Omega\) is the quarter circle domain; \(\mathcal{F}\) is the LdG energy with the fourth- or sixth-order potential, given by (\ref{4thLdGqi}) or (\ref{6thLdGqi}), respectively; \(s_+\) is the scalar order parameter of the global minimiser of the fourth- or sixth-order bulk potential in the class of uniaxial \textbf{Q}-tensors with positive scalar order parameter; the quantity \(\lambda\) is a small positive constant; and \(\boldsymbol{u}\) is the unstable eigendirection of the RH solution. Note that (\ref{GFIVP2}) describes two different initial values to compute two distinct index-$0$ LdG stationary points in the class of $\mathbf{Q}$-tensors with mirror and rotational symmetry.
 We discretise the \(\tau\)-dependent PDEs using an implicit Euler method. 
	
	As an example, the unstable eigendirection of the RH solution as a stationary point of \eqref{6thLdGqi} with \(t = -12, \varepsilon = 0.5, d = 1, e = 0, f = 1\) is 
		$\boldsymbol{u} = (6.11 \times 10^{-2}, 3.34 \times 10^{-18}, 1.15 \times 10^{-18})^T$,
	at the origin. A perturbation of the RH solution, \(\boldsymbol{q}_0 = \boldsymbol{q} + \lambda\boldsymbol{u}\), yields the biaxial torus configuration, while a perturbation \(\boldsymbol{q}_0 = \boldsymbol{q} - \lambda\boldsymbol{u}\) yields the split core solution. This is in agreement with the fact that \(q_1 > 0\) at the origin for the biaxial torus and \(q_1 < 0\) at the origin for split core solutions.
	
	Figures \ref{fig:4thgradientflow} and \ref{fig:6thgradientflow} show examples of two index-0 stationary points of the LdG free energy with the fourth- \eqref{4thLdGqi} and sixth-order potentials \eqref{6thLdGqi} respectively, for \(t = -12\), \(\varepsilon = 0.5\), with \(d = 1, e = 0, f = 1\) in the sixth-order case, via a gradient flow method with the perturbed index-1 RH solution as initial condition. This strongly suggests that there are transition pathways via our index-1 RH solutions between the index-$0$ biaxial torus and split core solutions at \(t = -12\) and \(\varepsilon = 0.5\) in both cases. Note that the split core solution may not be index-$0$ in the full class of admissible $\mathbf{Q}$-tensors without the symmetry constraints. Nevertheless, we speculate that these reduced examples can be generalized to show that the RH solution can act as a transition state between two index-$0$ LdG stationary points in the admissible class \eqref{AQ}, without the symmetry constraints and exploiting the full five degrees of freedom.
	
	\begin{figure}[!ht]
		\begin{minipage}{0.49\textwidth}
			\hspace*{-0.3cm}\includegraphics[width=0.7\textwidth,angle=-90]{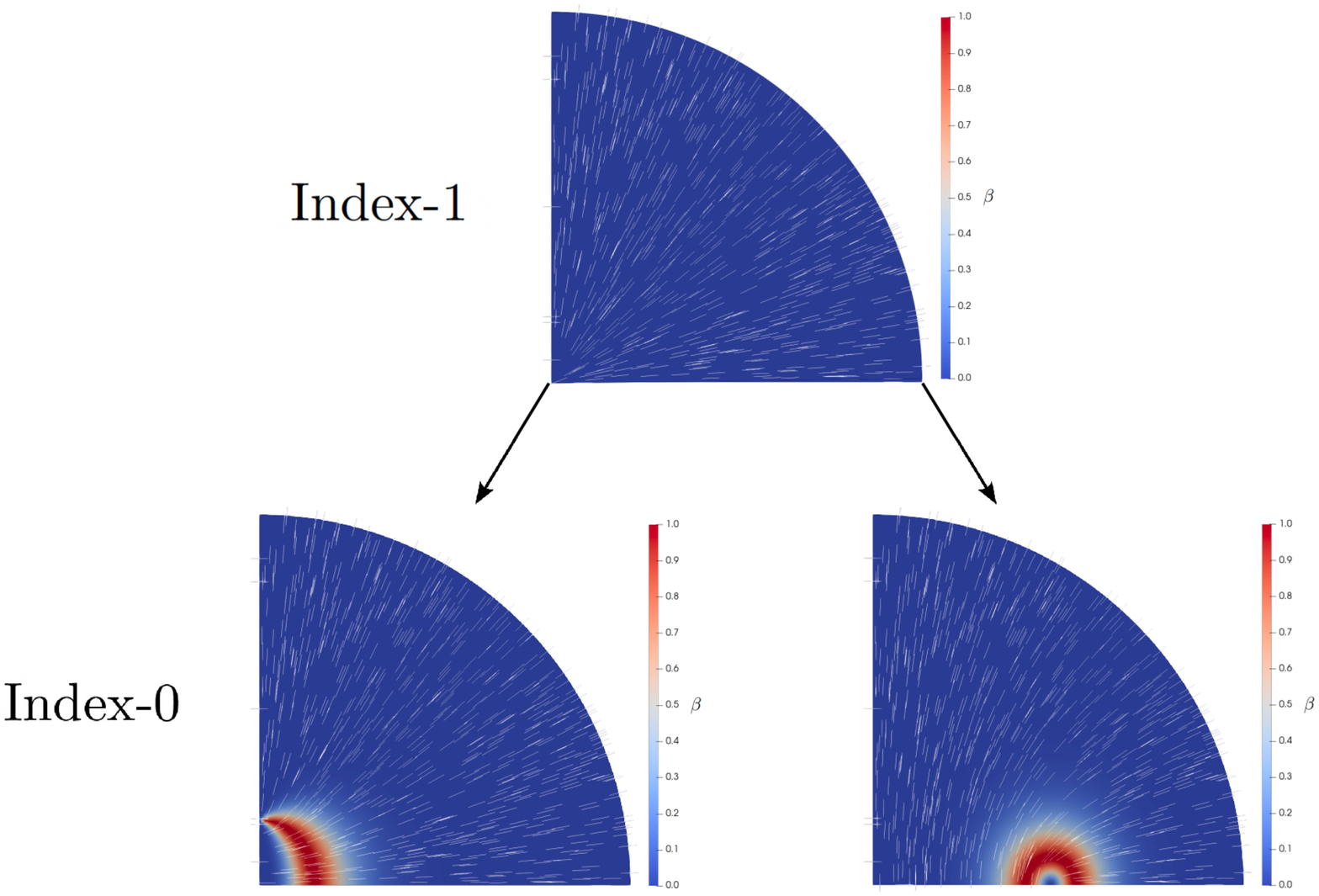}
                \subcaption{}\label{fig:4thgradientflow}
		\end{minipage}\hfill
		\begin{minipage}{0.49\textwidth}
			\hspace*{-0.3cm}\includegraphics[width=0.7\textwidth,angle=-90]{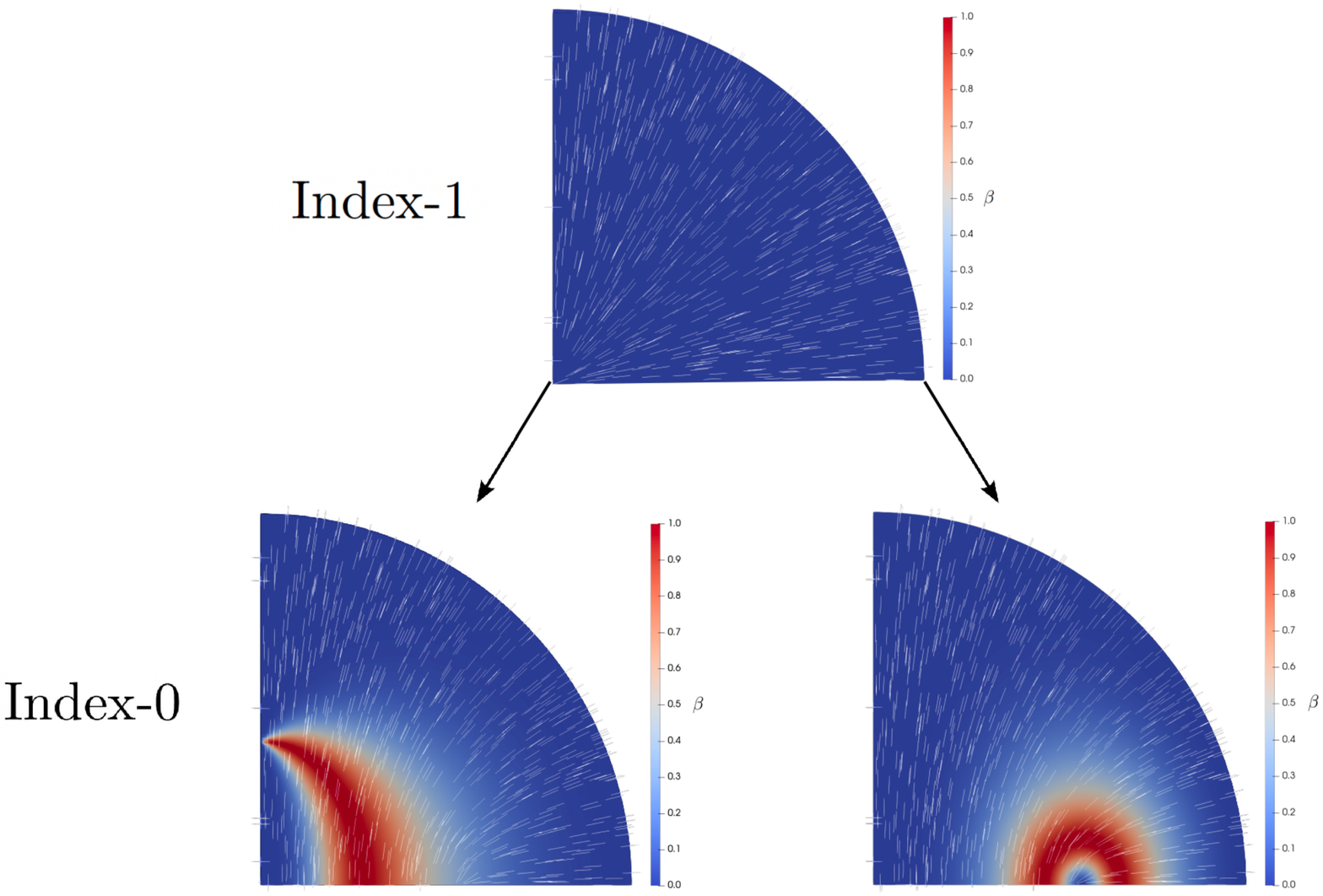}
			\subcaption{}\label{fig:6thgradientflow}
		\end{minipage}\hfill
		\caption{The transition pathways between two stable states. Split core and biaxial torus via index-1 transition state RH with (a) the fourth-order potential \eqref{4thLdGqi} at $t = -12$ and $\varepsilon = 0.5$, and (b) the sixth-order potential \eqref{6thLdGqi} with $t = -12$, $\varepsilon = 0.5$, $d = 1$, $e = 0$, $f = 1$.}

	\end{figure}

	\subsection{Bifurcation Diagrams}
 \label{sec:bifurcation}
	We numerically compute bifurcation diagrams with the LdG free energies \eqref{4thLdGqi} and \eqref{6thLdGqi}, in Figures \ref{fig:4thBifurcation} and \ref{fig:6thBifurcation} respectively.
 The value of the scalar order parameter, $s$, of each configuration at the origin, is plotted against temperature, noting that \(s(0) = \sqrt{\frac{3}{2}}q_1(0)\) and all configurations are uniaxial at the origin due to the boundary conditions (\ref{symBC1}) and (\ref{symBC2}). In what follows, we only consider RH solutions with positive order parameter profile, recalling that the global minimiser of \eqref{LdGs} can be negative for low temperatures. The two bifurcation diagrams are qualitatively similar and the bifurcation points are simply shifted: the RH solution, with \(s(0) = 0\), is the unique stationary point for high temperatures; the RH and biaxial torus configurations are stable at intermediate temperatures, where we also observe an unstable biaxial torus configuration. The RH configuration loses stability at low temperatures, while the globally minimising biaxial torus configuration remains stable, accompanied by the emergence of a locally stable split core configuration. Unsurprisingly, the RH solution loses stability at a higher critical temperature in the sixth-order case, compared to the fourth-order case, so that the RH solution is unstable over a wider temperature range with the sixth-order potential. We plot the stable and unstable biaxial torus configurations with the fourth-order potential at \(t = -6.5\) in Figure \ref{fig:4thstabletorus} and \ref{fig:4thunstabletorus}. We use a high-index optimisation-based shrinking dimer (HiOSD) method \parencite{JianyuanZhang2019} to compute the unstable biaxial torus configuration in both cases, and continuation methods to compute the bifurcation diagrams. The noticeable difference is that the biaxial region is much closer to the origin for the unstable biaxial torus solution. We deduce that there are only qualitative differences between the stationary points of the LdG free energies \eqref{4thLdGqi} and \eqref{6thLdGqi}, except that uniaxial solutions are more unstable (typically have higher Morse indices) in the sixth-order case and the biaxial stationary points have larger regions of stability, and larger biaxial regions in the sixth-order case compared to the fourth-order case. A relevant remark is that the biaxial torus solution is predominantly uniaxial away from the biaxial torus, and hence, it would be interesting to check if it retains stability when \eqref{eq:f2} strongly favours a bulk biaxial phase, with and without the symmetry constraints \eqref{symBC1}
	and \eqref{symBC2}.
	\begin{figure}[!ht]
		\begin{minipage}{0.49\textwidth}
			\hspace*{-0.1cm}\includegraphics[width=0.7\textwidth,angle=-90]{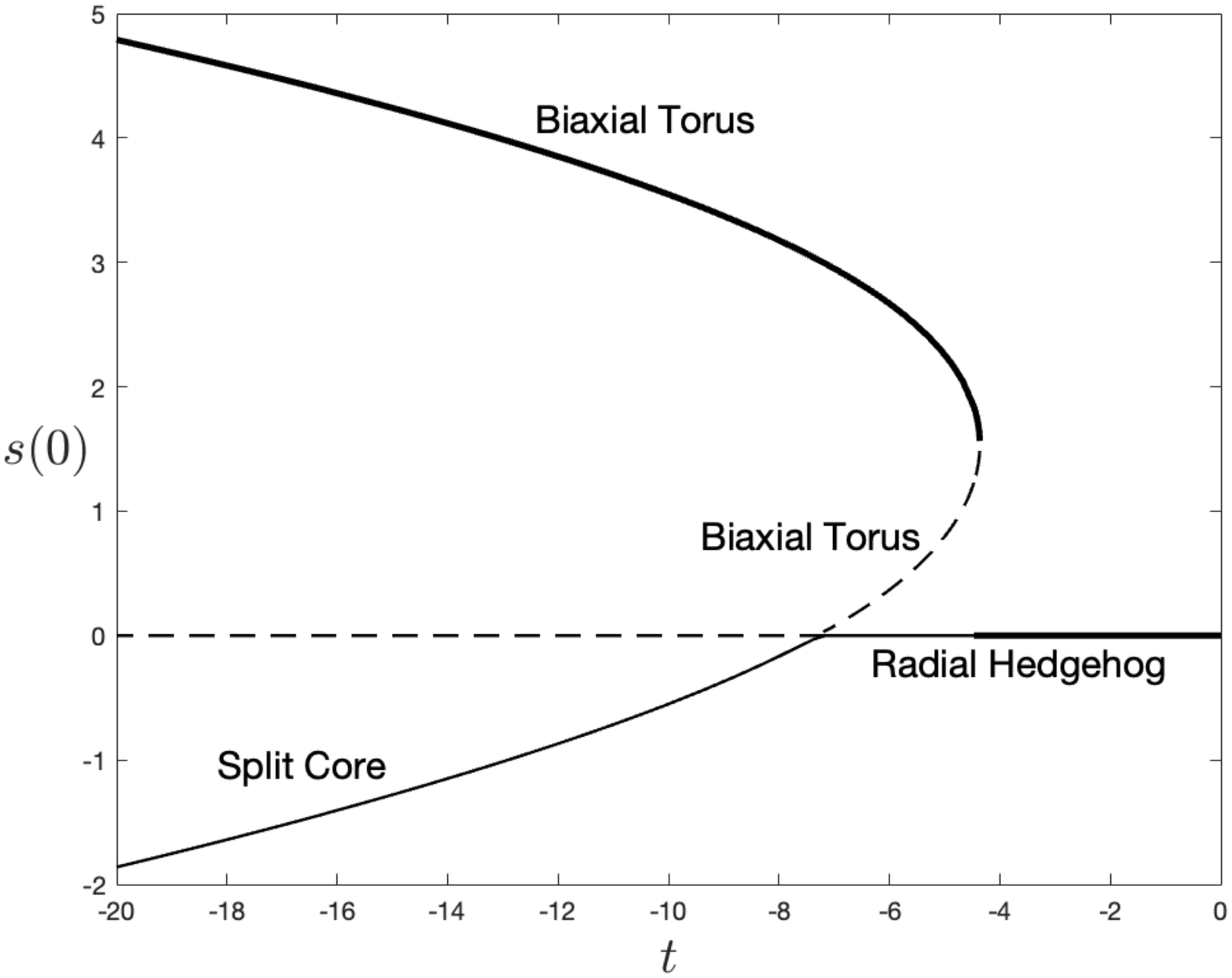}
			\subcaption{}\label{fig:4thBifurcation}
		\end{minipage}\hfill
		\begin{minipage}{0.49\textwidth}
			\hspace*{-0.1cm}\includegraphics[width=0.7\textwidth,angle=-90]{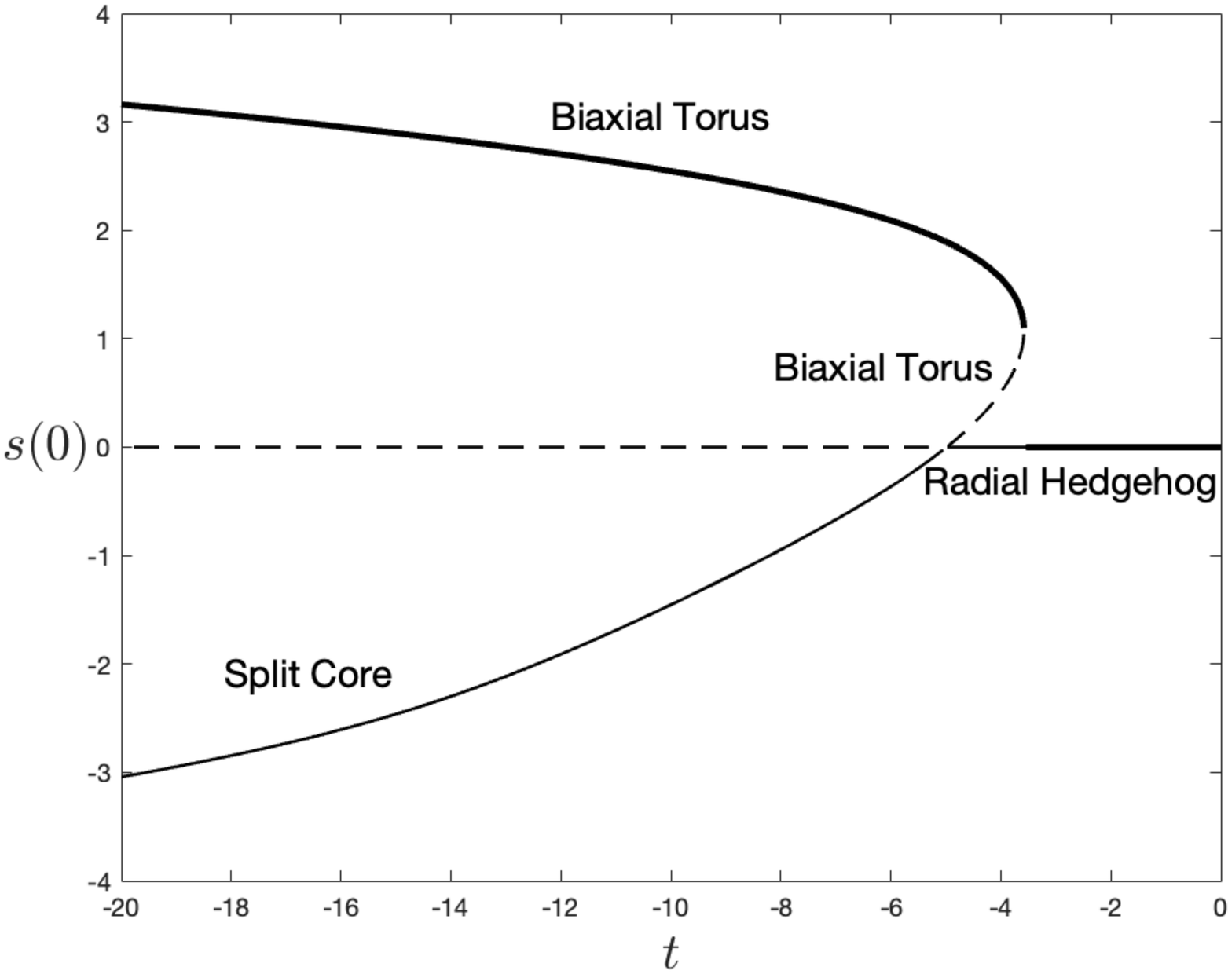}
			\subcaption{}\label{fig:6thBifurcation}
		\end{minipage}\hfill
		\caption{Bifurcation diagrams for the LdG free energies for \(\varepsilon = 0.5\) with (a) fourth-order potential \eqref{4thLdGqi} and (b) sixth-order potential \eqref{6thLdGqi} with \(d = 1, e = 0, f = 1\). We plot the scalar order parameter of each configuration. Bold solid lines indicate the global minimiser; thin solid lines indicate local minimality; and dashed lines indicate instability.}
	\end{figure}
        \begin{figure}[!ht]
            \centering
		\includegraphics[width=0.25\textwidth,angle=-90]{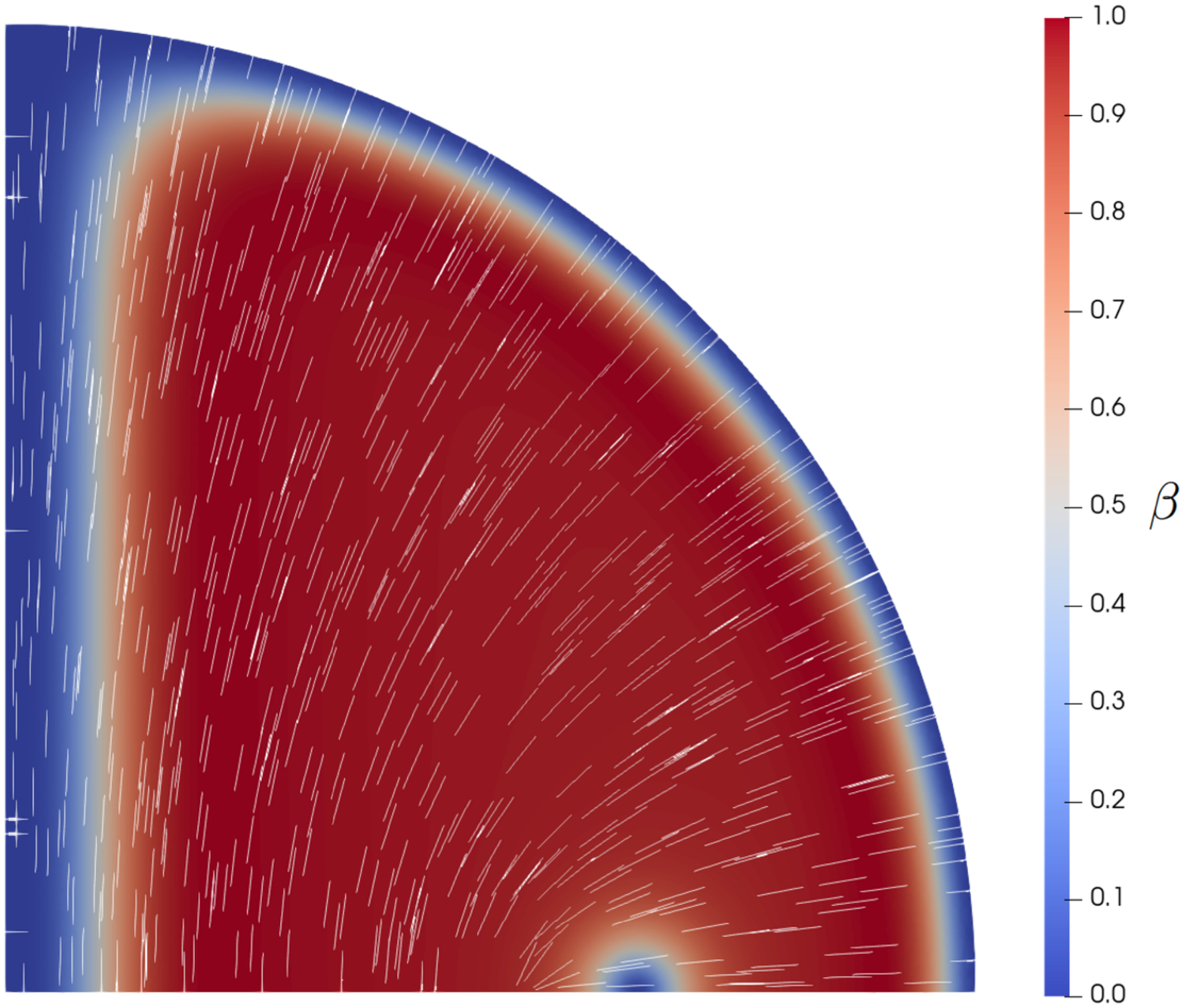}
		\caption{Biaxial torus configuration at \(t = -50\), \(\varepsilon = 0.5\), with \(d = 1, e = 0, f = 1\).} \label{fig:toruslowtemp}
	\end{figure}
	
	\begin{figure}[!ht]
		\begin{minipage}{0.245\textwidth}
			\centering
			\includegraphics[width=0.75\textwidth]{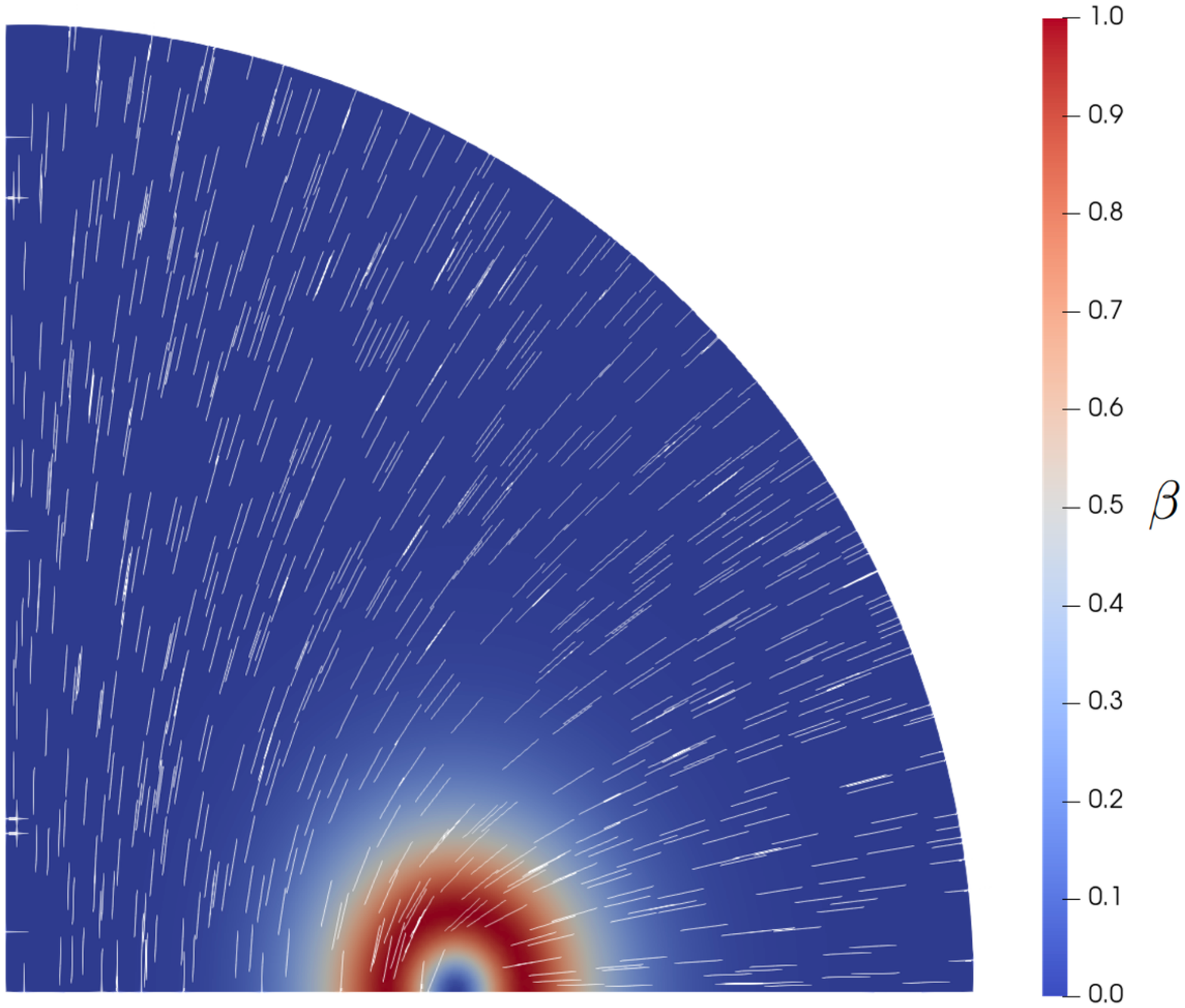}
			\subcaption{}\label{fig:4thstabletorus}
		\end{minipage}\hfill
		\begin{minipage}{0.245\textwidth}
			\centering
			\includegraphics[width=0.75\textwidth,angle=-90]{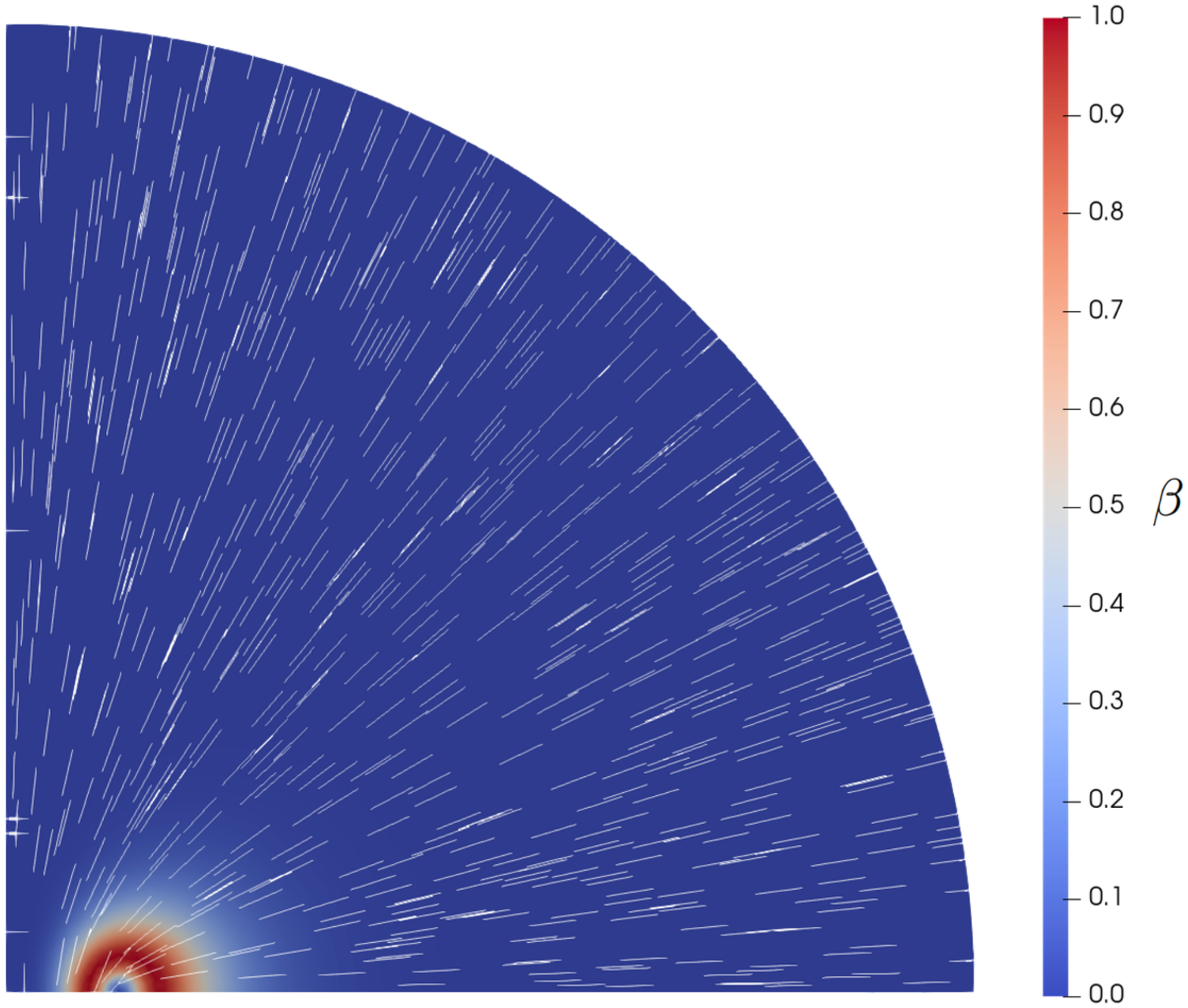}
			\subcaption{}\label{fig:4thunstabletorus}
		\end{minipage}
            \begin{minipage}{0.245\textwidth}
			\centering
			\includegraphics[width=0.75\textwidth,angle=-90]{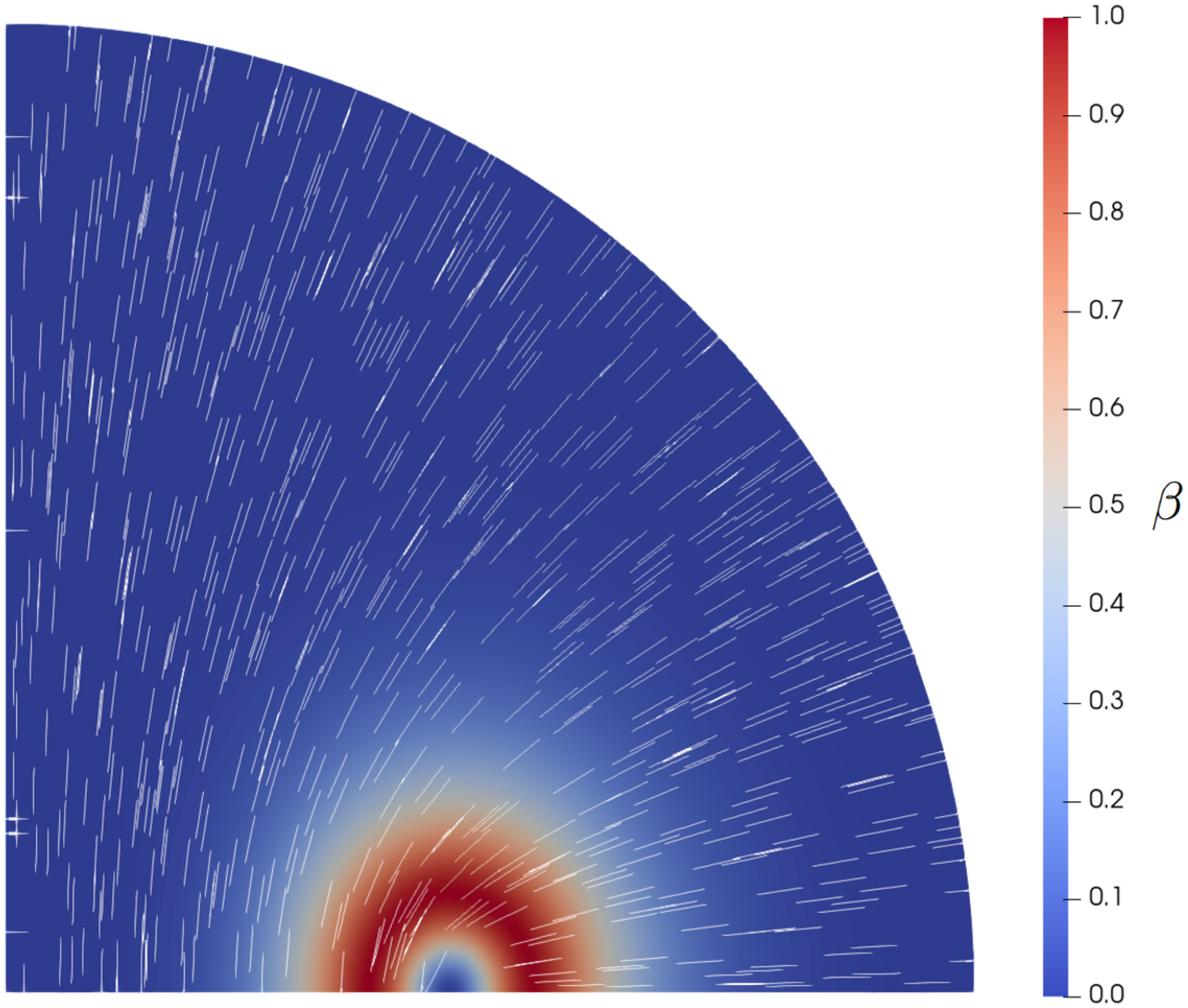}
			\subcaption{}\label{fig:6thstabletorus}
		\end{minipage}\hfill
		\begin{minipage}{0.245\textwidth}
			\centering
			\includegraphics[width=0.75\textwidth,angle=-90]{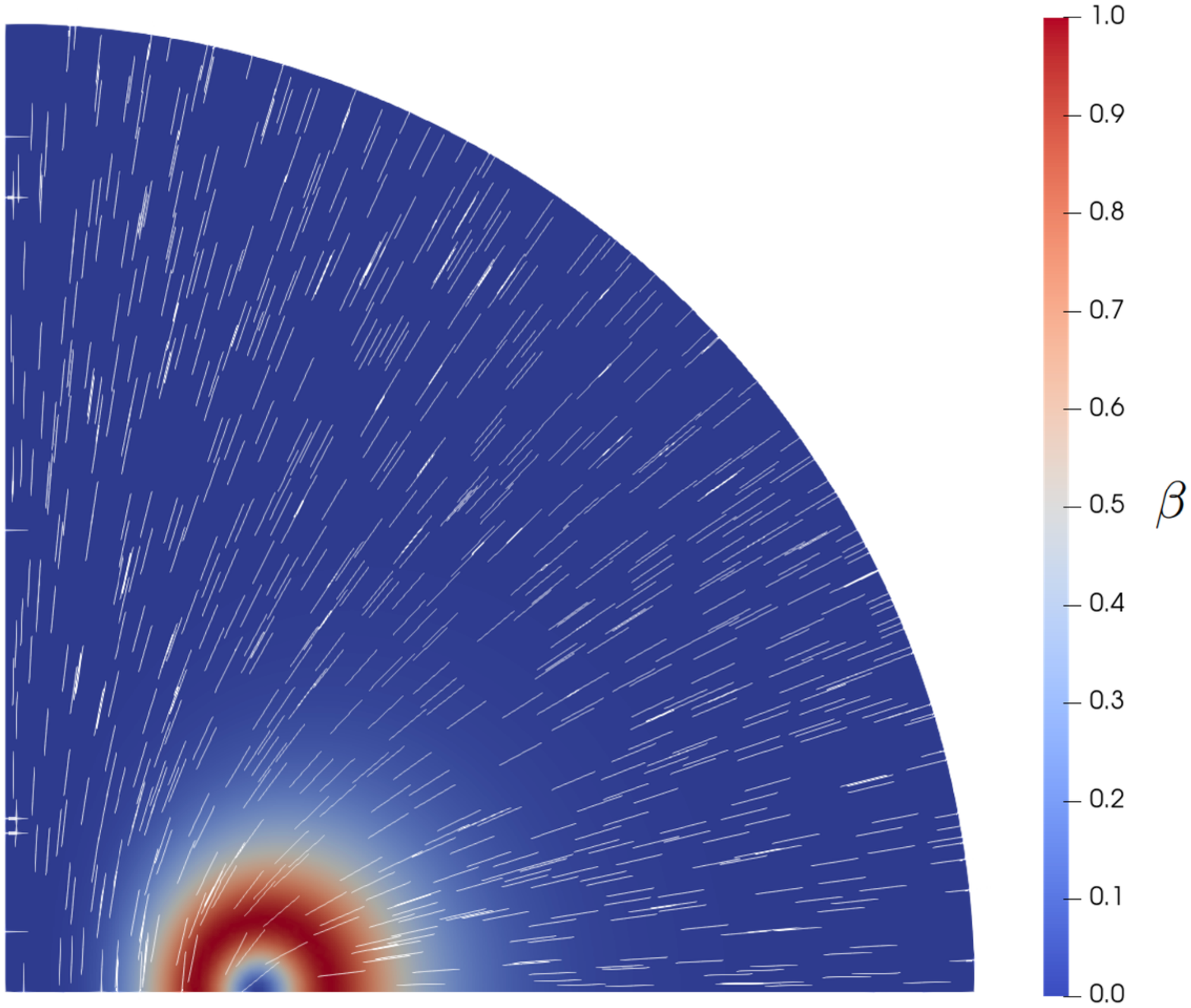}
			\subcaption{}\label{fig:6thunstabletorus}
		\end{minipage}
            \caption{(a) Stable biaxial torus configuration; and (b) unstable biaxial torus configuration for \eqref{4thLdGqi} with fourth-order potential, with \(t = -6.5\) with \(\varepsilon = 0.5\). (c) stable biaxial torus; and (d) unstable biaxial torus configuration for \eqref{6thLdGqi} with sixth-order potential, for \(t = -4\) and \(\varepsilon = 0.5, d = 1, e = 0, f = 1\).}
        \end{figure}

    \subsection{Conclusions}
    \label{sec:conclusion}

    In this paper, we perform some analytical and numerical studies of a LdG free energy with a sixth-order bulk potential \eqref{eq:f2}, as opposed to the vast majority of theoretical studies which rely on the fourth-order bulk potential \eqref{eq:f1}. The potential \eqref{eq:f2} admits a biaxial minimiser for sufficiently low temperatures, and in fact, does not admit stable uniaxial minimisers deep in the nematic phase. By direct analogy with \parencite{MajumdarZarnescu2010}, one can prove that global minimisers of \eqref{LdGgeneral} with \eqref{eq:f2}, will converge to minimisers of \eqref{eq:f2} almost everywhere (except for defects or boundary layers), for sufficiently large domains and are hence expected to demonstrate bulk biaxiality. This is of course, hugely interesting since bulk biaxiality is typically elusive and hard to detect experimentally. We give an example of a biaxial torus with a large biaxial region at the low temperature, \(t = -50\) in Figure \ref{fig:toruslowtemp}.

    We focus on the concrete example of the RH solution, as a critical point of \eqref{LdGgeneral} with \eqref{eq:f2}. There are many analogies with the fourth-order potential for moderately low temperatures, for which \eqref{eq:f2} admits a global uniaxial minimiser with positive order parameter, and differences arise deep in the nematic phase, when the RH solution need not be unique i.e. there are certainly multiple solutions of \eqref{sODE} for sufficiently low temperatures, and the global minimiser of \eqref{LdGs} is negative in the interior, away from $r=0$ and $r=1$. The non-uniqueness of solutions of \eqref{sODE} and negativity of the global minimiser of \eqref{LdGs} are outside the scope of the fourth-order potential. It is not clear if these results have physical implications. One could argue that the LdG model with the sixth-order potential is not necessarily valid for low temperatures, when \eqref{eq:f2} has a global biaxial minimiser. It is interesting that the global minimiser of \eqref{LdGs} corresponds to a critical point of \eqref{6thNDLdG} of the form
    \[
    \mathbf{Q}^* = s^* \left( \hat{\mathbf{r}} \otimes \hat{\mathbf{r}} - \frac{\mathbf{1}}{3}\textbf{I} \right)
    \]
    where $s^*$ is negative in the interior; this describes a uniaxial state for which the NLC molecules prefer to be orthogonal to the normal or prefer to be planar, which is consistent with cooling-induced homeotropic-planar structural transitions observed in some experiments on nematic shells \parencite{lagerwallmajumdarwang2020}. One could speculate that \eqref{eq:f2} captures this physical effect, which \eqref{eq:f1} cannot. Unsurprisingly, RH solutions have a smaller domain of stability as critical points of \eqref{6thNDLdG}, simply because \eqref{eq:f2} promotes bulk biaxiality for sufficiently low temperatures, and RH solutions are purely uniaxial with the exception of an isotropic point at the droplet centre.
    
  We also numerically compute the biaxial torus and split core solutions, as critical points of \eqref{6thLdGqi} with the additional symmetry constraints \eqref{symBC1} and \eqref{symBC2}. These critical points only exploit three out of the five degrees of freedom. We do not observe any significant differences between the fourth-order and sixth-order potential, except that the biaxial regions are larger with \eqref{eq:f2} and these ``locally" biaxial solutions have larger domains of stability as critical points of \eqref{6thLdGqi}, as opposed to critical points of \eqref{4thLdGqi}.

Besides the RH, split core, and biaxial torus solutions, which exist as critical points of \eqref{4thLdGqi} and \eqref{6thLdGqi}, we numerically compute a brand new biaxial critical point of \eqref{6thNDLdG} which exploits the full five degrees of freedom in Figure~\ref{fig:configuration}. This biaxial solution is almost maximally biaxial in the interior, except for the imposed uniaxial boundary condition and the labelled defect rings, and does not have rotational or mirror symmetry. Maximal biaxiality indicates a small or zero eigenvalue of the corresponding $\mathbf{Q}$-tensor. Focussing on the defect rings, this biaxial solution has a complete defect ring inside the sphere and two half defect rings connected to the boundary (Fig \ref{fig:configuration}(a)). The two half defect rings are located on the same plane, $x_1x_3$-plane (Fig \ref{fig:configuration}(d) and Fig \ref{fig:configuration}(c)), which is perpendicular to the $x_2x_3$-plane that contains the complete defect ring. From top to bottom in the $x_3$ direction, we cross the upper half-defect ring once, the complete defect ring twice, and the lower half-defect ring  once, and the defect lines are approximately uniaxial. Hence, in the second figure of Fig \ref{fig:line}, there are two low-biaxiality regions near $r = 1$ and the four low-biaxiality areas inside the sphere. In the third figure of Fig \ref{fig:line}, in the $x_2$ direction, the complete defect ring is crossed twice and hence, there are two uniaxial points at the end-points (because of the imposed boundary condition) and two regions of low biaxiality enclosed by the complete defect ring.
The eigenvector corresponding to the largest eigenvalue of $\mathbf{Q}$ is almost parallel to $x_1$ and perpendicular to the plane of the complete defect ring, inside the complete defect ring, and is almost radial elsewhere Fig \ref{fig:configuration}(b).  This biaxial solution exists when the temperature is low enough so that the global minimiser of \eqref{eq:f2} is biaxial. Are there experimental implications? Provided \eqref{6thNDLdG} is valid for such low temperatures ($t < -50$), Figure~\ref{fig:configuration}(b) suggests that the optical signature of this biaxial solution should be completely different from the optical signatures of the RH, biaxial torus and split core solutions which have a predominantly radial director. This could be verified by taking optical measurements in the $x_2 x_3$ plane. Of course, there are challenges related to the choices of the material parameters in \eqref{eq:f2}, about which little is known. Importantly, is this biaxial solution a potential route for observing bulk biaxiality, since it is almost maximally biaxial in the interior? We cannot comment on this, since approximately biaxial configurations are described by a triad of eigenvalues, $(\lambda_1, \lambda_2, \lambda_3) = (l + \nu, -\nu - \mu, -l + \mu)$ for some positive $l$, and small $\nu, \mu \in \mathbb{R}$. Physically, this means that the molecules prefer to align along the corresponding eigenvector $\mathbf{e}_1$ and orthogonal to the eigenvector $\mathbf{e}_3$, with disorder in the direction $\mathbf{e}_2$. From a theoretical perspective, this is biaxial but experimentally, one may only detect ordering along $\mathbf{e}_1$ and the optical textures may resemble uniaxial textures with the uniaxial director oriented along $\mathbf{e}_1$. In Figure~\ref{fig:eigenf2}, we plot the eigenvalues of the biaxial minimiser of \eqref{eq:f2}, when it exists and it seems that there is always one small eigenvalue. In this case, the eigenvector with the largest positive eigenvalue may be experimentally identifiable with the director, and one should certainly get different optical measurements along the remaining two eigenvectors. We cannot comment on whether this suggestion of taking experimental measurements in transverse cross-sections can provide a reliable tool for observing biaxiality, since these measurements can potentially be attributed to factors other than biaxiality too.
The details of the numerical method can be found in Appendix B. Future work will include a detailed study of solution landscapes of \eqref{6thELQ} for different model problems.

\begin{figure}[ht]
    \centering
\includegraphics[width=0.3\textwidth,angle=-90]{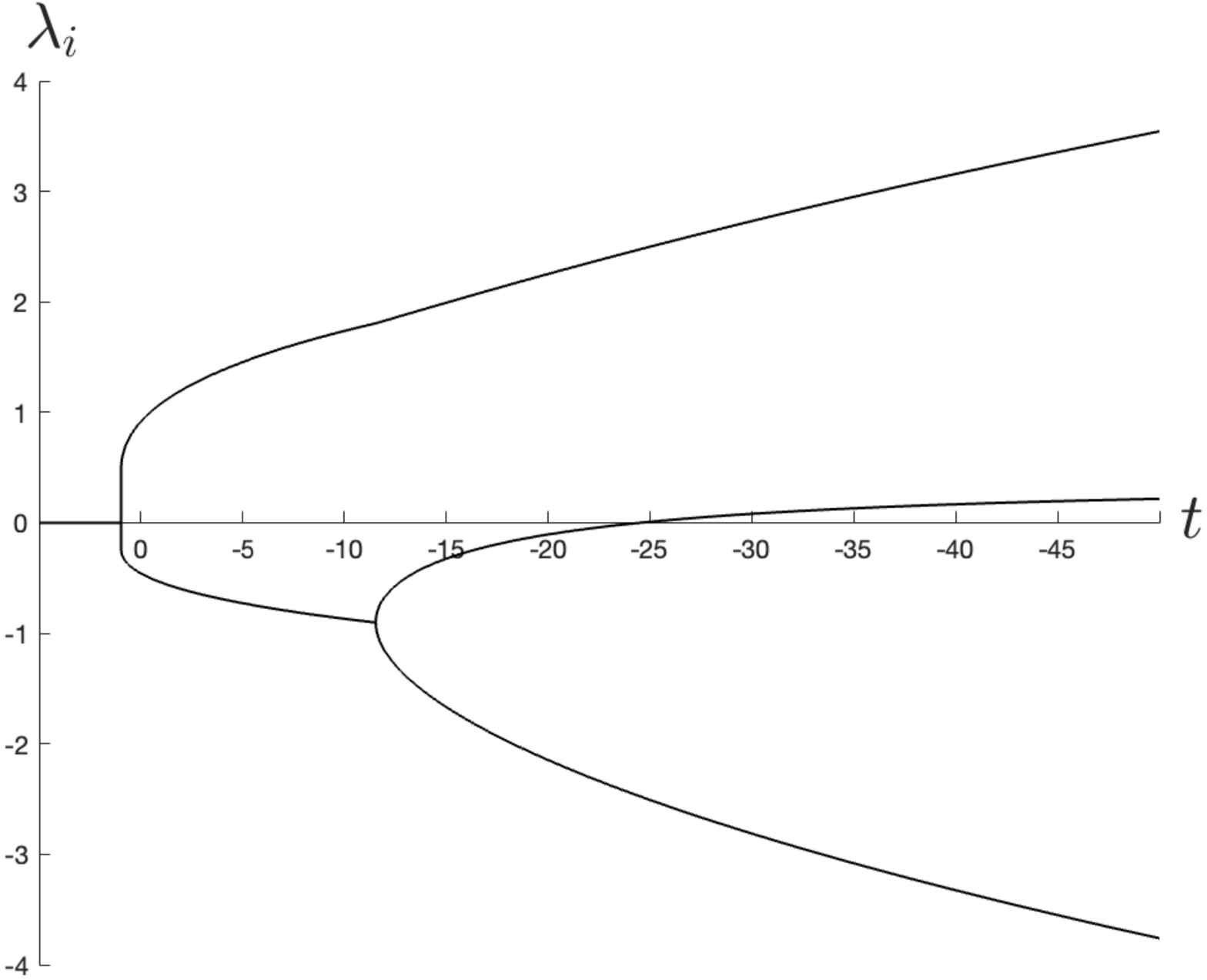}\\
\caption{The eigenvalues of the global minimiser of \eqref{eq:f2} as a function of the temperature.} \label{fig:eigenf2}
 \end{figure}
 \begin{figure}[ht]
    \centering
\includegraphics[width=\textwidth]{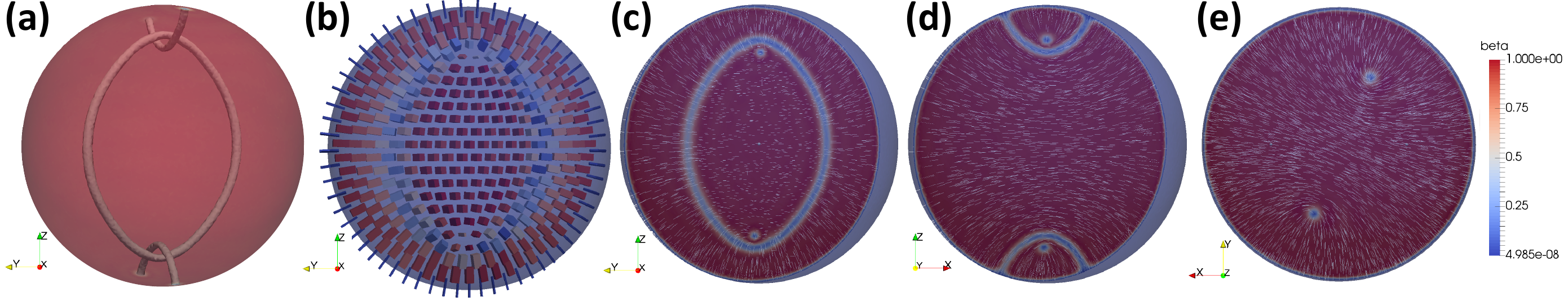}\\
    \caption{The profile of biaxial solution with parameters
    $t = -50$, $\epsilon = 0.2$, $d = 1$, $e = 0$, $f = 1$, $w=1e5$, $(N,L,M) = (64,64,32)$. (a) The contour of $\beta = 0.5$ inside the ball $B(0,0.9)$. (b) On the plane perpendicular to $x_1$, the box represents $\mathbf{Q}$-tensor with three edge lengths corresponding to $\lambda_i+s_+/3$, $i = 1,2,3$ where $\lambda_i$ are three eigenvalues of $\mathbf{Q}$ and three edge directions corresponding to the three eigenvectors of $\mathbf{Q}$, $\mathbf{n}_i$, $i = 1,2,3$.  (c-e) Cross-sections of the solution with normal vector $x_i$, $i = 1,2,3$ viewed along $x,y,z$-axes. The colour represents biaxiality $\beta$, and the white lines represent the leading eigenvector $\mathbf{n}_1$.}
    \label{fig:configuration}
\end{figure}
\begin{figure}[ht!]
    \centering
\includegraphics[width=\textwidth]{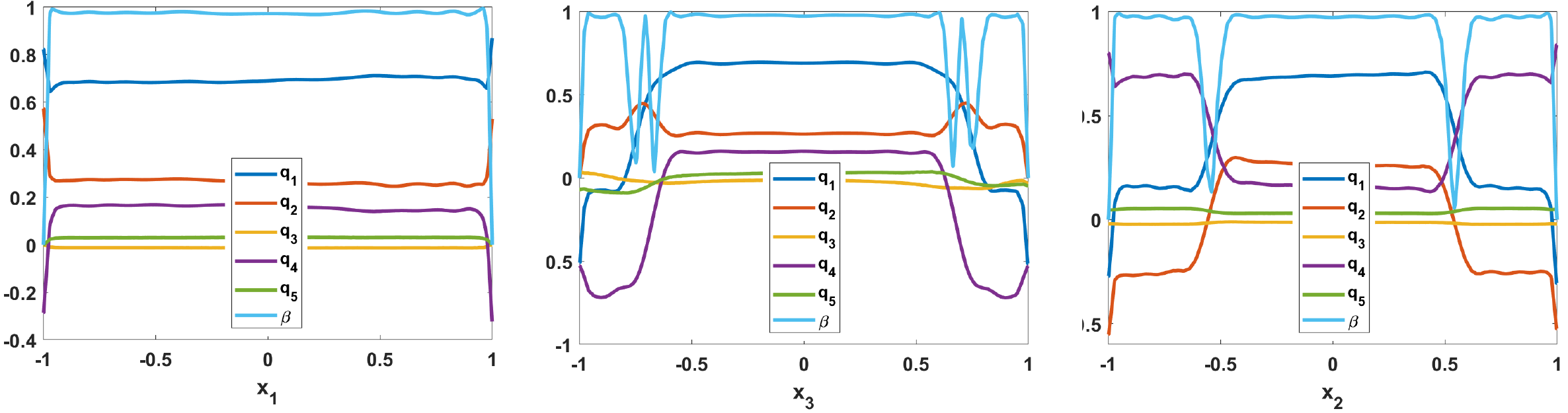}\\
    \caption{The plots of $q_i$, $i = 1,\cdots,5$ and $\beta$ on the lines via origin along $\mathbf{x}_1$, $\mathbf{x}_2$, and $\mathbf{x}_3$ directions.}
    \label{fig:line}
\end{figure}

 \section*{Acknowledgments}
AM is supported by a Leverhulme Research Project Grant, the University of Strathclyde New Professors Fund, a Leverhulme International Academic Fellowship, an OCIAM Visiting Fellowship at the University of Oxford and a Daiwa Foundation Small Grant. YH is supported by a Leverhulme Research Project Grant. SM is supported by an Engineering and Physical Sciences Research Council studentship.
 \newpage
	\printbibliography
 
	\newpage
	
	\section*{Appendix A}
	We note that the proofs of Propositions 2 and 3 begin with the general Euler–Lagrange equations corresponding to the free energy (\ref{6thNDLdG}). We then restrict both problems to the uniaxial case, and we use the structure of the RH solution to reduce the proofs to arguments involving just a scalar order parameter, \(s\).
	
	\noindent \textit{Proof of Proposition 2.} We consider two subsets:
	\begin{equation}
		\Omega_+ = \big\{\boldsymbol{r} \in B(0,1): s(\boldsymbol{r}) \geq 0\big\}, \quad \text{and} \quad  \Omega_- = \big\{\boldsymbol{r} \in B(0,1): s(\boldsymbol{r}) < 0\big\}.
	\end{equation}
	Suppose that the subset
	\begin{equation}
		\widetilde{\Omega} = \Bigg\{\boldsymbol{r} \in B(0,1): |\textbf{Q}^*(\boldsymbol{r}^*)|^2 > \frac{2}{3}\max\big\{s_+^2, s_-^2\big\}\Bigg\}\subset B(0,1)\setminus\partial B(0,1) \subset \Omega_+ \cup \Omega_-
	\end{equation}
	is nonempty. The subset \(\widetilde{\Omega}\) does not intersect \(\partial B(0,1)\) since \(\max\big\{s_+^2, s_-^2\big\} \geq s_+^2\). Moreover, the function \(|\textbf{Q}^*|:\overline{B(0,1)}\to\mathbb{R}\) attains a strict maximum at an interior point \(\boldsymbol{r}^*\in\widetilde{\Omega}\).
	
	We multiply the Euler–Lagrange equations (\ref{6thELQ}) by \(Q_{ij}^*\) to find that
	\begin{equation}
		\varepsilon\bigg(\frac{1}{2}\Delta|\textbf{Q}^*|^2 - |\nabla\textbf{Q}^*|^2\bigg) = |\textbf{Q}^*|^2 - 3\sqrt{6}\tr\textbf{Q}^{*3} + 2|\textbf{Q}^*|^4 + d|\textbf{Q}^*|^2\tr\textbf{Q}^{*3} + e|\textbf{Q}^*|^6 + \frac{(f - e)}{6}\big(\tr\textbf{Q}^{*3}\big)^2.
	\end{equation}
	since \(|\nabla\textbf{Q}^*|^2 + Q_{ij}^*\Delta Q_{ij}^* = \frac{1}{2}\Delta|\textbf{Q}^*|^2\). We note that \(\frac{1}{2}\Delta|\textbf{Q}^*(\boldsymbol{r}^*)|^2 - |\nabla\textbf{Q}^*(\boldsymbol{r}^*)|^2 \leq 0\) at the interior maximum. 
	
	Let us label
	\begin{equation}
		h(\textbf{Q}) := |\textbf{Q}|^2 - 3\sqrt{6}\tr\textbf{Q}^3 + 2|\textbf{Q}|^4 + d|\textbf{Q}|^2\tr\textbf{Q}^3 + e|\textbf{Q}|^6 + \frac{(f - e)}{6}\big(\tr\textbf{Q}^3\big)^2.
	\end{equation}
	The aim is to show that \(h\) is positive at \(\textbf{Q}^*(\boldsymbol{r}^*)\) for a contradiction. First, consider the case where \(\boldsymbol{r}^* \in \Omega_+\). Then we may write
	\begin{equation}
		h(\textbf{Q}^*(\boldsymbol{r}^*)) = t|\textbf{Q}^*(\boldsymbol{r}^*)|^2 - 3|\textbf{Q}^*(\boldsymbol{r}^*)|^3 + 2|\textbf{Q}^*(\boldsymbol{r}^*)|^4 + \frac{d}{\sqrt{6}}|\textbf{Q}^*(\boldsymbol{r}^*)|^5 + e|\textbf{Q}^*(\boldsymbol{r}^*)|^6 + \frac{(f - e)}{6}|\textbf{Q}^*(\boldsymbol{r}^*)|^6
	\end{equation}
	since \(\tr\textbf{Q}^{*3} = \frac{1}{\sqrt{6}}|\textbf{Q}^*|^3\). Note that
	\(h(\textbf{Q}) = \sqrt{\frac{3}{2}}|\textbf{Q}|g'\Big(\sqrt{\frac{3}{2}}|\textbf{Q}|\Big)\), and \(\sqrt{\frac{3}{2}}|\textbf{Q}| = |s|\), for an arbitrary uniaxial \textbf{Q}-tensor of the form \(\textbf{Q}_s = s\left(\boldsymbol{n}\otimes\boldsymbol{n} - \frac{1}{3}\textbf{I}\right)\).
	Therefore, the sign of \(h(\textbf{Q})\) is dictated by the sign of \(g'(|s|)\).
	
	Let us write \(|\textbf{Q}^*(\boldsymbol{r})| = \sqrt{\frac{2}{3}}|s^*(\boldsymbol{r})|\). We have noted in Section 3 that we are working in a parameter regime such that the function \(g\) is a double-welled potential with \(g'(s_-) = g'(s_+) = 0\) below some transition temperature \(t_0\). Moreover, we choose \(e\) and \(f\) so that \(g(s) \to +\infty\) as \(|s| \to +\infty\), and since \(|s^*(\boldsymbol{r}^*)| > \max\big\{s_+,|s_-|\big\}\), then we may conclude that \(g'(|s^*(\boldsymbol{r}^*)|) > 0\) at the interior maximum \(\boldsymbol{r}^* \in \widetilde{\Omega}\). Hence \(h(\textbf{Q}^*(\boldsymbol{r}^*)) > 0\), and there cannot be a strict interior maximum at \(\boldsymbol{r}^* \in \Omega_+\).
	
	Now consider the case where \(\boldsymbol{r}^* \in \Omega_-\). We may write
	\begin{equation}
		h(\textbf{Q}^*(\boldsymbol{r}^*)) = t|\textbf{Q}^*(\boldsymbol{r}^*)|^2 + 3|\textbf{Q}^*(\boldsymbol{r}^*)|^3 + 2|\textbf{Q}^*(\boldsymbol{r}^*)|^4 - \frac{d}{\sqrt{6}}|\textbf{Q}^*(\boldsymbol{r}^*)|^5 + e|\textbf{Q}^*(\boldsymbol{r}^*)|^6 + \frac{(f - e)}{6}|\textbf{Q}^*(\boldsymbol{r}^*)|^6, 
	\end{equation}
	since \(\tr\textbf{Q}^{*3} = -\frac{1}{\sqrt{6}}|\textbf{Q}^*|^3\),
	and we note that \(h(\textbf{Q}) = -\sqrt{\frac{3}{2}}|\textbf{Q}|g'\Big(-\sqrt{\frac{3}{2}}|\textbf{Q}|\Big)\).
	Therefore, the sign of \(h(\textbf{Q})\) is dictated by the sign of \(-g'(-|s|)\). Then, since \(g\) is a double-welled potential with \(g'(s_-) = g'(s_+) = 0\) such that \(g(s) \to +\infty\) as \(|s| \to +\infty\), and \(|s^*(\boldsymbol{r}^*)| > \max\big\{s_+, |s_-|\}\), then we may conclude that \(-g'(-|s^*(\boldsymbol{r}^*)|) > 0\) at the interior maximum \(\boldsymbol{r}^*\in\widetilde{\Omega}\). Hence \(h(\textbf{Q}^*(\boldsymbol{r}^*)) > 0\), and there cannot be a strict interior maximum at \(\boldsymbol{r}^* \in \Omega_-\).
	
	Thus, we combine the above two cases to find that the set \(\widetilde{\Omega}\) must be empty and the global minimiser \(\textbf{Q}^*\) in the class of uniaxial \textbf{Q}-tensors must satisfy the upper bound
	\begin{equation}
		|\textbf{Q}^*|^2 \leq \frac{2}{3}\max\big\{s_+^2,s_-^2\big\}.
	\end{equation}
	\qed
	\medskip
	
	\noindent \textit{Proof of Proposition 3.} We prove uniqueness via a contradiction argument, relying on a Pohozaev identity
	\begin{multline}
		\varepsilon^2\Bigg(\frac{1}{2}\int_{B(0,1)}Q_{ij,\ell}Q_{ij,\ell}\,dV + \int_{\partial B(0,1)}Q_{ij,k}x_kQ_{ij,\ell}x_{\ell}\,dS - \frac{1}{2}\int_{\partial B(0,1)}Q_{ij,\ell}Q_{ij,\ell}\,dS\Bigg) \\
		= \int_{\partial B(0,1)}f_B(\textbf{Q})\,dS - 3\int_{B(0,1)}f_B( \textbf{Q})\,dV, \label{Pohozaev}
	\end{multline}
	which is obtained from the Euler–Lagrange equations (\ref{6thELQ}) as is done in \parencite{MajumdarZarnescu2010}.
	
	We rewrite (\ref{Pohozaev}) as
	\begin{equation}
		\mathcal{F}[\textbf{Q}] + 2\int_{B(0,1)}f_B(\textbf{Q})\,dV = \int_{\partial B(0,1)}f_B(\textbf{Q})\,dS + \frac{1}{2}\int_{\partial B(0,1)}Q_{ij,\ell}Q_{ij,\ell}\,dS - \int_{\partial B(0,1)}\big(Q_{ij,k}x_k\big)^2\,dS. \label{uniqueeq1}
	\end{equation}
	Suppose for a contradiction that there exist \(s_1, s_2 \in \mathcal{A}_s, s_1 \neq s_2\), satisfying
	\begin{equation}
		I[s_1] = I[s_2] = \min_{\mathcal{A}_s}I.
	\end{equation}
	We apply (\ref{uniqueeq1}) to \(\textbf{Q}_{s_1}\) and \(\textbf{Q}_{s_2}\), simplify the resulting equations to obtain two equations involving \(s_1\) and \(s_2\), and subtract the second from the first to obtain the relation 
    \begin{equation}
		6\int_{0}^{1}r^2\big(g(s_1) - g(s_2)\big)\,dr = \big(s_2'(1)\big)^2 - \big(s_1'(1)\big)^2, \label{uniqueeq6}
	\end{equation}
	recalling that \(s_1(1) = s_2(1) = s_+\).
	
	The two functions \(s_1\) and \(s_2\) are distinct solutions of the Euler–Lagrange equation (\ref{sODE}) corresponding to the minimisation of \(I\). Hence Lemma 2 in \parencite{Lamy2013} ensures they cannot coincide on a neighbourhood of zero. Suppose, without loss of generality, that \(s_1 < s_2\) on \((0,\varepsilon)\). To show that we must in fact have \(s_1 < s_2\) on \((0,1)\), suppose for a contradiction that there exists an \(r_0 \in (0,1)\) such that \(s_1(r_0) = s_2(r_0)\). We define the function \(\tilde{s}\) by
	\begin{equation}
		\tilde{s}(r) = \begin{cases}
			s_2(r), &r \in (0,r_0], \\
			s_1(r), &r \in (r_0,1),
		\end{cases}
	\end{equation}
	and we show that \(\tilde{s}\) is a minimiser of \(I\). Denoting by \(h[s]\) the energy density (\(I[s] = \int_{0}^{1}h[s]\,dr\)), and setting
	\begin{equation}
		\bar{s}(r) = \begin{cases}
			s_1(r), &r \in (0,r_0], \\
			s_2(r), &r \in (r_0,1),
		\end{cases}
	\end{equation}
	we find that
	\begin{equation}
		I[s_2] \leq I[\bar{s}] = \int_{0}^{r_0}h[s_1]\,dr + \int_{r_0}^{1}h[s_2]\,dr,
	\end{equation}
	since \(s_2\) is a minimiser and \(\bar{s}\) lies in the admissible space. Therefore, it holds that \(\int_{0}^{r_0}h[s_2]\,dr \leq \int_{0}^{r_0}h[s_1]\,dr,\) since
	\begin{equation}
		I[s_2] = \int_{0}^{r_0}h[s_2]\,dr + \int_{r_0}^{1}h[s_2]\,dr \leq \int_{0}^{r_0}h[s_1]\,dr + \int_{r_0}^{1}h[s_2]\,dr.
	\end{equation}
	Adding \(\int_{r_0}^{1}h[s_1]\,dr\) to both sides of the inequality yields \(I[\tilde{s}] \leq I[s_1]\), so we may conclude that \(\tilde{s}\) is a minimiser. Since \(\tilde{s}\) is a minimiser, it must be analytic by Proposition 1. Therefore, at \(r_0\) all of its right derivatives are equal to those of \(s_1\). This tells us that \(\tilde{s} = s_1\) on a neighbourhood of \(r_0\), which implies that \(s_1 = s_2\). This contradicts the assumption that \(s_1 < s_2\) on \((0,\varepsilon)\). Therefore, we find that \(s_1 < s_2 \,\, \text{on} \,\, (0,1)\). This implies, together with \(s_1(1) = s_2(1)\), that
	\begin{equation}
		s_1'(1) \geq s_2'(1),
	\end{equation}
	so the right-hand side of (\ref{uniqueeq6}) is non-positive.
	
	On the other hand, since \(t < 0\), we can show that \(g'(u) < 0\) for \(u \in (0,s_+)\), where \(g\) is defined in \eqref{gdefn} which is the same as
	\begin{equation}
		\frac{4(f + 5e)}{81}s_+^3 + \frac{4d}{27}s_+^2 + \frac{8}{9}s_+ - \frac{2\sqrt{6}}{3} < 0,
	\end{equation}
	provided \(d, f > 0\) and \(e > - \frac{1}{5}f\). Hence the energy density \(g\) is decreasing on \([0,s_+]\). Then since \(s_1 < s_2\), we find that
	\begin{equation}
		\int_{0}^{1}\big(g(s_1) - g(s_2)\big)\,dr > 0.
	\end{equation}
	Therefore the left-hand side of (\ref{uniqueeq6}) is positive, and we have reached our contradiction.
 
 We prove nonnegativity via a contradiction with the assumption that there exists an interior measurable subset
	\begin{equation}
		\Gamma = \big\{r \in (0,1):s^*(r) < 0\big\}\subset [0,1],
	\end{equation}
	with \(s^*(r) = 0\) on \(\partial\Gamma\). We define the perturbation
	\begin{equation}
		\bar{s}^* = \begin{cases}
			s^*(r), &r \in [0,1]\setminus\Gamma, \\
			-s^*(r), &r \in \Gamma.
		\end{cases}
	\end{equation}
	Then
	\begin{equation}
		I[\bar{s}^*] - I[s^*] = \int_{\Gamma}\Bigg(\frac{4\sqrt{6}}{9}s^{*3} - \frac{8d}{135}s^{*5}\Bigg)r^2\,dr < 0,
	\end{equation}
	where \(I\) is defined in (\ref{LdGs}), if \(s^{*2} < \frac{15\sqrt{6}}{2d} \,\, \text{for} \,\, d > 0\) since \(s^*(r) < 0\) on \(\Gamma\) by assumption. Also, since \(s^{*2} \leq s_+^2\) by Proposition 2, we can guarantee that \(I[\bar{s^*}] - I[s^*] < 0\) if \(s_+^2 < \frac{15\sqrt{6}}{2d}\). However, this contradicts the energy minimality of \(s^*\). It follows that \(s^*(r)\geq 0\) for \(r \in [0,1]\) if \(s_+^2 < \frac{15\sqrt{6}}{2d}.\)
	
	To show that \(s^*(r) > 0\) for \(r > 0\) assume for a contradiction that there exists some \(r_0 \in (0,1]\) such that \(s^*(r_0) = 0\). Since we have already shown that \(s^*(r) \geq 0\) on \([0,1]\), the function \(s^*\) must therefore have a minimum at \(r_0\). Then
	\begin{equation}
		\frac{ds^*}{dr}\Big|_{r = r_0} = 0 \quad \text{and} \quad \frac{d^2s^{*}}{dr^2}\Big|_{r = r_0} \geq 0.
	\end{equation}
	However, if we substitute \(s^*(r_0)\) into (\ref{sODE}), we find that \(\frac{d^2s^{*}}{dr^2}\Big|_{r = r_0} = 0\). We can repeat this process to find that, in fact, \(\frac{d^ns^*}{dr^n}\Big|_{r = r_0} = 0\) for all \(n \in \mathbb{N}\). However, this cannot be true because we know from Proposition 1 that \(s^*\) is analytic and we have the boundary condition \(s^*(1) = s_+\). Therefore, we have reached a contradiction, so \(s^*(r) > 0\) on \((0,1]\).

    We prove monotonicity using an argument analogous to \parencite[Proposition 3]{Lamy2013}.
    
	\qed
	\medskip

    \noindent \textit{Proof of Proposition 5.} We consider a general biaxial perturbation
	\begin{equation}
		\hat{\textbf{Q}}(\boldsymbol{r}) = \begin{cases}
			\textbf{Q}^*(\boldsymbol{r}) + \tilde{p}(r)\bigg(\boldsymbol{z}\otimes\boldsymbol{z} - \frac{1}{3}\textbf{I}\bigg), &0 \leq r \leq \sigma, \\
			\textbf{Q}^*(\boldsymbol{r}), &\sigma \leq r \leq 1,
		\end{cases}
	\end{equation}
	where \(\tilde{p}:[0,1] \to \mathbb{R}\) is nonzero for \(0 < r < \sigma, \tilde{p}(0) = 0\), and \(\tilde{p}(r) = 0\) for \(\sigma \leq r \leq 1\), and \(\textbf{Q}^*\) is the RH solution. We find that
	\begin{equation}
		\frac{1}{4\pi}\big(\mathcal{F}[\hat{\textbf{Q}}] - \mathcal{F}[\textbf{Q}^*]\big) \leq \begin{aligned}[t]
			\int_{0}^{\sigma}\Bigg(\frac{\varepsilon^2}{3}\bigg(\frac{d\tilde{p}}{dr}\bigg)^2 &+ \frac{t}{3}\tilde{p}^2 - \frac{2\sqrt{6}}{9}\tilde{p}^3 + \frac{28}{45}s^{*2}\tilde{p}^2 \\
			&+ \frac{2}{9}\tilde{p}^4 + \frac{d}{5}\bigg(\frac{4}{27}s^{*3}\tilde{p}^2 + \frac{52}{135}s^{*2}\tilde{p}^3 + \frac{4}{27}\tilde{p}^5\bigg) \\
			&+ \frac{e}{6}\bigg(\frac{8}{5}s_+^4\tilde{p}^2 + \frac{128}{945}s^{*3}\tilde{p}^3 + \frac{8}{5}s^{*2}\tilde{p}^4 + \frac{16}{9}s^*\tilde{p}^5 + \frac{8}{27}\tilde{p}^6\bigg) \\
			&+ \frac{(f - e)}{6}\bigg(\frac{4}{45}s_+^4\tilde{p}^2 + \frac{112}{405}s^{*3}\tilde{p}^3 + \frac{4}{45}s^{*2}\tilde{p}^4 + \frac{4}{81}\tilde{p}^6\bigg)\Bigg)r^2\,dr
		\end{aligned}
	\end{equation}
	For large negative \(t\), we can approximate \(s_+\) by \(s_+ \approx \Bigg(\dfrac{-27t}{2(f + 5e)}\Bigg)^{1/4}\). Suppose we are working with large negative \(t\). Then, setting \(\sigma = 0.1\) and substituting 
	\begin{equation}
		\tilde{p}(r) = \frac{1}{(r^2 + 12)^2}\bigg(1 - \frac{r}{\sigma}\bigg) \label{pertpdefn}
	\end{equation}
	into the above, we find that \(\frac{1}{4\pi}\big(\mathcal{F}[\hat{\textbf{Q}}] - \mathcal{F}[\textbf{Q}^*]\big) < 0\) if
	\begin{equation}
		t\lesssim\frac{500\varepsilon^2(77184e + 16437f)}{428440e - 32975f}.
	\end{equation}
	Therefore we may conclude that the biaxial perturbation (\ref{perturbation})
	with \(\sigma = 0.1\), and \(\tilde{p}\) as in \eqref{pertpdefn} has lower free energy than the RH in the low temperature regime.
	\qed
	\medskip
	
	\section*{Appendix B}
 \subsection*{Numerical Method for Finding the Biaxial State}
Instead of using a Dirichlet boundary condition, the homeotropic anchoring is imposed on the surface of the sphere, $\partial B(0,1)$ by the following surface energy:
\begin{equation}\label{eq:surface}
\tilde{F}_s = \int_{\partial B(0,1)} \frac{w}{2}(\tilde{\mathbf{Q}}-\mathbf{Q}_{s_+})^2 dS,
\end{equation}
where $\mathbf{Q}_{s_+}$ is defined in \eqref{DirichletBC}, $w = \frac{27CW}{B^2R}$ is the nondimensionalised anchoring strength, and $W$ is the anchoring strength. We set $w = 1e5$ in our numerical calculation, which is extremely strong anchoring which plays almost the same role as the Dirichlet boundary conditions in \eqref{DirichletBC}.

Assuming the order parameter is given by
\begin{equation}
\mathbf{Q}=
\begin{pmatrix}
q_1 & q_2 & q_3\\
q_2 & q_4 & q_5\\
q_3 & q_5 & -q_1-q_4
\end{pmatrix},
\end{equation}
we can expand the component of $\mathbf{Q}$-tensor in terms of Zernike polynomials,
\begin{equation}\label{eq:expension}
q_i(r,\theta,\phi) = \sum_{m = 1-M}^{M-1}\sum_{l = |m|}^{L-1}\sum_{n=l}^{N-1} A_{nlm}^{(i)}Z_{nlm}(r,\theta,\phi),
\end{equation}
where $N\geq L\geq M\geq 0$ specify the truncation limits of the expanded series, with
\begin{align}
&Z_{nlm}(r,\theta,\phi) = R_n^{(l)}(r)P_l^{|m|}(\cos\theta)X_m(\phi),\\
&R_n^{(l)}(r) = \begin{cases}
\sum_{s=0}^{(n-l)/2}N_{nls} r^{n-2s},  & \text{if $\frac{n-l}{2}\geq 0$, $\frac{n-l}{2}\in\mathrm{Z}$,} \\
0, & \text{otherwise,}
\end{cases}\\
&N_{nls} = (-1)^s\sqrt{2n+3}\prod_{i=1}^{n-l}(n+l-2s+1+i)\prod_{i=1}^l(\frac{n-l}{2}-s+i)\frac{2^{l-n}}{s!(n-s)!}.
\end{align}
where
\begin{equation}
X_m(\phi)=
\begin{cases}
\cos m\phi,  & \text{if $m\geq 0$}, \\
\sin|m|\phi, & \text{if $m<0$.}
\end{cases}\\
\end{equation}
and 
$P_{l}^m(x)$ $(m\geq 0)$ are the normalized associated Legendre polynomials.

Substituting \eqref{eq:expension} into the sum of the energy functional in \eqref{6thNDLdG} and the surface energy in \eqref{eq:surface}, we obtain a free energy as a function of these unknown coefficients $A_{nlm}^{(i)}$.
We minimise the energy function by using a standard optimisation method, L-BFGS \cite{nocedal1999numerical} with a random initial condition $A_{nlm}^{(i)} = 0.2(2rand()/RAND_{MAX}-1)$, where $rand()$ returns a pseudo-random number in the range of $[0, RAND_{MAX})$ in C++, $i = 1,\cdots,5$.

\end{document}